\documentclass[12pt]{article}
\usepackage{amsmath, amsthm}
\usepackage{graphicx,psfrag,epsf}
\usepackage{enumerate, mbenotes}
\usepackage{natbib, algorithmic}
\usepackage{url} 
\usepackage{times, bm, threeparttable}
\usepackage{wrapfig}
\usepackage[linesnumbered,vlined,ruled]{algorithm2e}

\usepackage{amsthm,amsmath,amssymb,bbm}
\usepackage{enumitem}

\usepackage{caption,multirow}
\usepackage{subcaption,titling}
\usepackage{graphicx,color}

\usepackage{mathtools}
\usepackage{float}
\usepackage{titletoc}

\newtheorem{thm}{Theorem}
\newtheorem{lemma}{Lemma}
\newtheorem{remark}{Remark}
\newtheorem{corollary}{Corollary}
\newtheorem{cond}{Condition}

\usepackage{xr}
\usepackage{tikz}
\usetikzlibrary{arrows}
\tikzset{
	vertex/.style={circle,draw,minimum size=1.5em},
	edge/.style={->,> = latex'}
}

\def\bz{\mathbf{z}}
\def\bx{\mathbf{x}}
\def\tbz{\tilde{\mathbf{z}}}
\def\bc{\mathbf{c}}
\def\be{\mathbf{e}}
\def\pg{\pmb\gamma}

\def\pmu{\pmb\mu}
\def\ezhg{e^{{\mathbf{z}}_{ij}^\top \hat{\pmb\gamma}}}
\def\ezg{e^{{\mathbf{z}}^\top_{ij} \pmb\gamma^0}}
\def\etzg{e^{\tilde{\mathbf{z}}_{ij}^\top \pmb\gamma^0}}
\def\mcN{\mathcal{N}}
\def\mcR{\mathcal{R}}
\def\mcC{\mathcal{C}}
\def\Pr{\text{Pr}}
\def\Bm{\| \bar B \|_{\max} }
\def\nezhg{e^{-{\mathbf{z}}_{ij}^\top \hat{\pmb\gamma}}}

\newcommand{\blind}{1}

\begin{document}

\def\spacingset#1{\renewcommand{\baselinestretch}%
{#1}\small\normalsize} \spacingset{1}


\if1\blind
{
  \title{\bf Pairwise Covariates-Adjusted Block Model for Community Detection\thanks{Huang and Sun contribute equally to this work. Corresponding Author: Yang Feng (yang.feng@nyu.edu)}}
  \author{Sihan Huang, Jiajin Sun
  \hspace{.2cm}\\
    Department of Statistics, Columbia University\\
    Yang Feng\\
    Department of Biostatistics, New York University
   }
   \date{}
  \maketitle
} \fi

\if0\blind
{
  \bigskip
  \bigskip
  \bigskip
  \begin{center}
    {\LARGE\bf Pairwise Covariates-Adjusted Block Model for Community Detection}
\end{center}
  \medskip
} \fi

\bigskip
\begin{abstract}
One of the most fundamental problems in network study is community detection. The stochastic block model (SBM) is a widely used model, and various estimation methods have been developed with their community detection consistency results unveiled. However, the SBM is restricted by the strong assumption that all nodes in the same community are stochastically equivalent, which may not be suitable for practical applications. We introduce a pairwise covariates-adjusted stochastic block model (PCABM), a generalization of SBM that incorporates pairwise covariate information. We study the maximum likelihood estimates of the coefficients for the covariates as well as the community assignments. It is shown that both the coefficient estimates of the covariates and the community assignments are consistent under suitable sparsity conditions. Spectral clustering with adjustment (SCWA) is introduced to efficiently solve PCABM. Under certain conditions, we derive the error bound of community detection under SCWA and show that it is community detection consistent. In addition, we investigate model selection in terms of the number of communities and feature selection for the pairwise covariates, and propose two corresponding algorithms. PCABM compares favorably with the SBM or degree-corrected stochastic block model (DCBM) under a wide range of simulated and real networks when covariate information is accessible.
\end{abstract}

\noindent%
\vfill

\spacingset{1.4} 

\section{Introduction}\label{sec:intro}
Networks are used to represent connections among subjects within a population of interest, and their wide range of applications has drawn researchers from various fields. In social media, network analysis can reveal people's behaviors and interests through their connections, such as Facebook friends and Twitter followers. In ecology, a food web depicting predator-prey interactions offers valuable insights into individual habits and the structure of biocoenosis. Network analysis also has extensive applications in computer science, biology, physics, and economics \citep{getoor2005link, goldenberg2010survey, newman1963introduction, graham2014econometric}.

Community detection, one of the most studied problems for network data, is concerned with identifying groups of nodes that are densely connected within groups and sparsely connected between groups. Detecting network communities not only aids in understanding the structural features of networks, but also has practical applications. For instance, communities in social networks often share similar interests, which can help the development of recommendation systems. Community detection methods primarily fall into two categories: algorithm-based and model-based. Algorithm-based methods \citep{bickel2009nonparametric,newman2006modularity,zhao2011community,wilson2014testing,wilson2017community} involve devising an objective function (e.g., modularity) and optimizing it for community detection, while model-based methods assume that edges are generated from a probabilistic model. Popular models include the stochastic block model \citep{holland1983stochastic}, mixture model \citep{newman2007mixture}, degree-corrected stochastic block model \citep{karrer2011stochastic}, and latent space models \citep{hoff2002latent, handcock2007model, hoff2008modeling}. For a comprehensive review of statistical network models, refer to \cite{goldenberg2010survey} and \cite{fortunato2010community}.

The classical stochastic block model (SBM) posits that the connection between each pair of nodes depends solely on their community labels. For SBM, community detection consistency has been established for various methods, such as modularity maximization \citep{newman2006modularity}, profile likelihood \citep{bickel2009nonparametric,choi2012stochastic}, spectral clustering \citep{rohe2011spectral,lei2015consistency}, variational inference \citep{bickel2013}, and penalized local maximum likelihood estimation \citep{gao2017achieving}, among others. However, in real-world scenarios, node connections may depend not only on community structure but also on nodal or pairwise covariates. For example, in an ecological network, predator-prey links between species could be influenced by factors such as prey types, habits, body sizes, and living environments. By incorporating nodal and pairwise information into network models, a more accurate community structure can be obtained.

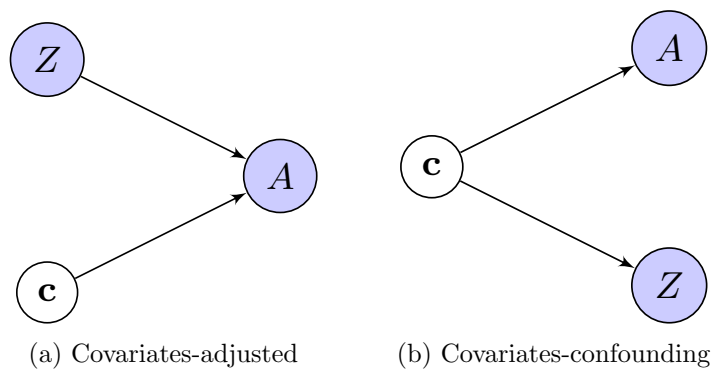
\begin{figure}[!h]
\centering
    \begin{subfigure}[b]{0.35\textwidth}
        \centering
        \resizebox{\linewidth}{!}{
            \begin{tikzpicture}
            \node[vertex,fill=blue!20] (1) at (0,0) {$Z$};
            \node[vertex,fill=blue!20] (2) at (2,-1) {$A$};
            \node[vertex] (3) at (0,-2) {$\mathbf{c}$};
            \draw[edge] (1) -- (2);
            \draw[edge] (3) -- (2);
            \end{tikzpicture}
        }
        \caption{Covariates-adjusted}
        \label{fig:subfig1}
    \end{subfigure}
    \hspace*{2em}%
    \begin{subfigure}[b]{0.35\textwidth}
        \centering
        \resizebox{\linewidth}{!}{
            \begin{tikzpicture}
            \node[vertex] (1) at (0,-1) {$\mathbf{c}$};
            \node[vertex,fill=blue!20] (2) at (2,0) {$A$};
            \node[vertex,fill=blue!20] (3) at (2,-2) {$Z$};
            \draw[edge] (1) -- (2);
            \draw[edge] (1) -- (3);
            \end{tikzpicture}
        }
        \caption{Covariates-confounding}   
        \label{fig:subfig2}
    \end{subfigure}
    \caption{Two different network models including covariates} 
    \label{fig:model}
\end{figure}

Depending on the relationship between communities and covariates, there are generally two classes of models, as depicted in Figure~\ref{fig:model}: \emph{covariates-adjusted} and \emph{covariates-confounding}. The symbols $\mathbf{c}$, $Z$, and $A$ represent latent community labels, pairwise covariates, and the adjacency matrix, respectively. In Figure~\ref{fig:subfig1}, the latent community and covariates jointly determine the network structure. One example of this model is the friendship network among students. Students may become friends for various reasons, such as being in the same class, sharing hobbies, or belonging to the same ethnic group. Without adjusting for these covariates, it is difficult to infer a single community membership from $A$. We will analyze one such example in detail in Section~\ref{sec:realdata}. Conversely, covariates may carry the same community information as the adjacency matrix, as shown in Figure~\ref{fig:subfig2}. The term "confounding" originates from graph models \citep{greenland1999confounding}. The citation network serves as an excellent example of this model \citep{tan2016topic}. When research topics are treated as community labels for articles, citation links largely depend on the research topics of the article pair. Simultaneously, the distribution of keywords is likely driven by the specific topic an article addresses.

Researchers often modify the SBM in the above two ways to incorporate covariate information. For the covariates-confounding model, \cite{newman2016structure} use covariates to construct the prior for community labels and then generate edges using a degree-corrected model.  \cite{zhang2016community} proposes a joint community detection criterion, an analog of modularity, to incorporate nodal features. \cite{deshpande2018contextual} establishes information-theoretic bounds for combining a block model and a spike covariance model that are conditionally independent given class assignments. \cite{yan2021covariate} suggests a semidefinite programming framework to aggregate network and covariate information, while \cite{xu2022covariate} considers an augmented adjacency tensor approach under an analogous setting in multilayer SBM. \cite{weng2016community} employs a logistic model as the prior for community labels. For the covariates-adjusted model, \cite{yan2019statistical} proposes a directed network model with a logistic function, but it does not consider potential community structures. \cite{wu2017generalized} introduces a generalized linear model with low-rank effects to model network edges, which could imply a community structure or a latent space structure, although not explicitly mentioned; \cite{ma2020universal} presents algorithms for a latent space model that incorporates edge covariates; both of these works consider penalized MLE with convex relaxation and gradient-based algorithms.

In this work, we propose a simple yet effective model called \emph{Pairwise Covariates-Adjusted Stochastic Block Model} (PCABM), which extends the SBM by adjusting the probability of connections according to the contribution of pairwise covariates\footnote{Note that these are ``edge-level" covariates instead of the nodal or vertex-level covariates that are often considered in other parts of the literature. Having said that, one can incorporate nodal information into our model by converting it into pairwise covariates, where an example will be presented in Section \ref{subsec::polblog}.}. Through this model, we can learn how each covariate affects the connections by examining its corresponding regression coefficient, for which asymptotic normality is established. In addition, we investigate the likelihood-based community detection method and propose an efficient pseudo-likelihood expectation-maximization (PLEM) algorithm. Consistency results for both the MLE and the PLEM algorithm are provided. Apart from likelihood methods, we also propose a novel spectral clustering method for PCABM. We prove desirable theoretical properties for the spectral clustering method, and demonstrate that, as a fast algorithm, using it as an initial estimator for the likelihood method results in more accurate community detection than random initialization. Furthermore, we consider the model selection problems of estimating the number of communities and selecting the important confounding covariates, providing algorithms to address these two issues based on the edge cross-validation framework proposed by \cite{li2020network}.

The remainder of the paper is organized as follows. In Section~\ref{sec:pcabm}, we introduce the PCABM. We then present the asymptotic properties of the coefficient estimates in Section~\ref{sec:gamma}. After that, we introduce two methods for community detection: a likelihood approach in Section~\ref{sec:likethm} and a spectral approach in Section~\ref{sec:scwa}. In addition, we present two algorithms for model selection in Section~\ref{sec:model_select}. Simulations and applications on real networks are discussed in Sections~\ref{sec:simu} and~\ref{sec:realdata}, respectively. We conclude the paper with a brief discussion in Section~\ref{sec:discussion}. All proofs are relegated to the Supplementary Materials.

Here, we introduce some notations to facilitate the discussion.  For a square matrix $M\in \mathbb{R}^{n\times n}$,  let $\|M\|$ be the operator norm of $M$, $\|M\|_F=\sqrt{\mbox{trace}(M^TM)}$, $\|M\|_{\infty}=\max_i\sum_{j=1}^n|M_{ij}|$, $\|M\|_0=\#\{(i,j)|M_{ij}\neq0\}$, and $\|M\|_{\max}=\max_{ij}|M_{ij}|$. $\lambda_{\min}(M)$ is the minimum eigenvalue of $M$.  For index sets $I,J\subset[n]:=\{1,2,\cdots,n\}$, $M_{I\cdot}$ and $M_{\cdot J}$ are the sub-matrices of $M$ consisting the corresponding rows and columns, respectively. For a vector $\mathbf{x}\in \mathbb{R}^n$, let $\|\mathbf{x}\|=\sqrt{\sum_{i=1}^n x_i^2}$ and $\|\mathbf{x}\|_{\infty}=\max_i |x_i|$. We define the Kronecker power by $\mathbf{x}^{\otimes(k+1)}=\mathbf{x}^{\otimes k}\otimes\mathbf{x}$, where $\otimes$ is the Kronecker product.

For any positive integer $K$, we define $I_K\in\mathbb{R}^{K\times K}$ to be the identity matrix and $\mathbf{1}_K$ to be the all-one vector. When there is no confusion, we will sometimes omit the subscript $K$. For a vector $\mathbf{x}\in\mathbb{R}^K$,  $D(\mathbf{x})\in\mathbb{R}^{K\times K}$ represents the diagonal matrix whose diagonal elements take the value of $\mathbf{x}$. For an event $A$, its indicator function is written as $\mathbbm{1}(A)$. 
For two real number sequences $x_n$ and $y_n$, we say $x_n=o(y_n)$ or $y_n = \omega(x_n) $ if $\lim_{n\to\infty}x_n/y_n=0$, $x_n=O(y_n)$ or $y_n = \Omega(x_n) $ if $\limsup_{n\to\infty}|x_n/y_n|\leq\infty$.

\section{Pairwise Covariates-Adjusted Stochastic Block Model}\label{sec:pcabm}

We consider a graph with $n$ nodes and $K$ communities, where $K$ could be fixed or increase with $n$. In this paper, we focus on undirected weighted graphs without self-loops. All edge information is incorporated into a symmetric adjacency matrix $A=[A_{ij}]\in\mathbb{N}^{n\times n}$ with diagonal elements being zero, where $\mathbb{N}$ represents the set of nonnegative integers. The total number of possible edges is denoted by $N_n=n(n-1)/2$. The true node labels $\mathbf{c}=\{c_1,\cdots,c_n\}\in \{1,\cdots,K\}^n$ are drawn independently from a multinomial distribution with parameter vector $\boldsymbol\pi=(\pi_1,\cdots,\pi_K)^T$, where $\sum_{k=1}^K\pi_k=1$ and $\pi_k>0$ for all $k$. The community detection problem aims to find a disjoint partition of the nodes, or equivalently, estimated node labels $\mathbf{e}=\{e_1,\cdots,e_n\}\in \{1,\cdots,K\}^n$ that is close to $\mathbf{c}$, where $e_i\in\{1,\cdots,K\}$ is the label for node $i$.

In the classical SBM, we assume $\text{Pr}(A_{ij}=1|\boldsymbol{c}) =B_{c_ic_j}$, where $B=[B_{ab}]\in [0,1]^{K\times K}$ is a symmetric matrix with no identical rows. In practice, the connection between two nodes may depend not only on the communities they belong to, but also on the nodal information (e.g., gender, age, religion). To fix the idea, assume in addition to $A$, we have observed a pairwise $p$-dimensional vector $\mathbf{z}_{ij}$ between nodes $i$ and $j$. Denote the collection of the pairwise covariates among nodes as $Z=[\mathbf{z}^T_{ij}]\in\mathbb{R}^{n^2\times p}$. Here, we assume $\mathbf{z}_{ij} = \mathbf{z}_{ji}$ and $\mathbf{z}_{ii}=\mathbf{0}$, for all $i$ and $j$.

Now, we are ready to introduce the \emph{Pairwise Covariates-Adjusted Stochastic Block Model} (PCABM).  For $i<j$, conditional on the community label $\mathbf{c}$ and the pairwise covariate matrix $Z$, $A_{ij}$'s are independent and 
$$A_{ij}\sim{\rm Poisson}(\lambda_{ij}),\ \lambda_{ij}=B_{c_ic_j}e^{\mathbf{z}_{ij}^T\boldsymbol{\gamma}^0},$$
where $\boldsymbol{\gamma}^0$ is the true coefficient vector for the pairwise covariates. In addition to the goal of recovering the community membership vector $\mathbf{c}$, we would also like to get an accurate estimate for $\boldsymbol{\gamma}^0$.

The specific term $\exp(\mathbf{z}_{ij}^T\boldsymbol{\gamma}^0)$ is introduced here to adjust the connectivity between nodes $i$ and $j$. Here, as in the vanilla SBM, we assume a sparse setting for $B=\rho_n\bar{B}$, with $\bar{B}$ fixed and $\rho_n\to0$ as $n\to\infty$. Note that due to the contribution of $Z$, $\varphi_n=n\rho_n$ is no longer the expected degree as in the vanilla SBM \citep{zhao2012consistency}, but it is still useful as a measure of the network sparsity. It is easy to observe that when $\boldsymbol{\gamma}^0=0$, PCABM reduces into the vanilla Poisson SBM.

Under PCABM, the likelihood function is
\begin{align*}
	\mathcal{L}(\mathbf{e},\boldsymbol{\gamma},B,{\boldsymbol{\pi}}|A,Z)\propto\prod_{i=1}^n\pi_{e_i}\prod_{i<j}B_{e_ie_j}^{A_{ij}}e^{A_{ij}\mathbf{z}_{ij}^T\boldsymbol{\gamma}}\exp\left(-B_{e_ie_j}e^{\mathbf{z}_{ij}^T\boldsymbol{\gamma}}\right).
\end{align*}
Define 
\begin{align*}
	n_k(\mathbf{e})&=\sum_{i=1}^n\mathbbm{1}(e_i=k),\ O_{kl}(\mathbf{e})=\sum_{ij}A_{ij}\mathbbm{1}(e_i=k, e_j=l),\\
	E_{kl}(\mathbf{e},\boldsymbol{\gamma})&=\sum_{i\neq j}e^{\mathbf{z}_{ij}^T\boldsymbol{\gamma}}\mathbbm{1}(e_i=k, e_j=l)=\sum_{(i,j)\in s_{\mathbf{e}}(k,l)}e^{\mathbf{z}_{ij}^T\boldsymbol{\gamma}},
\end{align*}
where $s_{\mathbf{e}}(k,l)=\{(i,j)|e_i=k,e_j=l,i\neq j\}$. Under the assignment $\mathbf{e}$, $n_k(\mathbf{e})$ represents the number of nodes estimated to be in the community $k$. For $k\neq l$, $O_{kl}$ is the total number of edges between estimated communities $k$ and $l$; for $k=l$, $O_{kk}$ is twice the number of edges within estimated community $k$. $E_{kl}$ is the summation of all pair-level factors between estimated communities $k$ and $l$. Up to a constant term, we can write the log-likelihood function as 
\begin{align*}
\log\mathcal{L}(\mathbf{e},\boldsymbol{\gamma},B,\boldsymbol{\pi}|A,Z) = &\sum_kn_k(\mathbf{e})\log\pi_k+\frac{1}{2}\sum_{kl}O_{kl}(\mathbf{e})\log B_{kl}\\
&-\frac{1}{2}\sum_{kl}B_{kl}E_{kl}(\mathbf{e},\boldsymbol{\gamma})+\sum_{i<j}A_{ij}\mathbf{z}_{ij}^T\boldsymbol{\gamma}.
\end{align*}
Given $\mathbf{e}$ and $\boldsymbol{\gamma}$, we derive the MLE $\hat{\pi}_k(\mathbf{e})=\frac{n_k(\mathbf{e})}{n}$ and $\hat{B}_{kl}(\mathbf{e},\boldsymbol{\gamma})=\frac{O_{kl}(\mathbf{e})}{E_{kl}(\mathbf{e},\boldsymbol{\gamma})}$. Plugging $\hat{B}(\mathbf{e},\boldsymbol{\gamma})$ and $\hat{\boldsymbol{\pi}}(\mathbf{e})$ into the original log-likelihood and discarding the constant terms, we have
\begin{align}\label{eq:log-lik}
\begin{split}
&\log\mathcal{L}(\mathbf{e},\boldsymbol{\gamma},\hat{B},\hat{{\boldsymbol{\pi}}}|A,Z)\\
\propto&\frac{1}{2}\sum_{kl}O_{kl}(\mathbf{e})\log\frac{O_{kl}(\mathbf{e})}{E_{kl}(\mathbf{e},\boldsymbol{\gamma})}+\sum_{i<j}A_{ij}\mathbf{z}_{ij}^T\boldsymbol{\gamma}+\sum_kn_k(\mathbf{e})\log\frac{n_k(\mathbf{e})}{n}.
\end{split}
\end{align}
Out target is to maximize (\ref{eq:log-lik}) with respect to $\mathbf{e}$ and $\boldsymbol{\gamma}$. We consider a two-step sequential estimation procedure by first studying the estimation of  $\boldsymbol{\gamma}^0$ in Section \ref{sec:gamma} and then the estimation of $\mathbf{c}$ in Section \ref{sec:likethm} (likelihood method) and Section \ref{sec:scwa} (spectral method).

It is worth mentioning that the proposed model includes  DCBM in the following sense: by choosing $p=1$, $\bz_{ij} =  \log (d_i d_j)$ and $\pg = 1$ where $d_i$ is the degree of node $i$, \eqref{eq:log-lik} becomes
\begin{align*}
\log \mathcal L \left(\be, \pg=1, \hat B, \hat{\pmb\pi} | A, Z = (\log(d_id_j))_{n^2\times 1}\right) \propto \frac12 \sum_{kl} O_{kl}(\be) \log \frac{O_{kl}(\be)}{n_k(\be)n_{l}(\be)},
\end{align*}
which is exactly the profile log-likelihood under DCBM  derived by maximizing over ``$\theta$ and $P$'' (degree parameter and block connection probability) in DCBM. From this perspective, one can view PCABM as a generalization of DCBM.

\section{Estimation of Coefficients for Pairwise Covariates\label{sec:gamma}}
As the first step to maximize the log-likelihood, we consider the estimation of coefficients $\boldsymbol{\gamma}^0$ for pairwise covariates. To this end, we impose the following conditions on $Z$.
\begin{cond}\label{cond:zbd}
$\{\mathbf{z}_{ij},i<j\}$ are i.i.d. and uniformly bounded, i.e., for $\forall i<j$, $\|\mathbf{z}_{ij}\|_{\infty}\leq\zeta$, where $\zeta>0$ is some constant. $\|\pg^0\|_1 $ is also bounded by a constant. Denote $\xi = \exp(\zeta \|\pg^0\|_1 )$.
\end{cond}

\begin{remark}
The bounded support condition for $\mathbf{z}_{ij}$ is introduced to simplify the proof. 
It could be relaxed to $\bz_{ij}$ to have a light tail or to allow the upper bound to grow slowly with network size $n$.
For example, our proofs could still go through if $\exp(\bz_{ij}^\top\pg^0) $ follows a sub-Gaussian distribution (with $\|\pg^0\|_1$ bounded), 
under slightly stronger conditions on the sparsity of the network.
\end{remark}



Under Condition~\ref{cond:zbd}, 
the following expectations exist: $\theta(\boldsymbol{\gamma}^0)\equiv\mathbb{E}e^{\mathbf{z}_{ij}^T\boldsymbol{\gamma}^0}\in\mathbb{R}^+$,  $\boldsymbol{\mu}(\boldsymbol{\gamma}^0)\equiv\mathbb{E}\mathbf{z}_{ij}e^{\mathbf{z}_{ij}^T\boldsymbol{\gamma}^0}\in\mathbb{R}^p$, and $\Sigma(\boldsymbol{\gamma}^0)\equiv\mathbb{E}\mathbf{z}_{ij}\mathbf{z}_{ij}^Te^{\mathbf{z}_{ij}^T\boldsymbol{\gamma}^0}\in\mathbb{R}^{p\times p}$.
To ensure that $\boldsymbol{\gamma}^0$ is the unique solution to maximize the likelihood in the population version, we impose the following regularity condition at the true $\boldsymbol{\gamma}^0$.
\begin{cond}\label{cond:zpd}
    $\Sigma(\boldsymbol{\gamma}^0)-\theta(\boldsymbol{\gamma}^0)^{-1}\boldsymbol{\mu}(\boldsymbol{\gamma}^0)^{\otimes 2}$ is positive definite.
\end{cond}

\begin{remark} 
To understand the implication of Condition \ref{cond:zpd}, consider the function $g(\boldsymbol{\gamma})=\theta(\boldsymbol{\gamma})\Sigma(\boldsymbol{\gamma})-\boldsymbol{\mu}(\boldsymbol{\gamma})^{\otimes2}$. In the special case of SBM where $\boldsymbol{\gamma}^0=\mathbf{0}$, we have  $g(\mathbf{0})=\mathbb{E}[\mathbf{z}^{\otimes2}]-\mathbb{E}[\mathbf{z}]^{\otimes2}={\rm cov}(\mathbf{z})$. To avoid multicollinearity, it's natural for us to require ${\rm cov}(\mathbf{z})$ to be positive definite. For a general PCABM, we require $g(\boldsymbol{\gamma})$ to be positive definite at the true value $\boldsymbol{\gamma^0}$.	
\end{remark}

For a given initial community assignment $\boldsymbol{e}_0$, denote by $\ell_{\mathbf{e}_0}$  the log-likelihood terms in (\ref{eq:log-lik}) containing $\boldsymbol{\gamma}$, which is
    $$\ell_{\mathbf{e}_0}(\boldsymbol{\gamma})\equiv\sum_{i<j}A_{ij}\mathbf{z}_{ij}^T\boldsymbol{\gamma}-\frac{1}{2}\sum_{kl}O_{kl}(\mathbf{e}_0)\log E_{kl}(\mathbf{e}_0,\boldsymbol{\gamma}).$$ 
    We consider the following estimate:
\begin{align}\label{eq::gamma_hat}
    \hat{\boldsymbol{\gamma}}(\boldsymbol{e}_0) = \arg\max_{\boldsymbol{\gamma}} \ell_{\textbf{e}_0}(\boldsymbol{\gamma}).
\end{align}
We point out that $\ell_{\be_0}(\pg)$ is concave in $\pg$, so the global optimizer in \eqref{eq::gamma_hat} can be efficiently solved by a BFGS algorithm.
When there is no ambiguity, we will just write it as $\hat{\boldsymbol{\gamma}}$, as we will see in the theory that under some mild conditions, the asymptotic result does not depend on the choice of $\mathbf{e}_0$.
In fact, one could simply choose $\be_0  = \mathbf 1$, the all-one vector, when estimating $\pg^0$.

 To accommodate the ``$K$ growing with $n$" case, we also need the following stability condition.

\begin{cond}\label{cond:e1} 
$\bar B_{\lim} = \lim_{n\to\infty} \sum_{a,b}^K \pi_a \pi_b \bar B_{ab}$ exists.
\end{cond}

\begin{remark}
Note that when $K$ is fixed, Condition~\ref{cond:e1} is automatically satisfied.  
When $K$ grows with $n$, we need the $\pmb\pi$-weighted average of matrix $\bar B$ to have a limit. This is a mild condition since, otherwise, the sequence of observed graphs indexed by $n$ does not come from a consistent data generating process.
\end{remark}
Now we are ready to present the consistency and asymptotic normality of $\hat{\boldsymbol{\gamma}}$.
\begin{thm}[Consistency and asymptotic normality of MLE of $\pg$]
    \label{THM:ASY}
    Under PCABM, assume Conditions 
    \ref{cond:zbd},
      \ref{cond:zpd}
      and \ref{cond:e1}
       hold, where the number of communities $K$ could either be fixed or grow to $\infty$ at an arbitrary rate.
       Then fixing $\be_0 = \mathbf 1$, as $n\to\infty$, if $N_n\rho_n\to\infty$ and $\rho_n\to 0$, 
       we have $\hat \pg(\be_0) \stackrel{p}{\to} \pg^0$ and
    \begin{equation}
    \label{gm_thm1_eq}
        \sqrt{N_n\rho_n}\left[\hat{\boldsymbol{\gamma}}(\mathbf{e}_0)-\boldsymbol{\gamma}^0\right] \stackrel{d}{\to} \mathcal{N}(\mathbf{0}, \Sigma_\infty^{-1}(\boldsymbol{\gamma}^0)),
    \end{equation}
     where $\Sigma_\infty(\boldsymbol{\gamma}^0)=\bar{B}_{\lim}
     [\Sigma(\boldsymbol{\gamma}^0)-\theta(\boldsymbol{\gamma}^0)^{-1}\boldsymbol{\mu}(\boldsymbol{\gamma}^0)^{\otimes 2}]$.
\end{thm}
Different from \cite{yan2019statistical}, in which the network is dense, the convergence rate is $\sqrt{N_n\rho_n}$ rather than $\sqrt{N_n}$ since the effective number of edges is reduced from $N_n$ to $N_n\rho_n$. The asymptotic covariance matrix $\Sigma_\infty^{-1}(\boldsymbol{\gamma}^0)$ depends on $\theta(\boldsymbol{\gamma}^0)$, $\boldsymbol{\mu}(\boldsymbol{\gamma}^0)$, and  $\Sigma(\boldsymbol{\gamma}^0)$, which can be estimated empirically by the plug-in method.
 
 Now, with a consistent estimate of $\boldsymbol{\gamma}^0$, we are ready to study the estimation of $\boldsymbol{c}$. In the next two sections, we will present two different methods for estimating $\boldsymbol{c}$, namely the likelihood-based estimate in Section \ref{sec:likethm} and the spectral method in Section \ref{sec:scwa}.

\section{Likelihood Based Estimate for Community Labels} \label{sec:likethm}
This section presents a likelihood-based estimate for community labels by maximizing $\log\mathcal{L}$ regarding $\mathbf{e}_0$ with $\hat{\boldsymbol{\gamma}}$ from Section \ref{sec:gamma}.
We only present the fixed $K$ setting here, and the results for the growing $K$ scenario are relegated to the Supplementary Materials, partly because in the class label MLE for growing $K$, we consider a slightly different regime from $B = \rho_n \bar B $: we need the signal-noise-ratio, or approximately in-class probability over between-class probability, to also grow with $n$ and $K$; and the conditions are imposed on $\sum_{i<j} B_{c_ic_j} $ rather than $\rho_n$. See section \ref{append_sect:KtoinftyMLE} in the Supplementary Materials for details.

We will show that as long as  
 $\hat{\boldsymbol{\gamma}}(\mathbf{e}_0)$ is consistent, the consistency of $\hat{\mathbf{c}}(\hat{\boldsymbol{\gamma}})$ is guaranteed. 
Plugging $\hat{\boldsymbol{\gamma}}$ into \eqref{eq:log-lik},  the log-likelihood function can be rewritten as
$$\ell_{\hat{\boldsymbol{\gamma}}}(\mathbf{e})=\frac{1}{2}\sum_{kl}O_{kl}(\mathbf{e})\log\frac{O_{kl}(\mathbf{e})}{E_{kl}(\mathbf{e},\hat{\boldsymbol{\gamma}})}+\sum_kn_k(\mathbf{e})\log\frac{n_k(\mathbf{e})}{n}. $$
Then, our maximum likelihood estimate for the community label is
\begin{align}\label{eq::mle-e}
    \hat{\mathbf{c}} = \hat{\mathbf{c}} (\hat{\boldsymbol{\gamma}}):=
    \arg\max_{\mathbf{e}}\ell_{\hat{\boldsymbol{\gamma}}}(\mathbf{e}). 
\end{align}
Note that here we omit $\mathbf{e}_0$ to avoid confusion.
%
Following \cite{zhao2012consistency}, we consider two versions of community detection consistency. Note that the consistency in community detection is understood under any permutation of the labels. To be more precise, let $\mathcal{P}_k$ be the collection of all permutation functions of $[K]$. (1) We say the label estimate $\hat{\mathbf{c}}$ is \emph{weakly consistent} if 
$\text{Pr}[n^{-1}\min_{\sigma\in \mathcal{P}_K}  \sum_{i=1}^n \mathbbm{1}(\sigma(\hat c_i)\neq c_i)<\varepsilon]\to1$ for any $\varepsilon>0$ as $n\to\infty$. (2) We say $\hat{\mathbf{c}}$ is \emph{strongly consistent}  if $\text{Pr}[\min_{\sigma\in \mathcal{P}_K}  \sum_{i=1}^n \mathbbm{1}(\sigma(\hat c_i)\neq c_i)=0]\to1$, as $n\to\infty$. 
We establish both versions of consistency for MLE $\hat\bc$ in the following theorem.

\begin{thm}\label{THM:MLE_fixK}
    Under PCABM that satisfies the Conditions~\ref{cond:zbd} and \ref{cond:zpd}, when $K$ is fixed, the community label estimate $\hat{\boldsymbol{c}}$ defined in (\ref{eq::mle-e}) is weakly consistent if $\varphi_n\to\infty$ and strongly consistent if $\varphi_n/\log n\to\infty$, where $\varphi_n = n\rho_n$.
    \end{thm}

In addition to the fixed $K$ case considered in Theorem \ref{THM:MLE_fixK}, we have also shown the consistency of maximum likelihood label estimate in the case when $K$ grows as fast as $K=O(\sqrt{n})$, where we require a slightly stronger condition on the sparsity, $\varphi_n/(\log n)^{3+\delta} \to \infty$. Details are presented in the Supplementary Materials, section \ref{append_sect:KtoinftyMLE}.

Finding the MLE involves optimizing over all possible label assignments, which is, in principle, NP-hard. 
General discrete optimization methods such as the tabu search \citep{beasley1998heuristic,zhao2012consistency} could be time-consuming and unstable. 
Taking advantage of the specific structure in our problem, we propose a pseudo-likelihood EM algorithm (PLEM) that computes an approximate solution to \eqref{eq::mle-e} efficiently. 

\begin{algorithm}[H]
    \caption{PCABM.PLEM0}\label{alg:pcabm}
    \SetKwInOut{Input}{Input}
        \SetKwInOut{Output}{Output}
  \Input{Adjacency matrix $A$; pairwise covariates $Z$; initial community assignment $\mathbf{e}$; number of communities $K$; iteration number $T$.}
  \Output{Coefficient estimate $\hat{\boldsymbol{\gamma}}$ and community label estimate $\hat{\mathbf{c}}$. }
        Maximize $\ell(\boldsymbol{\gamma})$ in \eqref{eq::gamma_hat} by some optimization algorithm (e.g., BFGS) to get $\hat{\boldsymbol{\gamma}}$.\\
        Initialization: $\hat\pi_l = \frac{n_l(\be)}n, \hat E_{lk}(\be)= \sum_{(i,j)\in s_{\be}(l,k)} \ezhg , \hat B_{lk} = \frac{O_{lk}(\be)}{\hat E_{lk}(\be)}$. 

         \For{$t=1$ \KwTo $T$}
  {
  Calculate $b_{ik}(\be) = \sum_{j=1}^n A_{ij}  \mathbbm{1}{(e_j = k)}$ and $\hat \Xi_{ik}(\be) = \sum_{j=1}^n e^{ \bz_{ij}^\top \hat\pg }  \mathbbm{1}{(e_j = k)}$;\\
  \While{\text{pseudo-likelihood has not converged}}
  {
  \vspace{0.25cm}
  E-step: $\hat \pi_{il} = \frac{ \hat\pi_l \prod_{k=1}^K \exp(b_{ik} \log\hat B_{lk} -\hat \Xi_{ik} \hat B_{lk}) }{ \sum_{m=1}^K \hat\pi_m \prod_{k=1}^K \exp(b_{ik} \log\hat B_{mk} -\hat \Xi_{ik} \hat B_{mk}) }$;\\
  \vspace{0.3cm}
  M-step: $\hat\pi_l = \frac1n \sum_{i=1}^n \hat\pi_{il}, \hat B_{lk} = \frac{\sum_{i=1}^n \hat\pi_{il}b_{ik} }{\sum_{i=1}^n \hat\pi_{il}\hat \Xi_{ik} } $;
  }
  Update label estimates: $e_i = \arg\max_{l} \hat\pi_{il} $.
  }
  Output $\hat\bc$ with $\hat c_i = e_i $.
\end{algorithm}

The algorithm is outlined in Algorithm \ref{alg:pcabm}. In the outer loop, we update the label estimate and related quantities in each iteration. The inner loop employs a latent class EM algorithm to derive a new label estimate based on an initial one. For each edge $A_{ij}$, the pseudo-likelihood treats node $i$ as belonging to the true community $c_i$ and node $j$ as belonging to an estimated community $e_j$. With this approximation, the latent class variables $c_i$'s are separated in the pseudo log-likelihood function, enabling an analytic expression for the EM updates. A similar idea was proposed in \cite{amini2013pseudo} for SBM and DCBM. In the case of PCABM, we adjust the algorithm to account for the covariates. A detailed derivation of the PLEM algorithm under PCABM, as well as its theoretical guarantees, are provided in Section \ref{sec::supple_plem}.

\section{Spectral Clustering with Adjustment}~\label{sec:scwa}
Though the likelihood-based method has appealing theoretical properties, it can sometimes be slow when the network size is large. In addition, the community detection results can be sensitive to the initial label assignments $\mathbf{e}$. In that concern, we aim to propose a computationally efficient algorithm in the flavor of spectral clustering \citep{rohe2011spectral}, which can also be used as the initial community label assignments for the likelihood-based methods. 

\subsection{A Brief Review on Spectral Clustering}
First, we introduce some notations and briefly review the classical spectral clustering with $K$-means for SBM. Let $\mathbb{M}_{n,K}$ be the space of all $n\times K$ matrices where each row has exactly one 1 and $(K-1)$ 0's. We usually call $M\in\mathbb{M}_{n,K}$ a \emph{membership matrix} with $M_{ic_i}=1$ for node $i$ with community label $c_i$. Note that $M$ contains the same information as $\mathbf{c}$, and is only introduced to facilitate the discussion.

From now on, we use PCABM$(M,B,Z,\boldsymbol{\gamma}^0)$ to represent PCABM generated with parameters in parentheses.  Let $G_k=G_k(M)=\{1\leq i\leq n:c_i=k\}$ and $n_k=|G_k|$ for $k=1,\cdots,K$. Let $n_{\min}=\min_{1\leq k\leq K} n_k$, $n_{\max}=\max_{1\leq k\leq K} n_k$ and $n'_{\max}$ is the second largest community size. 

For convenience, we define matrix $P=[P_{ij}]\in [0,\infty)^{n\times n}$, where $P_{ij}=B_{c_ic_j}$. Then it is easy to observe $P=MBM^T$. When $A$ is generated from a SBM with $(M,B)$, the $K$-dimensional eigen-decomposition of $P=UDU^T$ and $A=\hat{U}\hat{D}\hat{U}^T$ are expected to be close, where $\hat{U}^T\hat{U}=I_K$ and $D,\hat{D}\in\mathbb{R}^{K\times K}$. Since $U$ has only $K$ unique rows, which represent the community labels, the $K$-means clustering on the rows of $\hat{U}$ usually leads to a good estimate of $M$. 
While finding a global minimizer for the $K$-means problem 
is NP-hard \citep{aloise2009np}, for any positive constant $\epsilon$, we have efficient algorithms to find an $(1+\epsilon)$-approximate solution \citep{kumar2004simple,lu2016statistical}: 
\begin{align*}
(\hat{M},\hat{X})&\in\mathbb{M}_{n,K}\times\mathbb{R}^{K\times K}\quad
s.t. \quad \|\hat{M}\hat{X}-\hat{U}\|_F^2&\leq(1+\epsilon)\min_{M\in\mathbb{M}_{n,K},X\in\mathbb{R}^{K\times K}}\|MX-\hat{U}\|_F^2.
\end{align*}

The goal of community detection is to find $\hat{M}$ that is close to $M$. To define a loss function, we need to take permutation into account. Let  $\mathcal{S}_K$ be the space of all $K\times K$ permutation matrices. Following  \cite{lei2015consistency}, we define two measures of estimation error: the overall error and the worst-case relative error:
\begin{align*}
L_1(\hat{M},M)&=n^{-1}\min_{S\in \mathcal{S}_K}\|\hat{M}S-M\|_0,\quad
L_2(\hat{M},M)&=\min_{S\in \mathcal{S}_K}\max_{1\leq k\leq K}n_k^{-1}\|(\hat{M}S)_{G_k\cdot}-M_{G_k\cdot}\|_0.
\end{align*}
It can be seen that $0\leq L_1(\hat{M},M)\leq L_2(\hat{M},M)\leq2$. While $L_1$ measures the overall proportion of mis-clustered nodes, $L_2$ measures the worst-case performance across all communities.

Vanilla spectral clustering on SBM requires the average degree of the network to be of the order $\Omega(\log n)$ \citep{lei2015consistency}, mainly because sparser networks do not have desired concentration properties like $||A-\mathbb E A|| = O(\sqrt{\varphi_n})$.
In particular, because the true $\mathbb E A$ has elements of the same scale, one can imagine a node with a very large degree will harm the closeness between $A$ and $\mathbb E A$, which is the basis that spectral clustering lies on.
Recent works \citep{le2017concentration,gao2017achieving,joseph2016impact} have shown that regularized versions of spectral clustering  \citep{amini2013pseudo,qin2013regularized},
which basically means performing spectral clustering on a regularized adjacency matrix, 
could enable the concentration of the adjacency matrix under sparser settings and thus relax the average degree assumption required in vanilla spectral clustering.
In our algorithms, we adopt the ``reduce weight of edges proportionally to the excess of degrees'' version of regularization \citep{le2017concentration}, i.e. assigning weight $\sqrt{\lambda_i\lambda_j}$ to $A_{ij}$, where $\lambda_i := \min\{2 d / d_i, 1\}$, $d = \max_{ij} n P_{ij}$, and $d_i$ is the degree of node $i$. 
As $d$ is unknown, in practice we can take $\lambda_i = \min\{\lambda^R\bar d / d_i, 1 \} $, where $\bar d = \sum_{i=1}^n d_i /n $ is the average degree, and $\lambda^R$ is a  constant. For theoretical guarantee we need $\lambda^R$ to be large enough, but in practice $\lambda^R=2$ is sufficient to give satisfactory results from our simulation experience.

\subsection{Regularized Spectral Clustering with Adjustment}

The existence of covariates in PCABM prevents us from
applying (regularized) spectral clustering directly on $A$. Unlike SBM where $A$ is generated from a low-rank matrix $P$, $A$ in PCABM consists of both community and covariate information. Since $P_{ij}=\mathbb{E}[A_{ij}/e^{\mathbf{z}^T_{ij}\boldsymbol{\gamma}^0}]$, an intuitive idea to take advantage of the low-rank structure is to remove the covariate effects, i.e. using the adjusted adjacency matrix $[A_{ij}/e^{\mathbf{z}^T_{ij}{\boldsymbol{\gamma}^0}}]$ for spectral clustering. 

\begin{algorithm}[h]
    \caption{PCABM.SCWA}\label{alg:scwa}
     \SetKwInOut{Input}{Input}
        \SetKwInOut{Output}{Output}
        \Input{Adjacency matrix $A$; pairwise covariates $Z$; initial community assignment $\mathbf{e}$; number of communities $K$; approximation parameter $\epsilon$; constant $\lambda^R$.}
        \Output{Coefficient estimate $\hat{\boldsymbol{\gamma}}$; community estimate $\hat{\mathbf{c}}$.}
        Maximize $\ell(\boldsymbol{\gamma})$ as in \eqref{eq::gamma_hat} by some optimization algorithm (e.g., BFGS) to derive $\hat{\boldsymbol{\gamma}}$.
        \\
        Compute the adjusted adjacency matrix $A'=[A'_{ij}]$ where $A'_{ij} = A_{ij}\exp(-\mathbf{z}_{ij}^T\hat{\boldsymbol{\gamma}})$.
        \\
        Compute the weighted adjusted adjacency matrix $A'^R=[A'^R_{ij}]$, where $A'^R_{ij} = A'_{ij} \sqrt{\lambda_i \lambda_j}$, where $\lambda_i = \min\{\lambda^R d'/ d_i', 1\}$, $d_i'$ is the degree of node $i$ in $A'$ and $d' =\sum_i d_i' /n$.
        \\
        Calculate $\hat{U}\in\mathbb{R}^{n\times K}$ consisting of the leading $K$ eigenvectors (ordered in absolute eigenvalue) of $A'^R$.
        \\
        Calculate the $(1+\epsilon)$-approximate solution $\hat{M}$ to the $K$-means problem with $K$ clusters and input matrix $\hat{U}$.
        \\
        Output $\hat{\mathbf{c}}$ according to $\hat{M}$.
\end{algorithm}

In practice, we don't know the true value of the parameter $\boldsymbol{\gamma}^0$. Naturally, we replace $\boldsymbol{\gamma}^0$ with the empirical estimate $\hat{\boldsymbol{\gamma}}$ from \eqref{eq::gamma_hat}, and define the \emph{adjusted adjacency matrix} as $A'=[A'_{ij}]$ where $A'_{ij}=A_{ij}\exp(-\mathbf{z}_{ij}^T\hat{\boldsymbol{\gamma}})$. 
Furthermore, for regularized spectral clustering, define the weighted version of $A'$ to be $A'^R$, called \emph{weighted adjusted adjacency matrix}.
By the asymptotic properties of $\hat{\boldsymbol{\gamma}}$ proved in Theorem~\ref{THM:ASY}, we show that $\|A'^R-P\|$ achieves the desirable spectral bound of order $O_p(\sqrt{\varphi_n})$; 
the proof is given in Section~\ref{sec:supple:spectralbd} of the Supplementary Materials.

Based on this bound, we could then apply the regularized spectral clustering algorithm on matrix $A'$ to detect the communities. We call this adjustment scheme the \emph{Spectral Clustering with Adjustment} (SCWA) algorithm, which is elaborated in Algorithm \ref{alg:scwa}.

To show the consistency of Algorithm~\ref{alg:scwa}, one natural requirement is that $A'^R$ and $P$ are close enough, which is stated rigorously in the following theorem.

\begin{thm}[Spectral bound of adjusted, regularized Poisson random matrices]
    \label{THM:SC}
    Let $A$ be the adjacency matrix generated by the undirected PCABM $(M, B, Z, \pg^0)$.
Assume Conditions \ref{cond:zbd}, \ref{cond:zpd}, \ref{cond:e1} hold.
Further assume each element of $\bar B$ is bounded from above by a constant $C_{\bar B}$
and below by a constant $c_{\bar B} $.
For any $r>1$, the following holds with probability at least $1 - 5n^{-r} - C_{\eta} \exp(-v_{\eta} n) $ (where $\eta = (p\zeta)^{-1}$, $C_{\eta}$ and $v_{\eta}$ are constants in Lemma \ref{lem:eta}): 
the regularized adjusted adjacency matrix
$A'^{R}$ in Algorithm \ref{alg:scwa}
 satisfies
\begin{equation}
\label{specbd:A'R}
\|A'^{R} - P\| \leq C\sqrt{\varphi_n}
\end{equation}
where $C$ is a constant that depends on $p, r, \xi, \zeta, C_{\bar B}$ and $c_{\bar B}$.
\end{thm}

Similarly to the proof of Theorem 3.1 in \cite{lei2015consistency}, we can prove the following Theorem~\ref{main} by combining Lemmas 5.1 and 5.3 in \cite{lei2015consistency}, and Theorem~\ref{THM:SC}. Without loss of generality, we now assume $\|\bar{B}_{\max}\|\leq1$, which makes the statement of the theorem simpler.

\begin{thm}
    \label{main}
    In addition to the conditions of Theorem~\ref{THM:SC}, assume that $P=MBM^T$ is of rank $K$ with the smallest absolute non-zero eigenvalue at least $\varrho_n$. Let $\hat{M}$ be the output of spectral clustering using $(1+\epsilon)$ approximate $K$-means on $A'^R$ (defined in Algorithm \ref{alg:scwa}, step 3). For any constant $r>0$, there exists an absolute constant $C>0$, such that, if 
    \begin{equation}\label{eq:cond}
    (2+\epsilon)\frac{Kn\rho_n}{\varrho_n^2}<C,
    \end{equation}
    then, with probability at least $1 - 5n^{-r} - C_{\eta} \exp(-v_{\eta} n)$, there exist subsets $H_k\subset G_k$ for $k=1,\cdots,K$, and a $K\times K$ permutation matrix $J$ such that $\hat{M}_{G\cdot}J=M_{G\cdot}$, where $G=\cup_{k=1}^K(G_k\setminus H_k)$, and 
    \begin{equation}\label{eq:error}
    \sum_{k=1}^K\frac{|H_k|}{n_k}\leq C^{-1}(2+\epsilon)\frac{Kn\rho_n}{\varrho_n^2}.
    \end{equation}
\end{thm}

Inequality~(\ref{eq:error}) provides an error bound for 
the overall relative error. 
Theorem~\ref{main} doesn't provide us with an error bound in a straightforward form since $\varrho_n$ contains $\rho_n$. The following corollary gives us a clearer view of the error bound in terms of model parameters. The condition that the maximum normalized probability equals 1 can be replaced by any constant, but we just use 1 here for simplicity, since any constant can always be absorbed into the sparsity parameter $\rho_n$. 

\begin{corollary}\label{cor:scwa}
    In addition to  the conditions of Theorem~\ref{THM:SC}, assume that $\bar{B}'s$ minimum absolute eigenvalue bounded below is by $\tau>0$ and $\max_{kl}\bar{B}(k,l)=1$. Let $\hat{M}$ be the output of spectral clustering using $(1+\epsilon)$ approximate $K$-means on $A'^R$. For any constant $r>0$, there exists an absolute constant $C$ such that if 
    $$(2+\epsilon)\frac{Kn}{n^2_{\min}\tau^2\rho_n}<C,$$
    then with probability at least $1 - 5n^{-r} - C_{\eta} \exp(-v_{\eta} n) $,
    $$
    L_2(\hat M, M)
    \leq C^{-1}(2+\epsilon)\frac{Kn}{n^2_{\min}\tau^2\rho_n}, \quad
    L_1(\hat M, M)
    \leq C^{-1}(2+\epsilon)\frac{Kn'_{\max}}{n^2_{\min}\tau^2\rho_n}.$$
\end{corollary}
It is worth mentioning that Theorem \ref{THM:SC}, Theorem \ref{main}, and Corollary \ref{cor:scwa} all allow $K$ to go to infinity with $n$.

Compared to SCWA, the pseudo-likelihood EM algorithm can yield more accurate results, especially when provided with good initial labels. On the other hand, the SCWA algorithm is computationally more efficient. To combine the advantages of these two methods, we propose using the results of SCWA as the initial estimate for the pseudo-likelihood EM (PCABM.PL as described in Algorithm \ref{alg:scwapcabm}). We will conduct extensive simulation studies in Section~\ref{sec:simu} to evaluate the performance of both PCABM.SCWA and PCABM.PL.

\begin{algorithm}[H]
    \caption{PCABM.PL}\label{alg:scwapcabm}
        \SetKwInOut{Input}{Input}
        \SetKwInOut{Output}{Output}
        \Input{Adjacency matrix $A$; pairwise covariates $Z$; initial community assignment $\mathbf{e}$; number of communities $K$; approximation parameter $\epsilon$; iteration number $T$.}
        \Output{
        Community estimate $\hat{\mathbf{c}}$.}
        Use Algorithm~\ref{alg:scwa} to get 
        an initial community estimate $\hat\be$.
        \\
        Use Algorithm~\ref{alg:pcabm} with initial community estimate $\hat\be$ to get the community estimate $\hat{\mathbf{c}}$.
\end{algorithm}

\section{Model Selection}
\label{sec:model_select}
So far, we have been treating the number of communities $K$ as given. In practice, the true value of $K$ may be unknown to us. In that case, we would be interested in estimating $K$.
To provide a systematic approach, we propose adapting the edge-sampling cross-validation (ECV) method \citep{li2020network} to the PCABM. The main idea of the ECV procedure can be summarized as follows: in each iteration, we randomly sample a certain proportion of node pairs in the network, and predict the remaining node pairs under specific models based on matrix completion on the adjacency matrix containing the true edge information of the selected node pairs. After all iterations, we compare the average prediction performance or hold-out losses under different models and choose the best model accordingly. Algorithm \ref{ecv_k_pcabm} presents a detailed process of applying this idea to estimate $K$ in the PCABM. The notation $P_\Omega A$ represents the matrix that retains all elements of $A$ in the index set $\Omega$ while setting other elements to 0.

In step 5 of Algorithm \ref{ecv_k_pcabm}, $\hat A'_K$ denotes the rank-$K$ matrix completion from $P_\Omega A'$.
As suggested in \cite{li2020network}, we use the SVD truncation approach to obtain $\hat A'_K$.
In the SVD of $P_\Omega A' = UDV^\top$, we keep the $K$ largest elements of diagonal $D$ and set $\hat A'_K = \frac1p UD_K V^\top$. This simple matrix completion method efficiently serves our model selection goal while remaining computationally inexpensive.

For the loss evaluated in step 7 of Algorithm \ref{ecv_k_pcabm}, there are two options: the scaled negative log-likelihood (snll) $\sum_{(i,j)\in \Omega^c} \left[\hat B_{\hat e_i \hat e_j}  - A_{ij}\exp(-\bz_{ij}^\top \hat\pg) \log \hat B_{\hat e_i \hat e_j}\right]$ and the scaled $L_2$ loss $\sum_{(i,j)\in \Omega^c} \left[ \hat B_{\hat e_i \hat e_j}-A_{ij}\exp(-\bz_{ij}^\top \hat\pg)   \right]^2 $. We scale the loss functions by the covariate effect since the cross-validation is based on the block structure.

We present a theorem establishing the consistency of selecting $K$ using the proposed ECV Algorithm.
\begin{thm}[Consistency of Algorithm \ref{ecv_k_pcabm} under PCABM] \label{THM:ecv_consistency}
Let $A$ be the adjacency matrix generated by the undirected PCABM $(M, B, Z, \pg^0)$.
Assume Conditions \ref{cond:zbd}, \ref{cond:zpd} hold, and  
each element of $\bar B$ is bounded above by a constant $C_{\bar B}$, i.e. $\Bm\leq C_{\bar B}$.
The training proportion $p\in(0,1)$ is a constant.
The number of communities $K$ is fixed and to be estimated.
Further assume $\varphi_n/ \log n \to \infty$. Let $\hat K$ be the selected number of communities by using Algorithm \ref{ecv_k_pcabm} with the scaled $L_2$ loss. Then we have
$\Pr(\hat K < K ) \to 0$.

If we assume $\varphi_n / \sqrt{n} \to \infty $ and additionally assume all entries of $\bar B$ are bounded below by a constant $c_{\bar B} $, then the same result also holds for the scaled negative log-likelihood loss.
\end{thm}

\begin{algorithm}[h]
        \SetKwInOut{Input}{Input}
        \SetKwInOut{Output}{Output}
  \caption{ECV for selecting $K$ in PCABM}
  \label{ecv_k_pcabm}
  \Input{Adjacency matrix $A$, covariates $Z$,  the maximum number of communities to consider $K_{\max}$, training proportion $p$, number of replications $N_{rep}$.}
  \Output{Estimated number of communities $\hat K$.}
  Calculate MLE $\hat\pg$ with $A$, $Z$ with $\be$ being all 1 vector.

  \For{$m=1$ \KwTo $N_{rep} $}
  {
  Randomly choose a subset of node pairs $\Omega$: selecting each pair $(i, j), i < j$ independently
  with probability $p$, and adding $(j, i)$ if $(i, j)$ is selected.\\
  \For{$K=1$ \KwTo $K_{\max}$}
  {
  Apply matrix completion to $P_{\Omega} A'$ with rank constraint $K$ to obtain $\hat A'_K$, where $A'$ denotes the adjusted adjacency matrix $A_{ij}' = A_{ij} / \exp(\bz_{ij}^\top\hat\pg) $.\\
  Run spectral clustering on $\hat A'_K$ to obtain the estimated membership vector $\hat\be^{(m)}_K$.\\
  Estimate the probability matrix $\hat B^{(m)}_K$ with $ \hat B_{kl}(\hat\be,\hat\pg) = O_{kl}^{(\Omega)}(\hat\be)\slash E_{kl}^{(\Omega)}(\hat\be,\hat\pg) $, and
evaluate the corresponding losses $L^{(m)}_K$, by applying the loss function L with the 
 estimated parameters to $A_{ij} , (i, j) \in \Omega^c $.
  }}
  Let $L_K = \sum_{m=1}^{N_{rep}} L_K^{(m)} / N_{rep} $. Return $\hat K = \arg\min_{\{K=1,...,K_{\max}\}} L_K$ as the best model.
\end{algorithm} 

In addition to choosing the number of communities, another model selection problem of interest is distinguishing between covariate-adjusted and covariate-confounding models. 
As introduced in Section \ref{sec:intro}, in the covariate-adjusted model, covariates and class labels are independent, while in the covariate-confounding model, the distribution of covariates is governed by the community labels.

Given prior knowledge that a covariate is correlated with the block effect, one can extract cluster information from both the covariate and the network to improve the estimation accuracy of community labels. However, without that prior knowledge, fitting a confounding covariate in a covariate-adjusted model can undermine clustering performance. This phenomenon is illustrated in a simulation example provided in Section \ref{sec:ecv_z} of the Supplementary Material. A heuristic explanation is that the incorrect model mistakenly identifies the true underlying block effect as the covariate effect of the confounding covariate.

Motivated by the model-selection nature of the problem, we propose a covariate selection procedure based on the ECV framework. We present the detailed procedure in Algorithm \ref{ecv_z_pcabm} in Section \ref{sec:ecv_z} of the Supplementary Material and demonstrate that the proposed algorithm almost perfectly screens out false covariates and selects the correct model under various simulation study settings.

We note that while the proposed feature selection algorithm represents an interesting initial attempt to address confounding covariates, it is still based on the PCABM, which only models covariate adjusting. It would be desirable to propose a comprehensive covariate block model and corresponding community detection methods that could integrate both covariate-adjusted and covariate-confounding models.

\section{Simulations}\label{sec:simu}

For all simulations, we consider $K$ communities with prior probabilities $\pi_i=1/K, i=1,...,K$. In addition, we fix $\bar{B}$ to have all diagonal elements equaling 2 and off-diagonal elements 1; and we fix $K=2$ except in subsection \ref{sec:simu:chooseK} where $K$ varies.
We generate data by applying the following procedure:

S1. Determine parameters $\rho_n$ and $\mathbf{\gamma}^0$. Generate $\mathbf{z}_{ij}$ from certain distributions.

S2.  Generate adjacency matrix $A=[A_{ij}]$ from the Poisson distribution with the parameters estimated using PCABM with parameters in S1.

\subsection{$\boldsymbol{\gamma}$ Estimation\label{subsec::simu-gamma}}

For PCABM, estimating $\boldsymbol{\gamma}$ would be the first step, so we check the consistency and asymptotic normality of $\hat{\boldsymbol{\gamma}}$ claimed in our theory section. 

The pairwise covariate vector $\mathbf{z}_{ij}$ has five variables, generated independently from  $\text{Bernoulli}(0.1)$, $\text{Poisson}(0.1)$, $\text{Uniform}[0,1]$, $\text{Exponential}(0.3)$, and $N(0,0.3)$, respectively. The parameters for each distribution are chosen to make the variances of covariates similar.

We ran 100 simulations respectively for $n=100,300,500$. The parameters are set as $\rho_n=2\log n/n$, $\boldsymbol{\gamma}^0=(0.4,0.8,1.2,1.6,2)^T$. We obtained $\hat{\boldsymbol{\gamma}}$ by using BFGS to optimize the likelihood function under the initial community assignment $\be_0=\mathbf 1$ . 
We present the mean and standard deviation of $\hat{\boldsymbol{\gamma}}$ in Table~\ref{tab:gamma_K=1}. It is clear that $\hat{\boldsymbol{\gamma}}$ is very close to $\boldsymbol{\gamma}^0$ even for a small network. The shrinkage of standard deviation implies the consistency of $\hat{\boldsymbol{\gamma}}$. 
We also repeated the experiment by initializing with random community assignments, which leads to very similar results (Table \ref{tab:gamma} of Supplementary Materials). This validates the observation that estimating $\bm{\gamma}$ and communities is decoupled.
\begin{figure}[H]
    \centering
    \includegraphics[height=3cm]{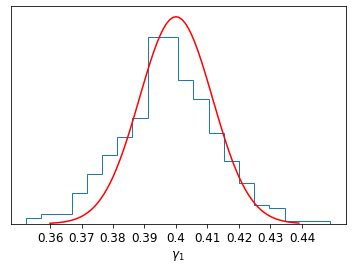}%
    \quad
    \includegraphics[height=3cm]{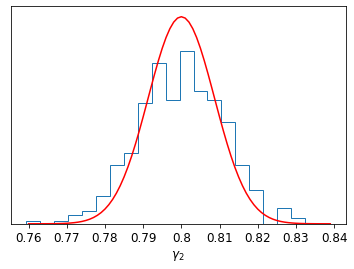}
    \quad
    \includegraphics[height=3cm]{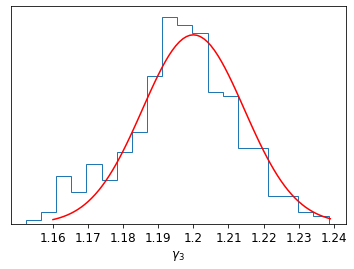}
    \caption{Simulation results for $\hat{\boldsymbol{\gamma}}$ compared with theoretical values. }
    \label{fig:gamma}
\end{figure}

By taking a closer look at the network of size $n=500$, we compare the distribution of $\hat{\boldsymbol{\gamma}}$ with the theoretical asymptotic normal distribution derived in Theorem~\ref{THM:ASY}. We  show the histogram for the first three coefficients in Figure~\ref{fig:gamma}. We can see that the empirical distribution matches well with the theoretical counterpart.
\begin{table}[!h]
    \centering
    \caption{Simulated results of $\hat{\boldsymbol{\gamma}}$ over 100 repetitions, displayed as mean (standard deviation).}
    \begin{tabular}{c|ccccc}
        \hline
        $n$&$\boldsymbol{\gamma}^0_1=0.4$&$\boldsymbol{\gamma}^0_2=0.8$&$\boldsymbol{\gamma}^0_3=1.2$&$\boldsymbol{\gamma}^0_4=1.6$&$\boldsymbol{\gamma}^0_5=2$\\ \hline
        {100}&0.393(0.0471)&0.796(0.0345)&1.206(0.0560)&1.596(0.0410)&2.005(0.0454)\\
        \hline
        {300}&0.399(0.0198)&0.801(0.0160)&1.198(0.0256)&1.603(0.0180)&2.003(0.0213)\\
        \hline
        {500}&0.399(0.0147)&0.800(0.0117)&1.197(0.0162)&1.599(0.0148)&2.002(0.0155)\\
        \hline
    \end{tabular}
    \label{tab:gamma_K=1}
\end{table}


\subsection{Community Detection\label{subsec:simu-cd}}
After obtaining $\hat{\boldsymbol{\gamma}}$, we now move on to the estimation of community labels. There are three parameters that we could tune to change the property of the network: $\boldsymbol{\gamma}^0$, $\rho_n$, and $n$. To illustrate the impact of these parameters on the performance of community detection,  we vary one parameter while fixing the remaining two in each experiment. More specifically, we consider the form $\rho_n = c_{\rho}\log n/n$ and $\boldsymbol{\gamma}^0 = c_{\gamma}(0.4, 0.8, 1.2, 1.6, 2)$ in which we will vary the multipliers $c_{\rho}$ and $c_{\gamma}$. The detailed parameter settings for the three experiments are as follows. 

(a) $n\in\{200,400,600,800,1000\}$, with $c_{\rho}=5$ and $c_{\gamma}=1.2$. 

(b) $c_{\rho}\in  \{2,3,4,5,6\}$, with $n=200$ and $c_{\gamma}=1.2$.

(c) $c_{\gamma} \in \{0,0.4,0.8,1.2,1.6,2.0\}$, with $n=200$ and $c_{\rho}=5$.

The results for the three experiments are presented in panels (a), (b), and (c) in Figure~\ref{fig:pcabm}. Each setting is simulated 100 times. The error rate is reported in terms of the average Adjusted Rand Index (ARI) \citep{hubert1985comparing}, which is a measure of the similarity between two data clusterings. SBM.MLE and SBM.SC refer to the likelihood and spectral clustering methods under SBM, respectively; DCBM.MLE is the maximum likelihood method based on DCBM \citep{zhao2012consistency}; PCABM.PL and PCABM.SCWA refer to Algorithms~\ref{alg:scwapcabm} and~\ref{alg:scwa}, respectively. 

\begin{figure}[h]
    \centering
    
    \begin{subfigure}{.45\textwidth}
  \includegraphics[width=.9\linewidth]{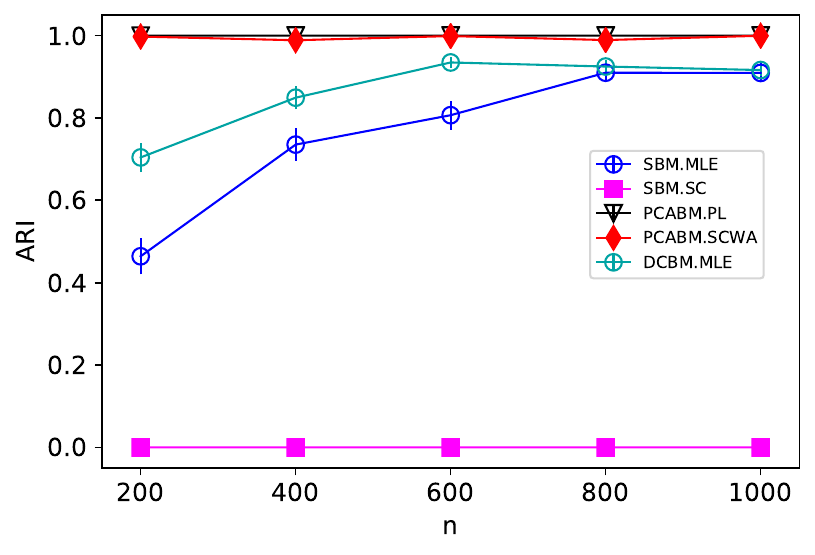}  
  \caption{Number of nodes $n$\label{sub:nodes}}
\end{subfigure}
\hspace{-2em}%
\begin{subfigure}{.45\textwidth}
  \includegraphics[width=.9\linewidth]{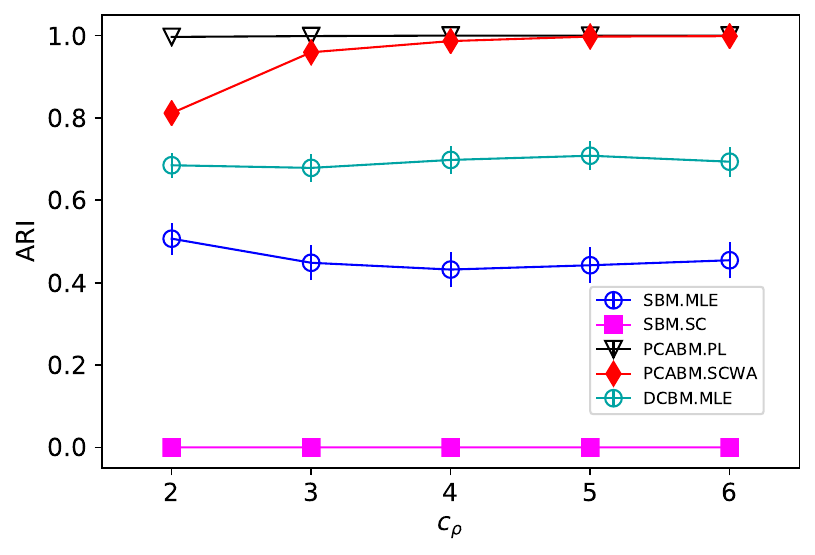}  
  \caption{Multiplier of sparsity $c_\rho$}
\end{subfigure}

\begin{subfigure}{.45\textwidth}
  \includegraphics[width=.9\linewidth]{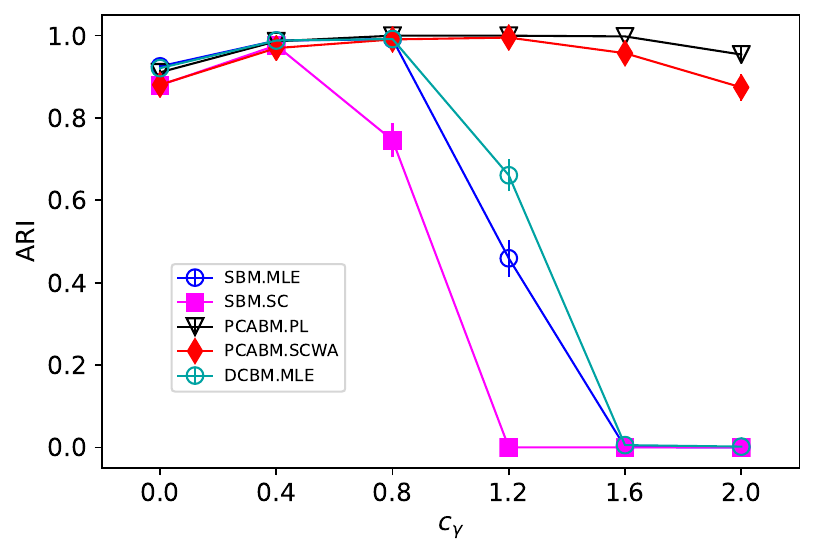}  
  \caption{Multiplier of coefficient $c_{\boldsymbol{\gamma}}$}
\end{subfigure}
\hspace{-2em}%
\begin{subfigure}{.45\textwidth}
  \includegraphics[width=.9\linewidth]{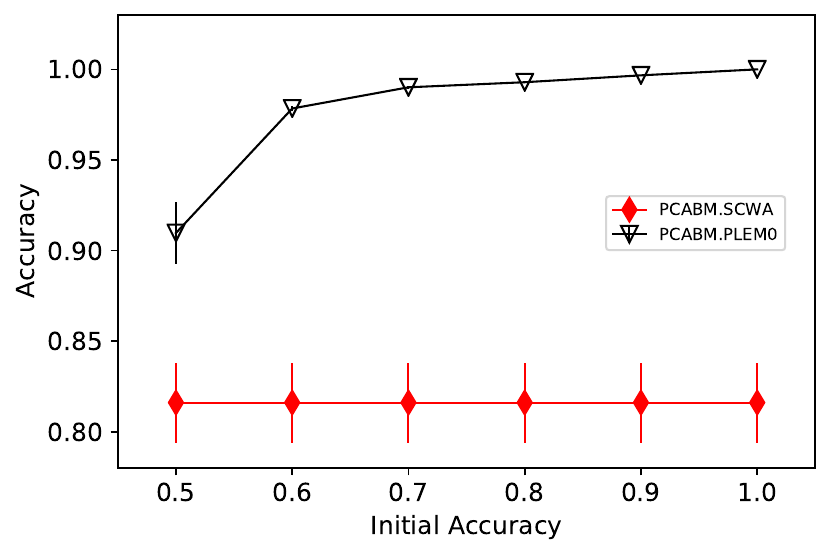}  
  \caption{Initial assignment accuracy}
  \label{fig:sub-second}
\end{subfigure}
    \caption{Simulation results under PCABM for different parameter settings}
    \label{fig:pcabm}
\end{figure}

As the number of nodes increases, it is evident from the first panel \ref{sub:nodes} in Figure \ref{fig:pcabm} that both PCABM-based algorithms perform exceptionally well, with PCABM.PL achieving nearly perfect community detection performance across all values of $n$. Spectral clustering under SBM results in nearly random guesses. DCBM and MLE under SBM perform better when $n$ is large but still underperform PCABM-based algorithms. As the density of the network increases, the performance does not change significantly within this range. When the scale of $\boldsymbol{\gamma}^0$ is changed, both PCABM algorithms continue to yield good results. As we know, when $\boldsymbol{\gamma}^0=\mathbf{0}$, our model reduces to SBM, so it is not surprising that SBM.MLE and SBM.SC both perform well when the magnitude of $\boldsymbol{\gamma}^0$ is relatively small and fail when the magnitude increases.




\subsection{Impact of Initial Assignments Accuracy} 
\label{subsec::simu_initial}
The performance of the pseudo-likelihood EM (Algorithm \ref{alg:pcabm}) depends on the initial assignments. To further understand its influence in our model, we simulate initial community assignments with different accuracy rates and examine how they affect prediction accuracy. The parameters are fixed to be $n=200, c_\rho=2, c_\gamma=1.5$. We change the accuracy of initial assignments from 0.5 to 1. To make the results easier to interpret, we use accuracy rather than ARI to evaluate performance. Note that SCWA does not use class assignment initialization, and we plot its accuracy as a reference flat line in panel (d) of Figure~\ref{fig:pcabm}. On one hand, even with completely random initial assignments, the PLEM algorithm yields satisfactory clustering accuracy. On the other hand, as the accuracy of initial assignments increases, the prediction accuracy of the PLEM method also improves. If we use the prediction of SCWA, with an accuracy of around 0.82, as the initial assignments for the PLEM method, we can enhance the prediction accuracy from around 0.9 (random initial) to almost 1. Therefore, it is preferable to use the output of SCWA as initial assignments for the PLEM method.


\subsection{DCBM\label{subsec::DCBM}}
\begin{figure}[h]
    \centering
    
    \begin{subfigure}{.45\textwidth}
  \includegraphics[width=.9\linewidth]{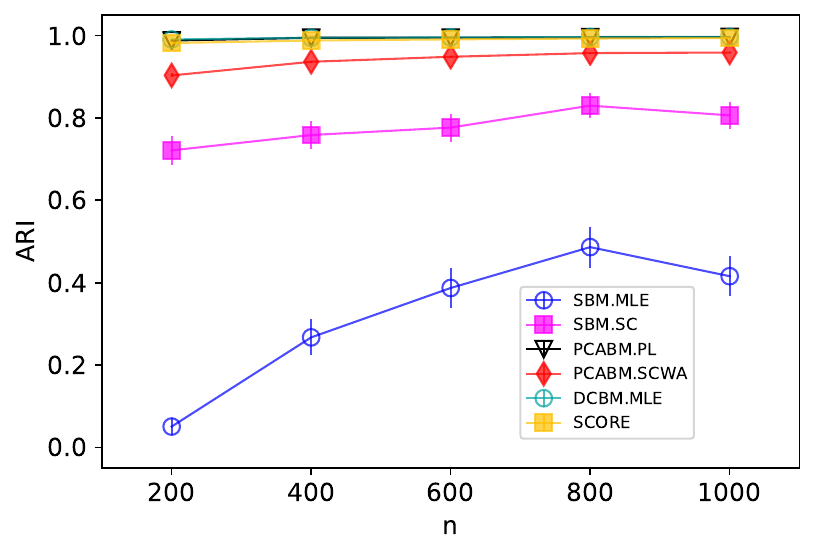}  
  \caption{Differrent $n$ with $c_\rho=3$}
\end{subfigure}
\begin{subfigure}{.45\textwidth}
  \includegraphics[width=.9\linewidth]{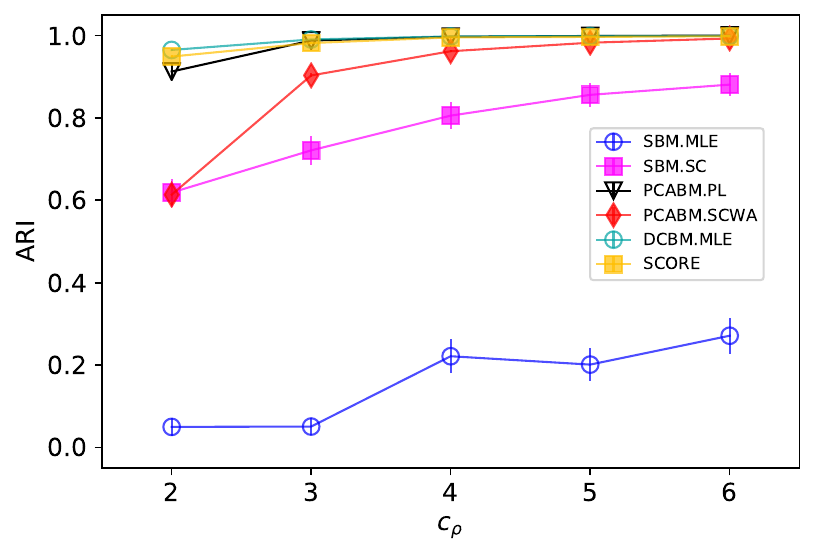}  
  \caption{Different $c_\rho$ with $n=200$}
  \label{fig:sub-second2}
\end{subfigure}

    \caption{Simulation results under DCBM for different parameter settings}
    \label{fig:deg}
\end{figure}

Considering that PCABM includes DCBM as a special case in terms of having the same profile likelihood, we are curious about the performance of Algorithms~\ref{alg:scwa} and~\ref{alg:scwapcabm} on networks generated by DCBM. The degree parameter for each node is chosen from $\{1,4\}$ with equal probability, $\bar{B} = \bigl( \begin{smallmatrix}2 & 1\\ 1 & 2\end{smallmatrix}\bigr)$, and $\rho_n = c_{\rho}\log n/n$. For covariates, we take $z_{ij}=\log d_i+\log d_j$, where $d_i$ is the degree of node $i$. As a comparison, we also implemented the likelihood method in \cite{zhao2012consistency} (DCBM.MLE) and the SCORE method in \cite{jin2015fast}. As in Section \ref{subsec:simu-cd}, we vary one parameter while fixing the remaining one in each experiment. The detailed parameter settings for the two experiments are as follows, with results presented in Figure \ref{fig:deg}.

(a) $n\in\{200,400,600,800,1000\}$, with $c_{\rho}=3$.

(b) $c_{\rho}\in \{2,3,4,5,6\}$, with $n=200$.

From the results, we observe that, except for SBM.MLE and SBM.SC, all the other methods work well, with the ARI being almost 1 when $n$ or $c_\rho$ is large. The flexibility of PCABM allows us to model any factors that may contribute to the network structure in addition to the underlying communities.

\subsection{Estimation of the Number of Communities $K$}
\label{sec:simu:chooseK}
In this subsection, we study the performance of our approach for selecting the number of communities $K$, i.e., Algorithm \ref{ecv_k_pcabm}. We set $\bar B$ to have diagonal elements of 2 and off-diagonal elements of 1. Additionally, we set $n=1000$ and $\rho_n = 5 \log n / n$. $\pg^0$ and covariates $Z$ are generated in the same way as in Section \ref{subsec::simu-gamma}. We consider cases where the true underlying $K$ is 2, 3, or 4, and let $K_{\max}=6$, i.e., selecting $\hat K$ from $\{1,2,...,6\}$. The simulation results are presented in Table \ref{ecv_k_25}.

\begin{table}[H]
\centering
\caption{Counts of ECV estimation of community number $K$ in 100 realizations under scaled negative log-likelihood (snll) and scaled $L_2$ loss.}
{\small
\begin{tabular}{c|c|c|c|c}
\hline
{Loss} & \multicolumn{2}{c|}{snll loss} & \multicolumn{2}{c}{scaled $L_2$ loss} \\
\hline
{$\hat K$} & $\Pr(\hat K = K)$ & $\Pr(\hat K \geq K)$  & $\Pr(\hat K = K)$ & $\Pr(\hat K \geq K)$\\
\hline
{$K=2$} & $100\%$ & $100\%$ & $91\%$ & $100\%$  \\
\hline
{$K=3$} & $99\%$ & $99\%$ & $91\%$ & $99\%$  \\
\hline
{$K=4$} & $95\%$ &$95\%$ & $74\%$ & $100\%$  \\
\hline
\end{tabular}
}
\label{ecv_k_25}
\end{table}
The results show that Algorithm \ref{ecv_k_pcabm} selects the correct $K$ with a high probability. Moreover, the scaled $L_2$ loss is more conservative than the snll loss, in the sense that it sometimes leads to a larger $\hat K$.

\section{Real Data Examples} \label{sec:realdata}

\subsection{Example 1: Political Blogs}\label{subsec::polblog}

The first real-world dataset we used is the network of political blogs created by \cite{adamic2005political}. The nodes represent blogs about US politics, and the edges indicate hyperlinks between them. We treated the network as undirected and focused only on the largest connected component of the network, resulting in a subnetwork with 1,222 nodes and 16,714 edges.

Since there are no other nodal covariates available in this dataset, we created one pairwise covariate by aggregating degree information. We set $z_{ij}=\log(d_i\times d_j)$, where $d_i$ is the degree for the $i$-th node. The coefficient estimate for the covariate $\hat{\boldsymbol{\gamma}}$ is 1.0005 with a 95\% confidence interval of $(0.9898,1.0111)$. Table~\ref{tab:politics} summarizes the performance comparison of PCABM with some existing results on this dataset. In addition to ARI, we also evaluated normalized mutual information (NMI) \citep{danon2005comparing}, which is a measure of mutual dependence.

We observed that the performance of our model is on par with previous methods designed specifically for DCBM, and the error rate is very close to the ideal results mentioned in \cite{jin2015fast}, which is 55/1222. This demonstrates that PCABM provides an alternative approach to DCBM by incorporating degree information into a specific pairwise covariate. As a more flexible model, PCABM also suggests that DCBM is indeed a suitable model for this dataset since the coefficient estimate is close to 1. This is consistent with our argument that PCABM includes DCBM from a profile likelihood perspective. Lastly, PCABM offers a significant improvement over the vanilla SBM, whose NMI is only 0.0001, as reported in \cite{karrer2011stochastic}.

\begin{table}[H]
    \centering
    {\small
    \begin{tabular}{c|ccccc}
        \hline
        &  DCBM.MLE & DCBM.RSC  & DCBM.CMM & SCORE
        & PCABM.PL   \\ \hline
        ARI & 0.819 &-- &--  & 0.819 &{0.813}  \\ 
        NMI    & 0.72 &-- &--        &  0.725&{0.725}  \\
        Errors       &-- &-- &62  &58&{60}\\ 
        Accuracy  &-- & 95\% & 94.9\% & 95.3\% & {95.1\%}   \\
               \hline
    \end{tabular} }
    \caption{Performance comparison on political blogs data. The performance of DCBM.MLE is taken from \cite{karrer2011stochastic,zhao2012consistency}; the performance of SCORE is from \cite{jin2015fast}; the performance of regularized spectral clustering (RSC) based on DCBM is reported in \cite{joseph2016impact}; the performance of convexified modularity maximization (CMM) for DCBM is from \cite{chen2018convexified}.}
    \label{tab:politics}
\end{table}

\subsection{Example 2: School Friendship}
In real networks, people often use specific nodal covariates as the ground ``truth" for community labels to evaluate the performance of various community detection methods. However, there could be different ``true" community assignments based on different nodal covariates (e.g., gender, job, and age). \cite{peel2017ground} mentioned that communities and covariates might capture various aspects of the network, which is in line with the idea presented in this paper. To examine whether PCABM can discover different community structures, in our second example, we treat one covariate as the indicator for the unknown ``true" community assignments while using the remaining covariates to construct the pairwise covariates in our PCABM model.

The dataset is a friendship network of school students from the National Longitudinal Study of Adolescent to Adult Health (Add Health). It contains 795 students from a high school (Grades 9-12) and its feeder middle school (Grades 7-8). The nodal covariates include grade, gender, ethnicity, and the number of friends nominated (up to 10). We focused on the largest connected component with at least one non-missing covariate and treated the network as undirected, resulting in a network with 777 nodes and 4,124 edges. For the nodes without gender, we assigned them to the female group, which is the smaller group. For those without grades, we generated a random grade within their schools.

Unlike traditional community detection methods that can only detect one underlying community structure, PCABM provides us with more flexibility to uncover different community structures by controlling different covariates. Our intuition is that social network is usually determined by multiple underlying structures and cannot be simply explained by one covariate. Sometimes one community structure seems to dominate the network, but if we adjust the covariate associated with that structure, we may discover other interesting community structures. 

In this example, we conducted two community detection experiments. In each experiment, out of the two nodal covariates, school and ethnicity, one was viewed as the proxy for the ``true" underlying community, and community detection was carried out using the pairwise covariates constructed with other covariates. For school and ethnicity, we created indicator variables to represent whether the corresponding covariate values were the same for the pair of nodes. For example, if two students come from the same school, the corresponding pairwise covariate equals 1; if they have different genders, the corresponding pairwise covariate equals 0. We also considered the number of nominated friends in all experiments and grades for predicting ethnicity and gender. For the number of nominated friends, we used $\log(n_i+1)+\log(n_j+1)$ as one pairwise covariate, where $n_i$ is the number of nominated friends by the $i$-th student. We added "+1" because some students did not nominate anyone. For grades, we used the absolute difference to form a pairwise covariate. Using random initial community labels, we computed the estimates $\hat{\boldsymbol{\gamma}}$ in each experiment. In Tables~\ref{tab:inf1} and~\ref{tab:inf2}, we show respectively the estimates when school and ethnicity are taken as the targeted community.

\begin{table}[h]
    \parbox{.45\linewidth}{\centering
        \caption{Inference results when school is targeted community.}
        {\small
        \begin{tabular}{c|cccc}
            \hline
            Covariate&Estimate&$t$ value&\text{Pr}($>|\text{t}|$)\\ \hline
            White&1.251&29.002&$<0.001$***\\ 
            Black&1.999&38.886&$<0.001$***\\
            Hispanic&0.048&0.091&0.927\\ 
            Others &0.019&0.035&0.972\\        
            Gender&0.192&5.620&$<0.001$***\\ 
            Nomination&0.438&18.584&$<0.001$***\\ \hline
        \end{tabular}
        }
        \label{tab:inf1}
        }
    \hspace{1cm}
    \parbox{.45\linewidth}{\centering
        \caption{Inference results when ethnicity is targeted community.}
        {\small
        \begin{tabular}{c|cccc}
            \hline
            Covariate&Estimate&$t$ value&\text{Pr}($>|\text{t}|$)\\ \hline
            School &1.005&13.168&$<0.001$***\\ 
            Grade&-1.100&-39.182&$<0.001$***\\        
            Gender&0.198&5.813&$<0.001$***\\         Nomination&0.498&21.679&$<0.001$***\\ \hline
        \end{tabular}
        }
        \label{tab:inf2}
        }
\end{table}

        


    
    

In both tables, the standard error is calculated using Theorem~\ref{THM:ASY}, with the theoretical values replaced by the estimated counterparts. Thus, we can calculate the $t$ value for each coefficient and perform the corresponding statistical tests. We can see that in both experiments, the coefficients for gender and the number of nominations are positive and significant in the creation of the friendship network. The significant positive coefficient of nominations shows that students with a large number of nominations have a higher chance to be friends with each other, which is intuitive. The positive coefficients of gender and school indicate that students of the same gender and school are more likely to be friends with each other, which aligns with our expectations. The negative coefficient of grade means that students with closer grades are more likely to be friends. If we examine the coefficients of different ethnic groups in Table~\ref{tab:inf1}, we find that only those corresponding to white and black are significant. This is understandable, as we observe that among 777 students, 476 are white, and 221 are black. As for school and grade, students in the same school or grade tend to be friends with each other, as expected.


The network is divided into two communities each time (we only look at white and black students in the second experiment because the sizes for other ethnicities are very small).
We apply our algorithm PCABM.PL, as well as some classic methods on SBM and DCBM, to cluster the network in both experiments.
The results in terms of ARI are shown in Table~\ref{tab:friendship}.
It can be seen that while DCBM can capture one main structure of the network, ``School," which is likely the dominating structure, our method can not only capture ``School" but also capture ``Race" when adjusting for the covariate ``School".
Note that for all methods other than ours, we would obtain only one community structure, whose performance is bound to be suboptimal for capturing different community structures. Additionally, to test the robustness of our method, in the experiment of detecting the ethnicity community, we tried using the square of the grade difference, which led to almost the same ARI.


        


\begin{table}[H]
    \centering
    \caption{ARI comparison on school friendship data.}
    {\small
    \begin{tabular}{c|ccccc}
        \hline
         &PCABM.PL& SBM.MLE   & SBM.SC &DCBM.MLE &  SCORE   \\ \hline
        School&{0.924}    & 0.048   &  0.043 &{0.909} &0.799  \\ 
                Race    &  {0.909}  &0.138 &-0.024 &0.001   &0.012       \\ 
 \hline
    \end{tabular}
    }
    \label{tab:friendship}
\end{table}

\section{Discussion}
\label{sec:discussion}
In this paper, we extend the classical stochastic block model to allow the connection rate between nodes to depend on not only the community memberships but also the pairwise covariates. We prove consistency in terms of both coefficient estimates and community label assignments for MLE under PCABM, and provide an efficient algorithm to solve an approximate MLE. Additionally, we introduce a fast spectral method, SCWA, with theoretical justification, which could serve as a good initial solution for the likelihood-based method. Furthermore, we propose cross-validation-based algorithms for estimating the number of communities and feature selection.

There are many interesting future research directions on PCABM. In our paper, we assume the entries in the adjacency matrix are non-negative integers. However, this can be relaxed to be any non-negative numbers, and we expect similar theoretical results to hold. It would also be interesting to consider highly imbalanced community sizes, where $n_{\min}/n_{\max}=o(1)$. Moreover, when we have high-dimensional pairwise covariates, adding a penalty term to conduct variable selection is worth investigating. For instance, in the estimation of $\pmb\gamma$, we can regularize \eqref{eq::gamma_hat} with an $L_1$ penalty $\hat\pg_{\lambda}(\be_0) = \arg\max_{\pg} \left\{\ell_{\be_0}(\pg) - \lambda\|\pg\|_1 \right\} $ to estimate a sparse high-dimensional $\pg$.

One model assumption in PCABM is the independence among edges conditional on observed covariates. However, the independence might be inappropriate if there are unobserved covariates. To address this, one possible extension is a degree-corrected pairwise covariate-adjusted block model, which can incorporate unobserved nodal covariates. The adjacency matrix could be modeled as, for example, $A_{ij} |\bc, Z, \pmb\theta \sim \text{Poisson}(B_{c_ic_j} \theta_i\theta_j \exp(\bz_{ij}^\top \pg))$, where $\pmb\theta$ represents degree correction parameters. From a modeling perspective, the $\pmb\theta$ term could be one way of incorporating unobserved nodal covariates or random effects. From a model fitting point of view, the first question to ask about this model is whether it is, in some sense, equivalent to PCABM by adding the covariate $\log(d_id_j)$, where $d_i$ is the degree of the $i$th node, or the degree after scaling by the covariate effect.


The code for implementing the proposed algorithms is available on GitHub.

\section*{Acknowledgements}
We thank the editor, the AE, and anonymous reviewers for their insightful comments which have greatly improved the scope and quality of the paper. \if1\blind
{This work was supported by NSF CAREER Grant DMS-2013789, NIH grant 1R21AG074205-01, NYU University Research Challenge Fund, and a grant from NYU School of Global Public Health.} \fi
\spacingset{1}

\bibliographystyle{biometrika}
\bibliography{reference}

\newpage

\setcounter{page}{1}
\setcounter{section}{1}
\renewcommand{\theequation}{A.\arabic{equation}}

\renewcommand{\thetable}{A.\arabic{table}}
\renewcommand{\thelemma}{\thesubsection.\arabic{lemma}}

\setcounter{table}{0}
\setcounter{equation}{0}

\title{\bf Supplementary Materials for ``Pairwise Covariates-Adjusted Block Model for Community Detection"}

\date{}
\maketitle


\spacingset{1.4} 

The Supplementary Material contains the proofs of the theoretical results, presents some technical lemmas, and shows additional simulation results. 

Section \ref{A} presents technical details of the theoretical results. In particular, 
Section \ref{sec:supple:proof_thm1} presents the proof of Theorem \ref{THM:ASY} (consistency and asymptotic normality of MLE of $\pg$). Section \ref{sec:supple:concen_lemmas} presents some concentration inequality lemmas which are going to be  repeatedly used in the proofs. Section \ref{sec:supple:proof_mle} presents the proof of Theorem \ref{THM:MLE_fixK} (consistency of the MLE $\hat\bc$ of community labels when $K$ is fixed); Section \ref{append_sect:KtoinftyMLE} gives a consistency result of the MLE $\hat\bc$ when $K$ grows as fast as $O(\sqrt{n})$. In Section \ref{sec::supple_plem} we derive the PLEM algorithm in detail and establish a theoretical guarantee for the algorithm in the two balanced communities case. Section \ref{sec:supple:spectralbd} presents the proof of Theorem \ref{THM:SC} (spectral bound of adjusted, regularized Poisson random matrices). Section \ref{sec:supple:proof_chooseK} presents the proof of Theorem \ref{THM:ecv_consistency} (consistency of the choose $K$ algorithm under PCABM).

In Section \ref{sec:ecv_z} we discuss the problem of confounding covariates and propose a covariate selection procedure whose utility is illustrated in simulation examples.

Section \ref{sec:supple:addi_simu_realdata} presents some additional results in simulation and real data examples. Section \ref{subsec:supple:addi_randominit_e} presents the simulation result when estimating $\gamma$ with random initial community assignment $\be$. Section \ref{subsec:supple:schoolfriend_visualize} presents a visualization for the estimated clusters in the school friendship data.
\\

\appendix
\renewcommand{\theequation}{\thesection.\arabic{equation}}

\section{Proofs}\label{A} 
\subsection{Proof of Theorem~\ref{THM:ASY}}
\label{sec:supple:proof_thm1}

\begin{proof}

In the following proof, we will use $\hat{\boldsymbol{\gamma}}$ instead of $\hat{\boldsymbol{\gamma}}(\mathbf{e}_0)$ for simplicity.
Since 
\begin{equation}
\label{lgammaTaylor}
- l_{\be_0}'(\pg^0) =l_{\be_0}'(\hat\pg) - l_{\be_0}'(\pg^0) = l_{\be_0}''(\bar\pg) (\hat\pg - \pg^0)
\end{equation}
where $\bar\pg = q \hat\pg + (1-q) \pg^0$ for some $q\in[0,1]$, we want to analyze the asymptotic behavior of $l_{\be_0}'(\pg^0)$.
 Define the empirical version of $\theta(\boldsymbol{\gamma})$, $\boldsymbol{\mu}(\boldsymbol{\gamma})$ and $\Sigma(\boldsymbol{\gamma})$ as
\begin{align*}
	\hat{\theta}(\boldsymbol{\gamma})=&\sum_{u,v\in [n],u\neq v}e^{\mathbf{z}_{uv}^T\boldsymbol{\gamma}}/(n(n-1)),\\
	\hat{\boldsymbol{\mu}}(\boldsymbol{\gamma})=&\sum_{u,v\in [n],u\neq v}\mathbf{z}_{uv}e^{\mathbf{z}_{uv}^T\boldsymbol{\gamma}}/(n(n-1)),\\
	\hat{\Sigma}(\boldsymbol{\gamma})=&\sum_{u,v\in [n],u\neq v}\mathbf{z}_{uv}\mathbf{z}_{uv}^Te^{\mathbf{z}_{uv}^T\boldsymbol{\gamma}}/(n(n-1)).
\end{align*}
For fixed $\boldsymbol{\gamma}$, by Chebyshev's inequality, we know the weak law of large numbers holds, i.e., $\hat{\theta}(\boldsymbol{\gamma})\xrightarrow{p}\theta(\boldsymbol{\gamma})$, $\hat{\boldsymbol{\mu}}(\boldsymbol{\gamma})\xrightarrow{p}\boldsymbol{\mu}(\boldsymbol{\gamma})$ and $\hat{\Sigma}(\boldsymbol{\gamma})\xrightarrow{p}\Sigma(\boldsymbol{\gamma})$. 

For the given cluster assignment $\be_0$, the log-likelihood for covariate coefficient $\pg$ is 
\begin{equation}
l_{\be_0}(\pg) = \sum_{i<j} A_{ij} \bz_{ij}^\top \pg - \frac12 \sum_{kl} O_{kl}(\be_0) \log E_{kl}(\be_0,\pg).
\end{equation}
Note that the likelihood is a concave random function of $\pg$. Thus a direct application of Theorem II.1 and Corollary II.2 of \cite{andersen1982cox} gives the consistency of the MLE $\hat\pg$.

Next we show the asymptotic normality of $\hat\pg$. The score function for $\pg$ is given by
\begin{equation}
l_{\be_0}'(\pg) = \sum_{i<j} A_{ij} \left[ \bz_{ij} - \frac{\hat\pmu(\pg)}{\hat\theta(\pg)} \right]
\end{equation}
which could be decomposed into two parts
\begin{equation}
\label{gm_twoparts}
\begin{aligned}
l_{\be_0}'(\pg)  &=: \uppercase\expandafter{\romannumeral1}(\pg) + \uppercase\expandafter{\romannumeral2}(\pg); \\
\uppercase\expandafter{\romannumeral1}(\pg)&= \sum_{i<j} A_{ij} \left[ \bz_{ij} - \frac{\pmu(\pg)}{\theta(\pg)} \right]; \\
\uppercase\expandafter{\romannumeral2}(\pg)&= \sum_{i<j} A_{ij} \left[ \frac{\pmu(\pg)}{\theta(\pg)}
 - \frac{\hat\pmu(\pg)}{\hat\theta(\pg)} \right].
\end{aligned}
\end{equation}
In the decomposition \eqref{gm_twoparts}, conditioning on $\bc$, $\uppercase\expandafter{\romannumeral1}(\pg^0)$ is a sum of independent random variables. The mean and variance of each summand in $\uppercase\expandafter{\romannumeral1}(\pg^0)$ (scaled by $\rho_n$) are given by
\begin{align*}
\mathbb E \left[\left. A_{ij} \left( \bz_{ij} - \frac{\pmu(\pg^0)}{\theta(\pg^0)} \right) / \rho_n \right| Z \right]=& \bar B_{c_ic_j} \ezg \left( \bz_{ij} - \frac{\pmu(\pg^0)}{\theta(\pg^0)} \right), \\
\mathbb E \left[ A_{ij} \left( \bz_{ij} - \frac{\pmu(\pg^0)}{\theta(\pg^0)} \right) / \rho_n  \right] =& 0,
\end{align*}
\begin{align*}
Var\left[ A_{ij} \left( \bz_{ij} - \frac{\pmu(\pg^0)}{\theta(\pg^0)} \right) / \rho_n  \right] =& \mathbb E \left\{ Var\left[\left. A_{ij} \left( \bz_{ij} - \frac{\pmu(\pg^0)}{\theta(\pg^0)} \right) / \rho_n \right| Z \right] \right\} \\
&+  Var\left\{ \mathbb E\left[\left. A_{ij} \left( \bz_{ij} - \frac{\pmu(\pg^0)}{\theta(\pg^0)} \right) / \rho_n \right| Z \right] \right\}\\
=& \mathbb E\left[ \frac{\bar B_{c_ic_j} \ezg}{\rho_n} \left( \bz_{ij} - \frac{\pmu(\pg^0)}{\theta(\pg^0)} \right)^{\otimes 2} \right] \\
&+ \mathbb E\left[ {\bar B_{c_ic_j} \ezg} \left( \bz_{ij} - \frac{\pmu(\pg^0)}{\theta(\pg^0)} \right) \right]^{\otimes 2} \\
=& \frac{\bar B_{c_ic_j} }{\rho_n} \left[ \Sigma(\pg^0) - \frac{\pmu(\pg^0)^{\otimes 2}}{\theta(\pg^0)} \ \right] (1+o(1)).
\end{align*}
Thus, by Lyapunov CLT, (noting that the third central moment of Poisson$(\lambda)$ is $\lambda$,) we have 
\begin{equation}
\frac{\uppercase\expandafter{\romannumeral1}(\pg^0)}{\sqrt {\rho_n}} \xrightarrow{d} N\left(\mathbf 0, \sum_{i<j} \bar B_{c_ic_j} \left[ \Sigma(\pg^0) - \frac{\pmu(\pg^0)^{\otimes 2}}{\theta(\pg^0)} \ \right] \right);
\end{equation}
and unconditioning on $\bc$, we obtain
\begin{equation}
\frac{\uppercase\expandafter{\romannumeral1}(\pg^0)}{\sqrt {N_n \rho_n}} \xrightarrow{d} N\left(\mathbf 0,   \Sigma_{\infty}(\pg^0)  \right)
\end{equation}
from the U-statistic type LLN $\sum_{i<j} \bar B_{c_ic_j} / N_n \xrightarrow{d} \bar B_{\lim}$.

Now we analyze part ${\uppercase\expandafter{\romannumeral2}(\pg^0)}$ in $l'_{\be_0}(\pg^0)$. $\sum_{i<j} A_{ij}$ is a sum of independent random variables conditioning on $\bc$ and by triangular array WLLN
\begin{equation}
\frac{\sum_{i<j} A_{ij}}{N_n \rho_n} \xrightarrow{d} \frac{\sum_{i<j}\bar B_{c_ic_j} \mathbb E\ezg  }{N_n};
\end{equation}
and unconditioning on $\bc$ we obtain
\begin{equation}
\label{AijWLLN}
\frac{\sum_{i<j} A_{ij}}{N_n \rho_n} \xrightarrow{d} {\bar B_{\lim} \theta(\pg^0)}.
\end{equation}
$\hat\pmu(\pg^0)$ and $\hat\theta(\pg^0)$ are both averages of independent random variables so by CLT we have
\begin{align*}
\sqrt{N_n} \begin{pmatrix} \hat\pmu(\pg^0) - \pmu(\pg^0) \\\hat\theta(\pg^0) - \theta(\pg^0)  \end{pmatrix} \xrightarrow{d} N\left(\mathbf 0, \begin{bmatrix} Var(\bz_{ij}\ezg) & Cov(\bz_{ij}\ezg, \ezg) \\ Cov(\bz_{ij}\ezg, \ezg)& Var(\bz_{ij}\ezg) \end{bmatrix} \right).
\end{align*}
Since $Z$ is bounded as is assumed in Condition \ref{cond:zbd}, by delta method we could see $ \sqrt{N_n} \left( \frac{\hat\pmu(\pg^0) }{\hat\theta(\pg^0)} - \frac{\pmu(\pg^0) }{\theta(\pg^0)} \right)$ converges to a certain normal distribution with bounded variance. Thus, ${\uppercase\expandafter{\romannumeral2}(\pg^0)} = \sum_{i<j} A_{ij} \left[ \frac{\hat\pmu(\pg^0) }{\hat\theta(\pg^0)} - \frac{\pmu(\pg^0) }{\theta(\pg^0)} \right] $ is of the order $O_p(\rho_n\sqrt{N_n}) $ while ${\uppercase\expandafter{\romannumeral1}(\pg^0)}$ is  of the order $O_p(\sqrt{N_n\rho_n})$. We could now conclude that $l'_{\be_0}(\pg^0) = {\uppercase\expandafter{\romannumeral1}(\pg^0)} (1+o_p(1)) $, and hence ${l'_{\be_0}(\pg^0)}/{\sqrt {N_n \rho_n}} \xrightarrow{d} N\left(\mathbf 0,   \Sigma_{\infty}(\pg^0)  \right)$.

A direct calculation gives us 
\begin{equation}
\label{ldoubleprime}
l''_{\be_0}(\pg) = \sum_{i<j} A_{ij} \left[ \frac{\hat\pmu(\pg)^{\otimes2}}{\hat\theta(\pg)} - \hat\Sigma(\pg)  \right].
\end{equation}
By a typical argument of uniform weak law of large numbers followed by continuous mapping theorem, we get
$$
\frac{\hat\pmu(\bar\pg)^{\otimes2}}{\hat\theta(\bar\pg)} - \hat\Sigma(\bar\pg) \xrightarrow{d} \frac{\pmu(\pg^0)^{\otimes2}}{\theta(\pg^0)} - \Sigma(\pg^0)
$$
where $\bar \pg$ is a mean value of $\hat\pg$ and $\pg^0$.
Thus
\begin{equation}
\frac{l''_{\be_0}(\bar\pg)}{N_n\rho_n} \xrightarrow{d} \Sigma_{\infty}(\pg^0).
\end{equation}
Substituting the above result and the asymptotic normality of $l'_{\be_0}(\pg^0)$ back into equation \eqref{lgammaTaylor} finishes the proof.
\end{proof}

\subsection{Some Concentration Inequalities and Notations}
\label{sec:supple:concen_lemmas}

To prepare later proofs, we introduce some concentration inequalities and additional notations in this part.

One inequality that we will apply repeatedly is an extended version of Bernstein inequality for unbounded random variables introduced in \cite{wellner2005empirical}. 
\begin{lemma}[Bernstein inequality]\label{lem:bernin}
	Suppose $X_1,\cdots,X_n$ are independent random variables with $\mathbb{E}X_i=0$ and $\mathbb{E}|X_i|^k\leq\frac{1}{2}\mathbb{E}X_i^2L^{k-2}k!$ for $k\geq2$. For $M\geq\sum_{i\leq n}\mathbb{E}X_i^2$ and $x\geq0$,
	$$\emph{Pr}(\sum_{i\leq n}X_i\geq x)\leq\exp\left(-\frac{x^{2}}{2(M+xL)}\right).$$
\end{lemma}

To show that all Poisson distributions satisfy the above Bernstein condition uniformly under some constant $\bar{L}$, we give the following lemma.
\begin{lemma}[Bernstein condition]\label{lem:bern}
Assume $A\sim Pois(\lambda)$, let $X=A-\lambda$, then for any $0<\lambda<1/2$,  there exists a constant $\bar{L}>0$ s.t. for any integer $k>2$, $\mathbb{E}[|X^k|]\leq\mathbb{E}[X^{2}]\bar{L}^{k-2}k!/2$.  
\end{lemma}

\begin{proof}
\begin{align*}
	&\frac{2\mathbb{E}[|A-\lambda|^k]}{\lambda k!}=\frac{2}{\lambda k!}\mathbb{E}[(A-\lambda)^k|A\geq1]\text{Pr}(A\geq1)+\frac{2\lambda^{k-1}e^{-\lambda}}{k!}\\
	\leq&\frac{2}{\lambda k!}\mathbb{E}[A^k|A\geq1]\text{Pr}(A\geq1)+e^{-\lambda}=\frac{2}{\lambda k!}\mathbb{E}[A^k]+e^{-\lambda}\\
	=&\frac{2}{k!}\sum_{i=1}^k\begin{Bmatrix}k\\i\end{Bmatrix}\lambda^{i-1}+e^{-\lambda}\leq\frac{1}{k!}\sum_{i=1}^k{k \choose i}i^{k-i}\lambda^{i-1}+e^{-\lambda}\\
	\leq&\frac{e^{k-1}}{k^k}\sum_{i=1}^k\left(\frac{ek}{i}\right)^ii^{k-i}\lambda^{i-1}+e^{-\lambda}=\sum_{i=1}^ke^{i+k-1}i^{k-2i}k^{i-k}\lambda^{i-1}+e^{-\lambda}\\
	<&\sum_{i=1}^ke^{i+k-1}e^{-i}\lambda^{i-1}+e^{-\lambda}=e^{k-1}\frac{1-\lambda^k}{1-\lambda}+e^{-\lambda}\leq\frac{e^{k-1}}{1-\lambda}+1\\
	\leq&\left(\frac{e^2+1}{1-\lambda}\right)^{k-2}.
\end{align*}
Notice that when $\lambda$ is bounded away from $1$, say $\lambda<1/2$, we can simply set $\bar{L}=2(e^2+1)$, then Bernstein condition is satisfied uniformly for all $\lambda$.
\end{proof}

We introduce some notations. Let $|\mathbf{e}-\mathbf{c}| = \sum_{i=1}^n\mathbbm{1}(e_i\neq c_i)$. 
Given a community assignment $\mathbf{e}\in[K]^n$, we define $R(\mathbf{e})\in\mathbb{R}^{K\times K}$ with its elements being $R_{ka}(\mathbf{e})=\frac{1}{n}\sum_{i=1}^n\mathbbm{1}(e_i=k, c_i=a)$,
and define $ V(\mathbf{e})\in\mathbb{R}^{K\times K}$ with their elements being
\begin{align*}
V_{ka}(\mathbf{e})=\frac{\sum_{i=1}^n\mathbbm{1}(e_i=k,c_i=a)}{\sum_{i=1}^n\mathbbm{1}(c_i=a)}=\frac{R_{ka}(\mathbf{e})}{\pi_a(\mathbf{c})}.
\end{align*}

One can view $R$ as the empirical joint distribution of $\mathbf{e}$ and $\mathbf{c}$, and $V$ as the empirical conditional distribution of $\mathbf{e}$ given $\mathbf{c}$. We can see that $V(\mathbf{e})=R(\mathbf{e})(D(\mathbf{c}))^{-1}$, where $D(\mathbf{c})=\text{diag}(\boldsymbol{\pi}(\mathbf{c}))$. Also, note that $V(\mathbf{e})^T\mathbf{1}=\mathbf{1}$, $V(\mathbf{e})\boldsymbol{\pi}(\mathbf{c})=\boldsymbol{\pi}(\mathbf{e})$ and $V(\mathbf{c})=I_K$. For the convenience of later proof, we also define $W(\mathbf{c})=D(\mathbf{c})\bar{B}D(\mathbf{c})$ and
\begin{align*}
	\hat{T}(\mathbf{e})\triangleq& R(\mathbf{e})\bar{B}R(\mathbf{e})^T=V(\mathbf{e})W(\mathbf{c})V(\mathbf{e})^T,\\
	\hat{S}(\mathbf{e})\triangleq& V(\mathbf{e})\boldsymbol{\pi}(\mathbf{c})\boldsymbol{\pi}(\mathbf{c})^TV(\mathbf{e})^T
\end{align*}

Replacing the empirical distribution $\boldsymbol{\pi}(\mathbf{c})$ by the true distribution $\boldsymbol{\pi}_0$, we define $W_0=D(\boldsymbol{\pi}_0)\bar{B}D(\boldsymbol{\pi}_0)$, where $D(\boldsymbol{\pi}_0)=\text{diag}(\boldsymbol{\pi}_0)$, and $T(\mathbf{e}), S(\mathbf{e})\in\mathbb{R}^{K\times K}$ as
\begin{align*}
T(\mathbf{e})\triangleq& V(\mathbf{e})W_0V(\mathbf{e})^T,\\
S(\mathbf{e})\triangleq& V(\mathbf{e})\boldsymbol{\pi}_0\boldsymbol{\pi}_0^TV(\mathbf{e})^T.
\end{align*}

The population version of $F\left(\frac{O}{2N_n\rho_n},\frac{E}{2N_n}\right)$ is
$$F(\theta(\boldsymbol{\gamma}^0)T(\mathbf{e}),\theta(\hat{\boldsymbol{\gamma}})S(\mathbf{e})).$$ 
To measure the discrepancy between empirical and population version of $F$, we define $X(\mathbf{e}), Y(\mathbf{e},\hat{\boldsymbol{\gamma}})\in\mathbb{R}^{K\times K}$ to be the rescaled difference between $O, E$ and their expectations
\begin{align*}
X(\mathbf{e})\triangleq&\frac{O(\mathbf{e})}{2N_n\rho_n}-\theta(\boldsymbol{\gamma}^0)\hat{T}(\mathbf{e}),\\
Y(\mathbf{e},\hat{\boldsymbol{\gamma}})\triangleq&\frac{E(\mathbf{e},\hat{\boldsymbol{\gamma}})}{2N_n}-\theta(\hat{\boldsymbol{\gamma}})\hat{S}(\mathbf{e}).
\end{align*}

Before we establish bound for $Y(\mathbf{e},\hat{\boldsymbol{\gamma}})$, we present the following lemma for $\hat{\boldsymbol{\gamma}}$. 
\begin{lemma}\label{lem:phi}
	For any constant $\phi>0$, $\exists$ positive constants $C_\phi$ and $v_\phi$ s.t., $\emph{Pr}(\|\hat{\boldsymbol{\gamma}}-\boldsymbol{\gamma}^0\|_\infty>\phi)<C_\phi\exp(-v_\phi N_n\rho_n)$.
\end{lemma}
We omit the proof of the Lemma as it  follows easily from the proof of Theorem~\ref{THM:ASY}. Conditioned on $\|\hat{\boldsymbol{\gamma}}-\boldsymbol{\gamma}^0\|_\infty\leq\phi$, we have $|e^{\mathbf{z}_{ij}\hat{\boldsymbol{\gamma}}}-\mathbb{E}e^{\mathbf{z}_{ij}\hat{\boldsymbol{\gamma}}}|\leq\exp\{p\alpha(\phi+\|\boldsymbol{\gamma}^0\|_\infty)\}\equiv\chi$ uniformly for any $i,j\in[n]$ and $\hat{\boldsymbol{\gamma}}$. Under this condition, we establish Lemma~\ref{lem:con} using Bernstein inequality.

\begin{lemma}\label{lem:con}
\begin{align}
\emph{Pr}(\max_{\mathbf{e}}\|X(\mathbf{e})\|_\infty\geq\epsilon)&\leq2K^{n+2}\exp(-C_1\epsilon^2N_n\rho_n)\label{eq:lm1}
\end{align}
for $\epsilon<\chi\|\bar{B}\|_{\max}/\bar{L}$.
\begin{align}\label{eq:lm31}
\begin{split}
&\emph{Pr}(\max_{|\mathbf{e}-\mathbf{c}|\leq m}\|X(\mathbf{e})-X(\mathbf{c})\|_\infty\geq\epsilon)\\
\leq&2{n \choose m}K^{m+2}\exp\left(-\frac{C_3n}{m}\epsilon^2N_n\rho_n\right)
\end{split}
\end{align}
for $\epsilon<\eta m/n$, where $\eta=2\chi\|\bar{B}\|_{\max}/\bar{L}$.
\begin{align}\label{eq:lm32}
\begin{split}
&\emph{Pr}(\max_{|\mathbf{e}-\mathbf{c}|\leq m}\|X(\mathbf{e})-X(\mathbf{c})\|_\infty\geq\epsilon)\leq2{n \choose m}K^{m+2}\exp\left(-C_4\epsilon N_n\rho_n\right)
\end{split}
\end{align}
for $\epsilon\geq\eta m/n$.
\begin{align}
\emph{Pr}(\max_{\mathbf{e}}\|Y(\mathbf{e},\hat{\boldsymbol{\gamma}})\|_\infty\geq\epsilon)&\leq2K^{n+2}
\max\left\{\left(\frac{C_M}{\epsilon\sqrt{n}}\right)^p,1 \right\}
\exp(-C_2\epsilon^2N_n)\label{eq:lm2}
\end{align}
for $\epsilon<\chi \kappa_2^2$, where $C_M$ is a constant.
\begin{align}\label{eq:lm41}
\emph{Pr}(\max_{|\mathbf{e}-\mathbf{c}|\leq m}\|Y(\mathbf{e},\hat{\boldsymbol{\gamma}})-Y(\mathbf{c},\hat{\boldsymbol{\gamma}})\|_\infty\geq\epsilon)&\leq2{n \choose m}K^{m+2}
\max\left\{\left(\frac{C_M}{\epsilon\sqrt{n}}\right)^p,1 \right\}
\exp\left(-\frac{C_5n}{m}\epsilon^2N_n\right)
\end{align}
for $\epsilon<\frac{2\chi m}{n}$.
\begin{align}\label{eq:lm42}
\emph{Pr}(\max_{|\mathbf{e}-\mathbf{c}|\leq m}\|Y(\mathbf{e},\hat{\boldsymbol{\gamma}})-Y(\mathbf{c},\hat{\boldsymbol{\gamma}})\|_\infty\geq\epsilon)&\leq2{n \choose m}K^{m+2}
\max\left\{\left(\frac{C_M}{\epsilon\sqrt{n}}\right)^p,1 \right\}
\exp\left(-C_6\epsilon N_n\right)
\end{align}
for $\epsilon\geq\frac{2\chi m}{n}$.
\end{lemma}

\begin{proof}
The proofs are all given conditioned on $|e^{\mathbf{z}_{ij}\hat{\boldsymbol{\gamma}}}-\mathbb{E}e^{\mathbf{z}_{ij}\hat{\boldsymbol{\gamma}}}|\leq\chi$. By combining Lemma~\ref{lem:phi}, we could have the conclusion directly. For any fixed $\be$ and $\hat{\boldsymbol{\gamma}}$, by Bernstein inequality, when $\epsilon<\chi \kappa_2^2$,
\begin{align*}
	&\text{Pr}(|Y_{kl}(\mathbf{e},\hat{\boldsymbol{\gamma}})|\geq\epsilon)\leq2\exp\left(-\frac{\frac{1}{2}(2N_n\epsilon)^2}{|s_{\mathbf{e}}(k,l)|\chi^2+\frac{2}{3}\chi N_n\epsilon}\right)\\
	\leq&2\exp\left(-\frac{6N_n\epsilon^2}{3\kappa_2^2\chi^2+2\chi\epsilon}\right)\leq2\exp\left(-\frac{6}{5\kappa_2^2\chi^2}\epsilon^2N_n\right).
\end{align*}
Note that in the above concentration inequality we are seeing $\hat\pg$ as fixed. To transfer this result to a random $\hat\pg$, we note that $\|\hat\pg - \pg^0\|_1 \sim O_P(1/\sqrt{N_n\rho_n}) $ from Theorem \ref{THM:ASY}. Thus it suffices to take supremum over $\pg \in B(\pg^0,C_M/\sqrt{n})$ as $n\rho_n\to\infty $, where $B(x,r)$ denotes a ball with center at $x$ and radius $r$, and $C_M$ is a  constant. By applying a mid-value theorem, we have $|Y_{kl}(\be,\pg) - Y_{kl}(\be,\pg')|\leq 3\chi \|\pg - \pg'\|_1  $ for $\pg,\pg'\in B(\pg^0,C_M/\sqrt{n})$. Therefore it suffices to take the supremum over grids with distance, say, $\epsilon/12$; with that we obtain
\begin{equation}
\label{eq::Y_sup_be_gamhat}
\begin{aligned}
\text{Pr}(\max_{\be} |Y_{kl}(\mathbf{e},\hat{\boldsymbol{\gamma}})|\geq\epsilon) \leq& \text{Pr}\left(\sup_{\pg\in B(\pg^0,C_M/\sqrt{n})} \max_{\be} |Y_{kl}(\mathbf{e},\pg)|\geq\epsilon\right) \\ \leq& 2\max\left\{\left(\frac{12C_M}{\epsilon\sqrt{n}}\right)^p,1 \right\} K^{n+2} \exp(-C_2 \epsilon^2 N_n).
\end{aligned}
\end{equation}

Let $X^{(1)}(\mathbf{e})=\frac{O(\mathbf{e})-\mathbb{E}[O(\mathbf{e})|Z]}{2N_n\rho_n}$ and $X^{(2)}(\mathbf{e})=X(\mathbf{e})-X^{(1)}(\mathbf{e})$, and we establish bound for $X^{(1)}(\mathbf{e})$ and $X^{(2)}(\mathbf{e})$ respectively.
By Lemma~\ref{lem:bernin}, for any $k,l\in[K]$, let $M=\chi\|\bar{B}\|_{\max}\chi|s_{\mathbf{e}}(k,l)|\rho_n$, $L=\bar{L}$, and $x=2N_n\rho_n\epsilon$, then for $\epsilon<\chi\|\bar{B}\|_{\max}\chi/\bar{L}$,
\begin{align*}
&\text{Pr}(X^{(1)}_{kl}(\mathbf{e})\geq\epsilon)\leq\exp\left(-\frac{4N_n^2\rho_n^2\epsilon^2}{2(\chi\|\bar{B}\|_{\max}\chi|s_{\mathbf{e}}(k,l)|\rho_n+2N_n\rho_n\epsilon\bar{L})}\right)\\
	\leq&\exp\left(-\frac{N_n\rho_n\epsilon^2}{\chi\|\bar{B}\|_{\max}\chi+\epsilon\bar{L}}\right)\leq\exp\left(-\frac{\epsilon^2N_n\rho_n}{2\chi\|\bar{B}\|_{\max}\chi}\right).
\end{align*}

Notice that $|X_{kl}^{(2)}(\mathbf{e})|/\|\bar{B}\|_{\max}\leq|Y_{kl}(\mathbf{e},\boldsymbol{\gamma}^0)|$. Thus, for $\epsilon<\chi \kappa_2^2\|\bar{B}\|_{\max}$,
\begin{align*}
	\text{Pr}(|X^{(2)}_{kl}(\mathbf{e})|\geq\epsilon)\leq&\text{Pr}\left(|Y_{kl}(\mathbf{e},\boldsymbol{\gamma}^0)|\geq\frac{\epsilon}{\|\bar{B}\|_{\max}}\right)\\
	\leq&2\exp\left(-\frac{6}{5\kappa_2^2\chi^2\|\bar{B}\|_{\max}}\epsilon^2N_n\right).
\end{align*}
Thus, the bound of $X(\mathbf{e})$ will be dominated by $X^{(1)}(\mathbf{e})$, and we will ignore the second term in the bound because it is just a small order and can be absorbed into the first one.

Similar to the arguments in \cite{zhao2012consistency}, for $|\mathbf{e}-\mathbf{c}|\leq m$, we have $Var[E_{kl}(\be,\hat\pg) - E_{kl}(\bc,\hat\pg)|\bc,\hat\pg] \leq 5mn\chi^2$ and $Var[O_{kl}(\be) - O_{kl}(\bc)|\bc] \leq 4mn\chi\|\bar B\|_{\max} \rho_n$;
then it follows that for $\epsilon<\frac{2\chi m}{n}$,
\begin{align*}
	\text{Pr}(|Y_{kl}(\mathbf{e},\hat{\boldsymbol{\gamma}})-Y_{kl}(\mathbf{c},\hat{\boldsymbol{\gamma}})|\geq\epsilon)\leq&~2\exp\left(-\frac{N_n\epsilon^2/2}{5\chi^2mn/N_n+2\chi\epsilon/3}\right)\\
  \leq&~2\exp\left(-\frac{n}{22\chi^2m}\epsilon^2N_n\right).
\end{align*}
By a similar argument as in \eqref{eq::Y_sup_be_gamhat}, we bypass the randomness in $\hat\pg$ and $\be$ via taking supremums, and obtain
\begin{equation}
\label{eq:supover_e_pghat}
  \text{Pr}(\max_{|\be-\bc|\leq m} |Y_{kl}(\mathbf{e},\hat{\boldsymbol{\gamma}})-Y_{kl}(\mathbf{c},\hat{\boldsymbol{\gamma}})|\geq\epsilon)
  \leq 2 \max\left\{\left(\frac{12C_M}{\epsilon\sqrt{n}}\right)^p,1 \right\} {n\choose m} K^{m+2}  \exp\left(-\frac{C_5n}{m}\epsilon^2N_n\right).
\end{equation}
For $\epsilon\geq\frac{2\chi m}{n}$,
\begin{align*}
  \text{Pr}(|Y_{kl}(\mathbf{e},\hat{\boldsymbol{\gamma}})-Y_{kl}(\mathbf{c},\hat{\boldsymbol{\gamma}})|\geq\epsilon)\leq&~2\exp\left(-\frac{N_n\epsilon^2/2}{5\chi^2mn/N_n+2\chi\epsilon/3}\right)\\
  \leq&~2\exp\left(-\frac{1}{2\chi}\epsilon N_n\right).
\end{align*}


Also, for $\epsilon<\frac{2\chi\|\bar{B}\|_{\max}m}{n\bar{L}}$,
\begin{align*}
&\text{Pr}(|X^{(1)}_{kl}(\mathbf{e})-X^{(1)}_{kl}(\mathbf{c})|\geq\epsilon)\leq\exp\left(-\frac{N_n\rho_n\epsilon^2}{\chi\|\bar{B}\|_{\max}mn/N_n+\epsilon\bar{L}}\right)\\
	\leq&\exp\left(-\frac{n-1}{2\chi\|\bar{B}\|_{\max}m}\epsilon^2N_n\rho_n\right)\leq\exp\left(-\frac{n}{4\chi\|\bar{B}\|_{\max}m}\epsilon^2N_n\rho_n\right).
\end{align*}
For $\epsilon\geq\frac{2\chi\|\bar{B}\|_{\max}m}{n\bar{L}}$,
\begin{align*}
\text{Pr}(|X^{(1)}_{kl}(\mathbf{e})-X^{(1)}_{kl}(\mathbf{c})|\geq\epsilon)\leq&\exp\left(-\frac{N_n\rho_n\epsilon^2}{\chi\|\bar{B}\|_{\max}mn/N_n+\epsilon\bar{L}}\right)\\
	\leq&\exp\left(-\frac{1}{3\bar{L}}\epsilon N_n\rho_n\right).
\end{align*}
We omit the bound for $|X^{(2)}_{kl}(\mathbf{e})-X^{(2)}_{kl}(\mathbf{c})|$ since it's a smaller order. By similar arguments as in \eqref{eq:supover_e_pghat}, we take supremum over $|\be-\bc|\leq m$ and $\hat\pg\in B(\pg^0,C_M/\sqrt{n}) $, and arrive at the stated results in the lemma. 
\end{proof}

\subsection{Proof of Theorem~\ref{THM:MLE_fixK}}
\label{sec:supple:proof_mle}
\subsubsection{Consistency of a General Class of Criteria}
Instead of directly analyzing $\ell_{\hat{\boldsymbol{\gamma}}}(\mathbf{e})$, similar to \cite{zhao2012consistency}, we first investigate the maximizer of a general class of criteria defined as 
    \begin{align}\label{eq:cri}
        Q(\mathbf{e},\hat{\boldsymbol{\gamma}}):= F\left(\frac{O(\mathbf{e})}{2N_n\rho_n},\frac{E(\mathbf{e},\hat{\boldsymbol{\gamma}})}{2N_n}\right),
    \end{align}
    where $O(\mathbf{e})=[O_{kl}(\mathbf{e}), k,l\in[K]]$ and $E(\mathbf{e},\hat{\boldsymbol{\gamma}})=[E_{kl}(\mathbf{e},\hat{\boldsymbol{\gamma}}), k,l\in[K]]$.
Then, we show our log-likelihood function falls in this class of criteria, implying the consistency of label estimation. We say the criterion $Q$ is \emph{consistent} if the estimated labels, obtained by maximizing the criterion, $\hat{\mathbf{c}}=\arg\max_{\mathbf{e}}Q(\mathbf{e},\hat{\boldsymbol{\gamma}})$ is \emph{consistent}.

One key condition of $Q$ for implying consistent community detection is that it reaches the maximum at $\boldsymbol{c}$ under the true parameter $\boldsymbol{\gamma}^0$ in the ``population version", which is $F\left(\frac{\mathbb{E}[O(\mathbf{c})]}{2N_n\rho_n},\frac{\mathbb{E}[E(\mathbf{c},\boldsymbol{\gamma}^0)]}{2N_n}\right)$. To further demonstrate what the ``population version" is, we introduce some notations. Given a community assignment $\mathbf{e}\in[K]^n$, we define $R(\mathbf{e})\in\mathbb{R}^{K\times K}$ with its elements being $R_{ka}(\mathbf{e})=\frac{1}{n}\sum_{i=1}^n\mathbbm{1}(e_i=k, c_i=a)$. One can view $R$ as the empirical joint distribution of $\mathbf{e}$ and $\mathbf{c}$.  Next, we introduce the key condition for the function $F$ in terms of $R$ as follows.
\begin{cond}\label{cond:max}
    $F(R\bar{B}R^T,RJR^T)$ is uniquely \footnote{The uniqueness is interpreted up to a permutation of the labels.} maximized over $\mathcal{R}=\{R:R\geq0,R^T\mathbf{1}=\boldsymbol{\pi}\}$ by $R=D(\boldsymbol{\pi})$, where $J$ is the matrix of ones and $D(\boldsymbol{\pi})$ is the diagonal matrix with diagonal entries $\boldsymbol{\pi}$.
\end{cond}

Besides the common factor $\mathbb{E}[\exp(\mathbf{z}_{ij}^T\boldsymbol{\gamma})^0]$, the first term is $\bar{B}$ weighted by pairwise community proportions, the second term is the normalized pairwise count between two communities. This reduces the criteria to the form described in \cite{zhao2012consistency}, thus similar methods can be applied to show the consistency of community detection. 
In addition, we need more regularity conditions for $F$, analogous to those in \cite{zhao2012consistency}.
\begin{cond}\label{cond:lip}
Some regularity conditions hold for $F$.
\begin{enumerate}
\item $F$ is Lipschitz in its arguments and $F(cX_0,cY_0)=cF(X_0,Y_0)$ for constant $c\neq0$.
\item The directional derivatives $\frac{\partial^2F}{\partial\epsilon^2}(X_0+\epsilon(X_1-X_0),Y_0+\epsilon(Y_1-Y_0))|_{\epsilon=0+}$ are continuous in $(X_1,Y_1)$ for all $(X_0,Y_0)$ that is in a neighborhood of $(D(\boldsymbol{\pi})\bar{B}D(\boldsymbol{\pi})^T,\boldsymbol{\pi}\boldsymbol{\pi}^T)$.
\item Let $G(R,\bar{B})=F(R\bar{B}R^T,RJR^T)$. On $\mathcal{R}$, for all $\boldsymbol{\pi}$, $\bar{B}$ and some constant $C>0$, the gradient satisfies $\frac{\partial G((1-\epsilon)D(\boldsymbol{\pi})+\epsilon R,\bar{B})}{\partial \epsilon}|_{\epsilon=0+}<-C$. 
\end{enumerate}
\end{cond}

Notice that the first condition in Condition~\ref{cond:lip} ensures that we could extract the common exponential factor. Thus we can ignore that term when we consider the population maximum in Condition~\ref{cond:max}. Naturally, the consistency of $\hat{\boldsymbol{\gamma}}$ is also required to ensure that the ``sample version'' is close to the ``population version''. Now the main theorem is stated as follows.

\begin{thm} \label{thm:gene}
    Under PCABM, if Conditions \ref{cond:zbd} and \ref{cond:zpd} hold for $Z$, then the criteria function $Q$ of the form (\ref{eq:cri}), which satisfies Conditions~\ref{cond:max},~\ref{cond:lip}, is weakly consistent if $\varphi_n\to\infty$ and strongly consistent if $\varphi_n/\log n\to\infty$.
\end{thm}

\begin{proof}

We divide the proof into three steps.

\emph{Step 1} : sample and population version comparison.
	We prove $\exists\ \epsilon_n\to0$, such that 
	\begin{align}\label{eq:sample}
	{\rm Pr}\left(\max_{\mathbf{e}}\left|F\left(\frac{O(\mathbf{e})}{2N_n\rho_n},\frac{E(\mathbf{e},\hat{\boldsymbol{\gamma}})}{2N_n}\right)-F(\theta(\boldsymbol{\gamma}^0)T(\mathbf{e}),\theta(\boldsymbol{\gamma}^0)S(\mathbf{e}))\right|\leq\epsilon_n\right)\to1,
	\end{align} 
	if $\varphi_n\to\infty$ and $\hat{\boldsymbol{\gamma}}\xrightarrow{p}\boldsymbol{\gamma}^0$.
	
	Since 
	\begin{align*}
	&\left|F\left(\frac{O(\mathbf{e})}{2N_n\rho_n},\frac{E(\mathbf{e},\hat{\boldsymbol{\gamma}})}{2N_n}\right)-\theta(\boldsymbol{\gamma}^0)F(T(\mathbf{e}),S(\mathbf{e}))\right|\\
	\leq&\left|F\left(\frac{O(\mathbf{e})}{2N_n\rho_n},\frac{E(\mathbf{e},\hat{\boldsymbol{\gamma}})}{2N_n}\right)-F(\theta(\boldsymbol{\gamma}^0)\hat{T}(\mathbf{e}),\theta(\hat{\boldsymbol{\gamma}})\hat{S}(\mathbf{e}))\right|\\
	&+\left|F(\theta(\boldsymbol{\gamma}^0)\hat{T}(\mathbf{e}),\theta(\hat{\boldsymbol{\gamma}})\hat{S}(\mathbf{e}))-\theta(\boldsymbol{\gamma}^0)F(\hat{T}(\mathbf{e}),\hat{S}(\mathbf{e}))\right|\\
	&+\theta(\boldsymbol{\gamma}^0)\left|F(\hat{T}(\mathbf{e}),\hat{S}(\mathbf{e}))-F(T(\mathbf{e}),S(\mathbf{e}))\right|,
	\end{align*}
	it is sufficient to bound these three terms uniformly. By Lipschitz continuity, 
	\begin{align}
	\begin{split}
	&\left|F\left(\frac{O(\mathbf{e})}{2N_n\rho_n},\frac{E(\mathbf{e},\hat{\boldsymbol{\gamma}})}{2N_n}\right)-\theta(\boldsymbol{\gamma}^0)F\left(\hat{T}(\mathbf{e}),\hat{S}(\mathbf{e})\right)\right|\\
	\leq&M_1\|X(\mathbf{e})\|_\infty+M_2\|Y(\mathbf{e},\hat{\boldsymbol{\gamma}})\|_\infty,\label{eq:sm1}\\
	\end{split}
	\end{align}
	\begin{align}
	\begin{split}
	&\left|F(\theta(\boldsymbol{\gamma}^0)\hat{T}(\mathbf{e}),\theta(\hat{\boldsymbol{\gamma}})\hat{S}(\mathbf{e}))-\theta(\boldsymbol{\gamma}^0)F(\hat{T}(\mathbf{e}),\hat{S}(\mathbf{e}))\right|\\
	\leq&M_2|\theta(\hat{\boldsymbol{\gamma}})-\theta(\boldsymbol{\gamma}^0)| \|\hat{S}(\mathbf{e})\|_\infty.\label{eq:sm2}
	\end{split}
	\end{align}
	By (\ref{eq:lm1}) and (\ref{eq:lm2}), (\ref{eq:sm1}) converges to $0$ uniformly if $\varphi_n\to\infty$. Since $\|\hat{S}(\mathbf{e})\|_\infty$ is uniformly bounded by $1$, (\ref{eq:sm2}) also converges to $0$ uniformly. 
	\begin{align}\label{eq:sm3}
	\begin{split}
	&\left|F\left(\hat{T}(\mathbf{e}),\hat{S}(\mathbf{e})\right)-F\left(T(\mathbf{e}),S(\mathbf{e})\right)\right|\\
	\leq& M_1\|\hat{T}(\mathbf{e})-T(\mathbf{e})\|_\infty+ M_2\|\hat{S}(\mathbf{e})-S(\mathbf{e})\|_\infty.
	\end{split}	
	\end{align}
	Since $\boldsymbol{\pi}(\mathbf{c})\xrightarrow{p}\boldsymbol{\pi}_0$, (\ref{eq:sm3}) converges to $0$ uniformly. So we prove (\ref{eq:sample}).
	
\emph{Step 2} : proof of weak consistency.
	We prove that there exists $\delta_n\to0$, such that
	\begin{align}\label{eq:weak}
		{\rm Pr}\left(\max_{\{\mathbf{e}:\|V(\mathbf{e})-I_K\|_1\geq\delta_n\}}F\left(\frac{O(\mathbf{e})}{2N_n\rho_n},\frac{E(\mathbf{e},\hat{\boldsymbol{\gamma}})}{2N_n}\right)<F\left(\frac{O(\mathbf{c})}{2N_n\rho_n},\frac{E(\mathbf{c},\hat{\boldsymbol{\gamma}})}{2N_n}\right)\right)\to1.
	\end{align}
	
	By continuity property of $F$ and Condition~\ref{cond:max}, there exists $\delta_n\to0$, such that 
	\begin{align*}
	\theta(\boldsymbol{\gamma}^0)F(T(\mathbf{c}),S(\mathbf{c}))-\theta(\boldsymbol{\gamma}^0)F(T(\mathbf{e}),S(\mathbf{e}))>2\epsilon_n
	\end{align*}
	if $\|V(\mathbf{e})-I_K\|_1\geq\delta_n$, where $I_K=V(\mathbf{c})$. Thus, following (\ref{eq:sample}),	
	\begin{alignat*}{3}
	&\text{Pr}\Bigg( &&\max_{\{\mathbf{e}:\|V(\mathbf{e})-I_K\|_1\geq\delta_n\}}F\left(\frac{O(\mathbf{e})}{2N_n\rho_n},\frac{E(\mathbf{e},\hat{\boldsymbol{\gamma}})}{2N_n}\right)<F\left(\frac{O(\mathbf{c})}{2N_n\rho_n},\frac{E(\mathbf{c},\hat{\boldsymbol{\gamma}})}{2N_n}\right)\Bigg)\\
	\geq &\text{Pr}\Bigg(&&\Bigg|\max_{\mathbf{e}:\|V(\mathbf{e})-I_K\|_1\geq\delta_n}\theta(\boldsymbol{\gamma}^0)F(T(\mathbf{e}),S(\mathbf{e}))\\
	& &&-\max_{\mathbf{e}:\|V(\mathbf{e})-I_K\|_1\geq\delta_n}F\left(\frac{O(\mathbf{e})}{2N_n\rho_n},\frac{E(\mathbf{e},\hat{\boldsymbol{\gamma}})}{2N_n}\right)\Bigg|\leq\epsilon_n,\\
	& &&\Bigg|\theta(\boldsymbol{\gamma}^0)F(T(\mathbf{c}),S(\mathbf{c}))-F\left(\frac{O(\mathbf{c})}{2N_n\rho_n},\frac{E(\mathbf{c},\hat{\boldsymbol{\gamma}})}{2N_n}\right)\Bigg|\leq\epsilon_n\Bigg)\to1.
	\end{alignat*}
	
	(\ref{eq:weak}) implies $\text{Pr}(\|V(\mathbf{e})-I_K\|<\delta_n)\to1$. Since 
	\begin{align*}
	&\frac{1}{n}|\mathbf{e}-\mathbf{c}|=\frac{1}{n}\sum_{i=1}^n\mathbbm{1}(c_i\neq e_i)=\sum_k\pi_k(1-V_{kk}(\mathbf{e}))\\
	\leq&\sum_k(1-V_{kk}(\mathbf{e}))=\|V(\mathbf{e})-I_K\|_1/2,
	\end{align*}
	weak consistency follows.

\emph{Step 3} : proof of strong consistency.
	
	To prove strong consistency, we need to show
	\begin{align}\label{eq:strong}
		{\rm Pr}\left(\max_{\{\mathbf{e}:0<\|V(\mathbf{e})-I_K\|_1<\delta_n\}}F\left(\frac{O(\mathbf{e})}{2N_n\rho_n},\frac{E(\mathbf{e},\hat{\boldsymbol{\gamma}})}{2N_n}\right)<F\left(\frac{O(\mathbf{c})}{2N_n\rho_n},\frac{E(\mathbf{c},\hat{\boldsymbol{\gamma}})}{2N_n}\right)\right)\to1.
	\end{align}

	Combining (\ref{eq:weak}) and (\ref{eq:strong}), we have 
	$${\rm Pr}\left(\max_{\{\mathbf{e}:\mathbf{e}\neq\mathbf{c}\}}F\left(\frac{O(\mathbf{e})}{2N_n\rho_n},\frac{E(\mathbf{e},\hat{\boldsymbol{\gamma}})}{2N_n}\right)<F\left(\frac{O(\mathbf{c})}{2N_n\rho_n},\frac{E(\mathbf{c},\hat{\boldsymbol{\gamma}})}{2N_n}\right)\right)\to1,$$
	which implies strong consistency.

	By Lipschitz continuity and the continuity of derivative of $F$ w.r.t. $V(\mathbf{e})$ in the neighborhood of $I_K$, we have
	\begin{align}\label{eq:st1}
	\begin{split}
	&F\left(\frac{O(\mathbf{e})}{2N_n\rho_n},\frac{E(\mathbf{e},\hat{\boldsymbol{\gamma}})}{2N_n}\right)-F\left(\frac{O(\mathbf{c})}{2N_n\rho_n},\frac{E(\mathbf{c},\hat{\boldsymbol{\gamma}})}{2N_n}\right)\\
	=&\theta(\boldsymbol{\gamma}^0)F(\hat{T}(\mathbf{e}),\hat{S}(\mathbf{e}))-\theta(\boldsymbol{\gamma}^0)F(\hat{T}(\mathbf{c}),\hat{S}(\mathbf{c}))+\Delta(\mathbf{e},\mathbf{c}),
	\end{split}
	\end{align}
	where $|\Delta(\mathbf{e},\mathbf{c})|\leq M_3(\|X(\mathbf{e})-X(\mathbf{c})\|_\infty)+M_4\|Y(\mathbf{e},\hat{\boldsymbol{\gamma}})-Y(\mathbf{c},\hat{\boldsymbol{\gamma}})\|_\infty)$, and 
	$$F(T(\mathbf{e}),S(\mathbf{e}))-F(T(\mathbf{c}),S(\mathbf{c}))\leq-\bar{C}\|V(\mathbf{e})-I_K\|_1+o(\|V(\mathbf{e})-I_K\|_1).$$
	
	Since the derivative of $F$ is continuous w.r.t. $V(\mathbf{e})$ in the neighborhood of $I_K$, there exists a $\delta'$ such that, 
	\begin{align}\label{eq:st2}
	F(\hat{T}(\mathbf{e}),\hat{S}(\mathbf{e}))-F(\hat{T}(\mathbf{c}),\hat{S}(\mathbf{c}))\leq-(C'/2)\|V(\mathbf{e})-I_K\|_1+o(\|V(\mathbf{e})-I_K\|_1)
	\end{align}
	holds when $\|\boldsymbol{\pi}(\mathbf{c})-\boldsymbol{\pi}_0\|_\infty\leq\delta'$. Since $\boldsymbol{\pi}(\mathbf{c})\to\boldsymbol{\pi}_0$, (\ref{eq:st2}) holds with probability approaching $1$. Combining (\ref{eq:st1}) and (\ref{eq:st2}), it is easy to see strong consistency follows if we can show
	$$\text{Pr}(\max_{\{\mathbf{e}\neq\mathbf{c}\}}|\Delta(\mathbf{e},\mathbf{c})|\leq C'\|V(\mathbf{e})-I_K\|_1/4)\to1.$$
	Note $\frac{1}{n}|\mathbf{e}-\mathbf{c}|\leq\frac{1}{2}\|V(\mathbf{e})-I_K\|_1$. So for each $m\geq1$,
	\begin{align}\label{eq:st3}
	\begin{split}
		&\text{Pr}\left(\max_{|\mathbf{e}-\mathbf{c}|=m}|\Delta(\mathbf{e},\mathbf{c})|>C'\|V(\mathbf{e}-I_K)\|_1/4\right)\\
		\leq&\text{Pr}\left(\max_{|\mathbf{e}-\mathbf{c}|\leq m}\|X(\mathbf{e})-X(\mathbf{c})\|_\infty>\frac{C'm}{4M_3n}\right)(\equiv I_1)\\
		&+\text{Pr}\left(\max_{|\mathbf{e}-\mathbf{c}|\leq m}\|Y(\mathbf{e},\hat{\boldsymbol{\gamma}})-Y(\mathbf{c},\hat{\boldsymbol{\gamma}})\|_\infty>\frac{C'm}{4M_4n}\right)(\equiv I_2).
	\end{split}
	\end{align}

	Let $\eta_1=C'/4M_3$, if $\eta_1<\eta$, by (\ref{eq:lm31}),
	\begin{align*}
		I_1\leq&2K^{m+2}n^m\exp(-\eta_1^2\frac{C_3m}{n}N_n\rho_n)=2K^2[K\exp(\log n-\eta_1^2C_3N_n\rho_n/n)]^m.
	\end{align*}
	If $\eta_1>\eta$, by (\ref{eq:lm32}),
	\begin{align*}
		I_1\leq&2K^{m+2}n^m\exp(-\eta_1\frac{C_4m}{n}N_n\rho_n)=2K^2[K\exp(\log n-\eta_1C_4N_n\rho_n/n)]^m.
	\end{align*}
	
	Similar arguments hold for $I_2$ by using (\ref{eq:lm41}) and (\ref{eq:lm42}). In all cases, since $\varphi_n/\log n\to\infty$, 
	\begin{align*}
	&\text{Pr}(\max_{\{\mathbf{e}\neq\mathbf{c}\}}|\Delta(\mathbf{e},\mathbf{c})|>C'\|V(\mathbf{e})-I_K\|_1/4)\\
	=&\sum_{m=1}^{\infty}\text{Pr}(\max_{|\mathbf{e}-\mathbf{c}|=m}|\Delta(\mathbf{e},\mathbf{c})|>C'\|V(\mathbf{e})-I_K\|_1/4)\to0.
	\end{align*}
	as $n\to\infty$. The proof is completed.

\end{proof}

\subsubsection{Proof of Theorem~\ref{THM:MLE_fixK}}\label{D}
By Theorem \ref{thm:gene}, it suffices to show that the log-likelihood satisfies the above conditions~\ref{cond:max} and \ref{cond:lip}.
By scaling $\ell_{\hat{\boldsymbol{\gamma}}}(\mathbf{e})$, we have
\begin{align*}
&\frac{1}{N_n}\ell_{\hat{\boldsymbol{\gamma}}}(\mathbf{e})=\rho_n F\left(\frac{O}{2N_n\rho_n},\frac{E}{2N_n}\right)+(\rho_n\log\rho_n)\sum_{kl}\frac{O_{kl}(\mathbf{e})}{2N_n\rho_n}+O(n^{-1}),
\end{align*}
where $F(X,Y)=\sum_{kl}X_{kl}\log\left(\frac{X_{kl}}{Y_{kl}}\right)$, $X,Y\in\mathbb{R}^{K\times K}$. Note that $F$ is closely related to the likelihood criterion used in \cite{zhao2012consistency}. In our case, 
\begin{align*}
&F(\theta(\boldsymbol{\gamma}^0)R\bar{B}R^T,\theta(\boldsymbol{\gamma}^0)\boldsymbol{\pi}\boldsymbol{\pi}^T)-F(\theta(\boldsymbol{\gamma}^0)R\bar{B}R^T,\theta(\hat{\boldsymbol{\gamma}})\boldsymbol{\pi}\boldsymbol{\pi}^T)\\
=&\theta(\boldsymbol{\gamma}^0)\log\frac{\theta(\boldsymbol{\gamma}^0)}{\theta(\hat{\boldsymbol{\gamma}})}(\mathbf{1}^TR\bar{B}R^T\mathbf{1}),
\end{align*}
which is basically the population degree up to a constant. This fact shows that the consistency of $\hat{\boldsymbol{\gamma}}$ is unnecessary to ensure the consistency of community detection. Thus, we can plug in any random fixed $\hat{\boldsymbol{\gamma}}$, which is different from our general theorem. For simplicity, just assume we use the true value $\boldsymbol{\gamma}^0$ here. Observe that
\begin{align*}
F(\theta(\boldsymbol{\gamma}^0)R\bar{B}R^T,\theta(\boldsymbol{\gamma}^0)\boldsymbol{\pi}\boldsymbol{\pi}^T)=\theta(\boldsymbol{\gamma}^0)F(R\bar{B}R^T,\boldsymbol{\pi}\boldsymbol{\pi}^T),
\end{align*}
then the form of $F$ is exactly the same as $F$ defined in \cite{bickel2009nonparametric}, which automatically satisfies all conditions for $F$.

\subsection{Consistency of Maximum Likelihood Community Detection When Number of Communities $K$ Grows with $n$}
\label{append_sect:KtoinftyMLE}
In this section we consider the MLE for cluster assignment $\bc$ when the number of communities $K$ grows with $n$.
We start with some notations and definitions.
In our model $A_{ij} | Z \stackrel{indpt}\sim \text{Poisson}(B_{e_ie_j}\exp({\bz_{ij}^\top \pg}))$ we denote the true value of parameters $\be,B$ and $\pg$ by $\bc,B^0$ and $\pg^0$, respectively. Let $P_{ij} := B^0_{c_ic_j}$. Recall the log-likelihood of our model is
\begin{equation}
\label{logmathcalL}
\log \mathcal L(\be,\pg,B,\pmb\pi|A,Z) \propto \sum_{i=1}^n \pi_i + \sum_{i<j} A_{ij} \log B_{e_ie_j} + \sum_{i<j} A_{ij} \bz_{ij}^\top \pg - \sum_{i<j} B_{e_ie_j} \exp(\bz_{ij}^\top \pg).
\end{equation}
Since we already have an estimate of $\pg$, the $\sum_{i<j} A_{ij} \bz_{ij}^\top \pg$ term in \eqref{logmathcalL} does not contribute to the estimation of $\be$. Besides, under the average degree $ = \omega(\text{Poly}(\log n))$ regime, the $\sum_{i=1}^n \pi_i$ term is smaller order of the other terms. Thus, for the MLE of $\be$ we are (asymptotically) actually optimizing the following loss function:
\begin{equation}
L(A,Z;\be,B,\pg) := \sum_{i<j} A_{ij} \log B_{e_ie_j} - \sum_{i<j} B_{e_ie_j} \exp(\bz_{ij}^\top \pg).
\end{equation}
We define a ``population'' version of the above loss as
\begin{equation}
L_p(Z;\be,B,\pg) := \sum_{i<j} P_{ij} \exp(\bz_{ij}^\top \pg) \log B_{e_ie_j} - \sum_{i<j} B_{e_ie_j} \exp(\bz_{ij}^\top \pg)
\end{equation} which is the expectation of $L(A,Z;\be,B,\pg)$ given $\be$ and $Z$.
$L(A,Z;\be,B,\pg)$ and $L_p(Z;\be,B,\pg)$ could be optimized with respect to $B$ with $\hat B^{(\be)}_{ab} := \frac{O_{ab}(\be)}{E_{ab}(\be)}$ and $\tilde B^{(\be)}_{ab} := \frac{\sum_{(i,j)\in s_{\be}(a,b)} P_{ij} \exp(\bz_{ij}^\top \pg)}{E_{ab}(\be)}$. Thus, profiling $B$ we define 
\begin{equation}
\begin{aligned}
L(A,Z;\be, \hat B^{(\be)}, \pg)  &=\frac12 \sum_{ab} [O_{ab} \log(\hat B_{ab}) - E_{ab} \hat B_{ab}] =: PL(A,Z;\be,\pg) ,\\
L_p(Z;\be, \tilde B^{(\be)}, \pg)&= \frac12    \sum_{ab} [E_{ab} \tilde B_{ab} \log(\tilde B_{ab}) - E_{ab} \tilde B_{ab}] =: PL_p(Z;\be,\pg) ,
\end{aligned}
\end{equation}
where for simplicity we omit the $(\be)$ for $\hat B$ and $\tilde B$ when there is no confusion. 

The goal is to prove consistency of the MLE of cluster assignment $\hat\be=\arg\max_{\be} PL(A,Z;\be,\pg^0)$. In fact, the main idea is to adopt the very classical approach of first showing a  ``uniform weak law of large numbers'' type result (Theorem \ref{Ktoinf_thm12}), and then establish some identifiability conditions such that the expected likelihood $PL_p(Z;\be,\pg^0)$ is large and close to $PL_p(Z;\bc,\pg^0)$ only if $\be$ is close enough to the true $\bc$ (Theorem \ref{Ktoinfty_thm3b}).

First we state some conditions we work on. We want to point that in this section (class assignment MLE when $K$ grows) we are not using the $B = \rho_n \bar B$ setting, and are imposing assumptions directly on $P_{ij}$'s. This is mainly because when $K\to\infty$ we need stronger signal to noise ratio, which can be approximately understood as in-class probability over between-class probability, to identify the communities (see the remark under Theorem \ref{Ktoinfty_thm3b} for more detailed discussions).

\begin{cond}\label{avgdeg_logn3} 
$M = \sum_{i<j} P_{ij}$ satisfies $M = \omega(n(\log n)^{3+\delta})$ for some positive constant $\delta$.
\end{cond}

\begin{cond}\label{K_rate} 
The number of communities $K$ satisfies $K = O(n^{1/2})$.
\end{cond}

\begin{cond}\label{bdd_B}
There exist some constants $0 < c_{B} < C_{B}$ such that
$c_{B} / (n^2) \leq P_{ij} \leq n^2 C_{B}$.
\end{cond}

Condition \ref{avgdeg_logn3} requires the average degree to grow in at least Poly$(\log n)$ rate, which is still a sparse network setting. 
Condition \ref{K_rate} allows the number of communities $K$ to grow at a rate as fast as $n^{1/2}$, which matches the growth rate allowed in SBM's MLE consistency \citep{choi2012stochastic}. Condition \ref{bdd_B} is a mild condition: 
note that roughly speaking $B_{c_ic_j}$ is of the order $\rho_n = \omega((\log n)^{3+\delta}/n)$, so the required range from $O(n^{-2})$ to $O(n^2)$ is very loose for $B$. The sparsity of the network is already controlled by Condition \ref{avgdeg_logn3}; and Condition \ref{bdd_B} is not a sparsity condition, but just some technical requirement so that $\log B_{c_ic_j}$ does not blow up too much.

Now we present our main results.
\begin{thm}\label{Ktoinf_thm12}
Under Conditions \ref{cond:zbd},  \ref{avgdeg_logn3}, \ref{K_rate} and \ref{bdd_B},
\begin{equation}
\label{KinfMLE_lem234}
\max_{\be} \{|PL(A,Z;\be,\pg^0) - PL_p(Z;\be,\pg^0)| \} = o_p(M),
\end{equation}
where $M=\sum_{i<j} P_{ij}$.{}
\end{thm}

Theorem \ref{Ktoinf_thm12} states a uniform concentration result of the profile likelihood around its population version, which plays the role of ``uniform weak law of large numbers'' in the classical MLE consistency proof, where $M$ is the actual ``sample size''.
Our next step is to show $\hat\be$ is close enough to $\bc$ in their expected likelihood.

\begin{thm} \label{Ktoinfty_thm3a}
Let $\hat\be = \arg\max_{\be} PL(A,Z;\be,\hat\pg)$ to be the MLE for the true communnity assignment $\bc$. 
Then under Conditions \ref{cond:zbd}, \ref{cond:e1}, \ref{avgdeg_logn3}, \ref{K_rate} and \ref{bdd_B},
we have
\begin{equation}
\label{eq_Kinf_thm3a}
PL_p(Z; \bc, \pg^0) - PL_p(Z; \hat\be, \pg^0) = o_p(M).
\end{equation}
\end{thm}

Theorem \ref{Ktoinfty_thm3a} already shows that $\hat\be$ is close to the true $\bc$ in some sense. With some identifiability conditions, we will be able to translate this closeness into consistency. In particular, this consistency is defined in terms of a notion of classification error used in \cite{choi2012stochastic}: $N(\hat\be)$, the number of incorrect class assignment under $\hat\be$, is counted for every node whose
true class under $\bc$ is not in the majority within its estimated class $\hat\be$; and $\hat\be$ achieves weak consistency when $N(\hat\be) = o_p(n)$.
\begin{thm} \label{Ktoinfty_thm3b}
Suppose the following two conditions hold for $\forall a,b,c\in [K]$:
\begin{enumerate}
    \item  $\min_{a}(n_a(\bc)) = \Omega(n/K)$, i.e. all cluster sizes are of the same scale;
    \item $\min_{a\neq b} \max_{c} D'(B_{ac}^0,B^0_{bc}) = \Omega\left(\frac{MK}{n^2}\right)$ where $D'(a,b) := a\log a + b\log b - (a+b) \log \frac{(a+b)}2 $.
\end{enumerate}
Then \eqref{eq_Kinf_thm3a} implies $N(\hat\be) = o_p(n)$, where $\hat\be$ denotes the MLE.
\end{thm}
The second condition is basically saying any two classes $a$ and $b$ are ``well-separated'' in the sense that there exists a class $c$ that connects with $a$ and $b$ very differently so that one can distinguish $a$ and $b$ through their connections with $c$. 
This condition may not be trivially satisfied. A simple calculation could show that the left hand side of the  condition is of order $\min_a\max_c B^0_{ac}$, while the right hand side is $\Omega(\text{Poly}(\log n)K/n)$. When $K$ is fixed, this is satisfied in the usual $\rho_n =  O(\text{Poly}(\log n)/n)$ scenario. When $K$ is growing at a rate no faster than $O(\sqrt n) $,
an example scenario for this to hold would be $B^0_{aa} = O(\text{Poly}(\log n)K/n)$ and $B^0_{ab} = O(\text{Poly}(\log n)/ n)$ for $a\neq b$. Note that in the second scenario, though we have a stronger requirement for in-class edge probability, the average degree of each node is still of the order $O(\text{Poly}(\log n))$, which is not beyond the sparse setting satisfying Condition \ref{avgdeg_logn3}.

The proofs of our main theorems are basically divided into several lemmas.

\begin{lemma} For any $\be,B$,
\label{KinfMLE_lem0}
\begin{align*}
L_p(Z;\bc,B^0,\pg^0) - L_p(Z;\be,B,\pg^0) = \sum_{i<j} \ezg D(P_{ij} || B_{e_ie_j}) \geq 0,
\end{align*}
where $D(\lambda||\mu) := \lambda \log(\lambda / \mu) - \lambda + \mu \geq 0$ is the KL divergence from a Poisson($\lambda$) distribution to a Poisson($\mu$) one.
\end{lemma}

Lemma \ref{KinfMLE_lem0} basically says the population version of likelihood $L_p(Z;\cdot,\cdot, \pg^0)$ achieves its maximum at true $(\be,B^0)$. 
The next few lemmas establish concentration results of the profile likelihood around its population version. 

\begin{lemma} For any $\be$, \label{KinfMLE_lem1}
\begin{align*}
PL(A,Z;\be,\pg^0) - PL_p(Z;\be,\pg^0) = \frac12 \sum_{ab} \left[
E_{ab} D(\hat B_{ab}||\tilde B_{ab}) + E_{ab} (\hat B_{ab} - \tilde B_{ab}) \log \tilde B_{ab}
\right]
\end{align*}
where $E_{ab} = E_{ab}(\be,\pg^0)$.
\end{lemma}

\begin{lemma} \label{KinfMLE_thm1} Under Conditions \ref{cond:zbd} and \ref{avgdeg_logn3}, \ref{K_rate}, 
\begin{align*}
\max_{\be} \left\{ \sum_{a\leq b} E_{ab} D(\hat B_{ab} || \tilde B_{ab}) \right\} = o_p(M)
\end{align*}
where $E_{ab} = E_{ab}(\be,\pg^0)$.
\end{lemma}
\begin{proof}[Proof of Lemma \ref{KinfMLE_thm1}]
Given covariates $Z$ and class assignment $\be$, let $\hat\Theta_{\epsilon} := \{\hat B : \sum_{a\leq b} E_{ab} D(\hat B_{ab} || \tilde B_{ab}) \geq \epsilon \} $ and $\hat\Theta_1 := \{ \hat B : \exists a,b\in[K] \ s.t. \ \hat B_{ab} \geq 2M\xi/ E_{ab}  \}$.

We first bound the probability of event $\hat\Theta_1$.
Under the conditions $||\bz_{ij}||_{\infty} \leq \zeta$ and $\xi = \exp(\zeta||\pg^0\|_1)$, by Bernstein's inequality for Poisson variables (Lemma A.1 and A.2), 
\begin{align*}
\Pr\left(\hat B_{ab} \geq \frac{2M\xi}{E_{ab}} \right) \leq& \Pr\left(\sum_{(i,j)\in S_{\be}(a,b)} A_{ij} \geq M\xi + \sum_{(i,j)\in S_{\be}(a,b)} P_{ij}\ezg \right) \\
\leq& \exp\left( -\frac{(M\xi)^2}{2(M\xi + M\xi \bar L ) } \right)\\
\leq& 
\exp(-C_1M)
\end{align*}
where $C_1, C_2,...$ are constants. Thus, 
\begin{align*}
\Pr(\hat\Theta_1) \leq \sum_{a\leq b} \Pr\left(\hat B_{ab} \geq \frac{2M\xi}{E_{ab}} \right) \leq K^2 \exp(-C_1M).
\end{align*}

Next we bound the probability of event $\hat\Theta_\epsilon$.
Applying a Poisson variable's Chernoff inequality (\cite{vershynin2018high}, p20) we have
\begin{align*}
&\text{Pr}(\hat B_{ab} \geq \tilde B_{ab} + t) = \Pr\left(O_{ab} \geq \sum_{(i,j)\in S_{\be}(a,b)} P_{ij} \ezg  + E_{ab} t \right) \\
\leq&  \left( \frac{e \sum_{(i,j)\in S_{\be}(a,b)} P_{ij} \ezg }{ \sum_{(i,j)\in S_{\be}(a,b)} P_{ij} \ezg + E_{ab} t} \right)^{\sum_{(i,j)\in S_{\be}(a,b)} P_{ij} \ezg+E_{ab}t} \exp\left(-\sum_{(i,j)\in S_{\be}(a,b)} P_{ij} \ezg\right)\\
=& \frac{e^{E_{ab}t} }{ (1+\frac{t}{\tilde B_{ab}})^{E_{ab}\tilde B_{ab} + E_{ab}t} }.
\end{align*}
Since $D(\tilde B_{ab}+t || \tilde B_{ab}) = (\tilde B_{ab}+t) \log(1+\frac{t}{\tilde B_{ab}}) -t$, we have
\begin{align*}
\exp(-E_{ab} D(\tilde B_{ab}+t || \tilde B_{ab}) ) \geq \text{Pr}(\hat B_{ab} \geq \tilde B_{ab} + t),
\end{align*}
which indicates
\begin{align*}
 \text{Pr}(\hat B_{ab} = v) \leq \exp(-E_{ab} D(v || \tilde B_{ab}) ).
\end{align*}
And by the independence between $\{A_{ij}\}_{i<j}$ we have
\begin{align*}
 \text{Pr}(\hat B) \leq \exp(- \sum_{a\leq b} E_{ab} D(\hat B_{ab} || \tilde B_{ab}) ).
\end{align*}
Thus, for any $\hat B \in \hat\Theta_\epsilon$, we have $\Pr(\hat B) \leq \exp(-\epsilon)$.
Now we bound $\Pr(\hat B \in \hat\Theta_\epsilon)$ by
\begin{align*}
\Pr(\hat\Theta_\epsilon) \leq& \Pr(\hat\Theta_1) + \Pr(\hat\Theta_\epsilon \setminus \hat\Theta_1)\\
\leq& K^2\exp(-C_1 M) +  |\hat\Theta_1^c| e^{-\epsilon}\\
\leq& K^2\exp(-C_1 M) +  (2M\xi+1)^{\frac{K^2+K}2 } e^{-\epsilon},
\end{align*}
where the bound on cardinality of complement of $\hat\Theta_1$ comes from the fact that $\hat B_{ab} = O_{ab} / E_{ab}$ and $O_{ab}$ only takes integer values. Therefore, for any $\epsilon'>0$, a union bound over all $[K]^n$ possible $\be$'s gives
\begin{equation}
\label{Ktoinfty_thm1_finalineq}
\begin{aligned}
&\Pr\left(\max_{\be} \left\{ \sum_{a\leq b} E_{ab} D(\hat B_{ab} || \tilde B_{ab}) \right\} > \epsilon' M \right) \\ 
\leq&  \exp\left((n+2)\log K -C_1 M\right) + \exp\left(n\log K + \frac{K^2+K}2 \log(2M\xi+1) - \epsilon' M \right).
\end{aligned}
\end{equation}
Under the conditions $M = \omega(n(\log n)^{3+\delta} )$ and $K = O(n^{1/2})$, the probability bound in \eqref{Ktoinfty_thm1_finalineq} goes to $0$ as $n\to\infty$, which proves the desired result.
\end{proof}

\begin{lemma} \label{KinfMLE_thm2} Under Conditions \ref{cond:zbd} and  \ref{avgdeg_logn3}, \ref{K_rate} and \ref{bdd_B},
\begin{align*}
\max_{\be} \left\{ \left| \sum_{a b} E_{ab} (\hat B_{ab} - \tilde b_{ab}) \log \tilde B_{ab} \right| \right\} = o_p(M)
\end{align*}
where $E_{ab} = E_{ab}(\be,\pg^0)$.
\end{lemma}
\begin{proof}[Proof of Lemma \ref{KinfMLE_thm2}]
Given $Z$ and $\be$, let $X_{ij} := A_{ij} \log \tilde B_{e_ie_j}$, then $\mathbb E X_{ij} = P_{ij} \ezg \log \tilde B_{e_ie_j}$. 
The term we are considering could be expressed as $\sum_{a b} E_{ab} (\hat B_{ab} - \tilde b_{ab}) \log \tilde B_{ab} = 2 \sum_{i<j} (X_{ij} -\mathbb E X_{ij}) =: 2(X-\mathbb E X)$.
$\{X_{ij}\}_{i<j}$ follow independent scaled Poisson distributions, so they satisfy the Bernstein condition in Lemma A.2 with $L = 2  \bar L \log n\geq \bar L \log |\tilde B_{e_ie_j}|$. Thus, by Bernstein's inequality, for any $\epsilon>0$
\begin{align*}
\Pr(|X-\mathbb E X|\geq \epsilon) \leq 2 \exp\left( -\frac{\epsilon^2}{2 (M \xi (2\log n)^2 + 2\epsilon \bar L \log n) } \right).
\end{align*}
And a union bound over all possible $\be$ gives
\begin{equation}
\label{Ktoinfty_thm2_finalineq}
\begin{aligned}
\Pr(\max_{\be} |X-\mathbb E X|\geq \epsilon M) \leq& 2 K^n \exp\left( -\frac{\epsilon^2 M^2}{4 (2M \xi (\log n)^2 + \epsilon M \bar L \log n) } \right)\\
\leq& 2 \exp\left(n\log K -  \frac{C_2 \epsilon^2M}{(\log n)^2} \right).
\end{aligned}
\end{equation}
Under the conditions $M = \omega(n(\log n)^{3+\delta} )$ and $K = O(n^{1/2})$, the probability bound in \eqref{Ktoinfty_thm2_finalineq} goes to $0$ as $n\to\infty$, which proves the desired result.
\end{proof}

Combining Lemma \ref{KinfMLE_lem1}, \ref{KinfMLE_thm1} and \ref{KinfMLE_thm2} together we immediately derive Theorem \ref{Ktoinf_thm12}.
Our next step is to show $\hat\be$ is close enough to $\bc$ in their expected likelihood (Theorem \ref{Ktoinfty_thm3a}).
First we show a lemma that bridges likelihood with true parameter $\pg^0$ to that with the MLE $\hat\pg$.


\begin{lemma} Assume MLE $\hat\pg$ is consistent. Then under
 Conditions \ref{cond:zbd} and \ref{avgdeg_logn3}, for any $\be$,
\label{KinfMLE_lemhat}
\begin{align*}
| PL(A,Z;\be,\hat\pg) - PL(A,Z;\be,\pg^0)| = o_p(M).
\end{align*}
\end{lemma}
\begin{proof}[Proof of Lemma \ref{KinfMLE_lemhat}] 
Since $PL(A,Z;\be,\pg) = \sum_{i<j} (A_{ij} \log \hat B_{e_ie_j} - A_{ij})$, we could see that 
\begin{equation}
\label{lemhat_eq1}
PL(A,Z;\be,\hat\pg) - PL(A,Z;\be,\pg^0) = \sum_{i<j} A_{ij} \log \frac{E_{e_ie_j}({\be,\pg^0})}{E_{e_ie_j}({\be,\hat\pg})}.
\end{equation} 
Under the assumption $||\bz_{ij}||_{\infty} \leq \zeta$, we have 
\begin{equation}
\label{lemhat_eq2}
\left|\log \frac{E_{e_ie_j}({\be,\pg^0})}{E_{e_ie_j}({\be,\hat\pg})} \right| \leq \zeta||\hat\pg-\pg^0||_1 = o_p(1).
\end{equation}
Since $\{A_{ij}\}_{i<j}$ are independent, by Bernstein's inequality for Poisson variables (Lemma A.1 and A.2), for any constant $a>0$
\begin{align*}
\Pr(\sum_{i<j} A_{ij} \geq (a+1) M) \leq \exp\left( -\frac{a^2M^2}{2(M  \xi + aM\bar L)} \right) \leq \exp(-C_3 aM),
\end{align*}
i.e. $\sum_{i<j} A_{ij} = O_p(M)$. Combining that result with \eqref{lemhat_eq1} and \eqref{lemhat_eq2} finishes the proof.
\end{proof}

With Theorem \ref{Ktoinf_thm12} and Lemma \ref{KinfMLE_lemhat} we are ready to show Theorem \ref{Ktoinfty_thm3a} characterizing the MLE $\hat\be$ being close to the true $\bc$ in their population version profile likelihood, and furthermore, Theorem \ref{Ktoinfty_thm3b} stating the consistency of $\hat\be$ to $\bc$.

\begin{proof}[Proof of Theorem \ref{Ktoinfty_thm3a}]
First note that by Lemma \ref{KinfMLE_lem0} $PL_p(Z; \bc, \pg^0) = L_p(Z;\bc,B^0,\pg^0)  \geq PL_p(Z; \hat\be, \pg^0)$. Hence it suffices to upper bound $PL_p(Z; \bc, \pg^0) - PL_p(Z; \hat\be, \pg^0)$. By the definition of $\hat\be$, Lemma \ref{KinfMLE_lemhat} and equation \eqref{KinfMLE_lem234}, 
\begin{align*}
&PL_p(Z; \bc, \pg^0) - PL_p(Z; \hat\be, \pg^0) \\
\leq& PL(A,Z;\bc,\hat\pg) + |PL(A,Z;\bc,\hat\pg)-PL(A,Z;\bc,\pg^0)| + |PL(A,Z;\bc,\pg^0) - PL_p(Z;\bc,\pg^0)| \\
&- PL(A,Z;\hat\be,\hat\pg) + |PL(A,Z;\hat\be,\hat\pg)-PL(A,Z;\hat\be,\pg^0)| + |PL(A,Z;\hat\be,\pg^0)-PL_p(Z;\hat\be,\pg^0)| \\
\leq& o_p(M).
\end{align*}
\end{proof}

\begin{proof}[Proof of Theorem \ref{Ktoinfty_thm3b}]
We first define a partition $\Pi$ of the edge set $\{(i,j)\}_{i<j}$ to be a collection of disjoint subsets $T_1^\Pi ,...,T_R^\Pi$ such that $\cup_{r=1}^R T_{r}^\Pi = \{(i,j)\}_{i<j}$. We denote by $\Pi_{ij} = T_r^\Pi$ if $(i,j) \in T_r^\Pi$.
A example is that node class assignments ${\Pi(\be)}$ naturally induces a partition of the edge set $\{T^{\Pi(\be)}_{kl}\}_{1\leq k\leq l\leq K}$, in which case $\Pi_{ij} = T_{(e_i,e_j)}^{\Pi(\be)}$. In general a partition could be more flexible than one induced by class assignments. For any partition $\Pi = \{T_1^\Pi,...,T_R^\Pi\}$ of $\{(i,j)\}_{1\leq i<j\leq n}$, define $\tilde P_r := (\sum_{(i,j)\in T_r^\Pi} P_{ij} \ezg)/(\sum_{(i,j)\in T_r^\Pi} \ezg)$ which corresponds to the $\tilde B$ defined previously; and define 
\begin{equation}
\label{PLPi}
PL_p^*(Z;\Pi) := \sum_{i<j} (P_{ij}\ezg \log \tilde P_{\Pi_{ij}} - \tilde P_{\Pi_{ij}} \ezg) = \sum_{i<j} P_{ij} \ezg (\log \tilde P_{\Pi_{ij}}-1)
\end{equation}
which corresponds to the $PL_p(Z;\be,\pg)$ previously (we omitted the argument $\pg$ in \eqref{PLPi} since we are only dealing with $\pg^0$ now by Theorem \ref{Ktoinfty_thm3a}). When the partition is induced by a class assignment $\be$, it is easy to see that $PL_p(Z;\be,\pg^0) = PL_p^*(Z;\Pi(\be))$. Besides, by noting that $PL_p^*(Z;\Pi)$ is the optimal value of the problem $\max_Q L_p^*(Z;\Pi,Q) := \sum_{i<j} (P_{ij}\ezg \log Q_{ij} - Q_{ij} \ezg)$ subject to the constraint $Q_{i_1j_1} = Q_{i_2j_2}$ if $\Pi_{i_1j_1}= \Pi_{i_2j_2}$, we have the following property:
\begin{lemma}
\label{xifen}
Let $\Pi'$ be a refinement of partition $\Pi$ of set $\{(i,j) \}_{i<j}$, then $PL_p^*(Z;\Pi)\leq PL_p^*(Z;\Pi')$.
\end{lemma}
Next we want to construct a refinement $\Pi^*$ of $\Pi(\hat\be)$ such that
\begin{equation}
\label{thm3blema2}
PL_p(Z;\bc,\pg^0) - PL_p^*(Z;\Pi^*) = N(\hat\be)\Omega(M/n).
\end{equation}
Combining \eqref{thm3blema2} with \eqref{eq_Kinf_thm3a} and Lemma \ref{xifen} we have $ N(\hat\be)\Omega(M/n)=PL_p(Z;\bc,\pg^0) - PL_p^*(Z;\Pi^*) \leq PL_p(Z;\bc,\pg^0) - PL_p(Z;\hat\be,\pg^0)=o_p(M)$. Hence it suffices to show there exists a refinement such that \eqref{thm3blema2} holds. We construct $\Pi^*$ as follows. In each class $k$ of $\hat\be$, we take out pairs $(i,j)$ such that $\hat e_i = \hat e_j = k$ but $c_i\neq c_j$. Continue this process in this cluster of $\hat\be_{\cdot}=k$ until all nodes remaining in it have the same true class membership under $\bc$. Denote the total number of pairs we have taken out by $N_1$. Since by definition we see each node whose true class under $\bc$ is in the majority within its estimated class $\hat\be$ as being correctly classified, the total number of nodes remaining after the pairing process must be smaller than or equal to the number of correctly classified nodes, and hence $2N_1\geq N(\hat\be)$. Next, for each picked out pair $(i,j)$, select all nodes $h$ such that $D'(P_{ih}, P_{jh}) \geq CMK/(n^2)$ where $C$ is the constant from the second condition of the theorem. In $\Pi(\hat\be)$ we know that $\Pi(\hat\be)_{ih} = \Pi(\hat\be)_{jh} = T^{\Pi(\hat\be)}_{(\hat e_i, \hat e_h)}$; now we separate $T^{\Pi(\hat\be)}_{(\hat e_i, \hat e_h)}$ into $\{(i,h), (j,h)\}$ and $T^{\Pi(\hat\be)}_{(\hat e_i, \hat e_h)}\setminus \{(i,h), (j,h)\}$. Perform this separation for all $((i,j),h)$ such that $(i,j)$ pair is picked out in the first step and $D'(P_{ih}, P_{jh}) \geq CMK/(n^2)$, and the resulted refinement of $\Pi(\hat\be)$ is the $\Pi^*$ we wanted. To see this, denote the number of triples $((i,j),h)$ selected in the second step by $N_2$. By condition 2 in the theorem, for each pair $(i,j)$ there is at least one true class such that all nodes in it could form a selected triple with $(i,j)$. Hence $N_2\geq N_1 \min_a n_a(\bc) \geq N_1 \Omega(n/K)$. From \eqref{PLPi} we could calculate 
\begin{align*}
PL_p(Z;\bc,\pg^0) - PL_p^*(Z;\Pi^*) =& \sum_{i<j} P_{ij}\ezg\log\frac{P_{ij}}{\tilde P_{\Pi_{ij}^*}} \\
\geq& \sum_{((i,j),h) \text{picked out}} P_{ih} \log\frac{P_{ih}}{\tilde P_{\Pi_{ih}^*}} + P_{jh} \log\frac{P_{jh}}{\tilde P_{\Pi_{jh}^*}} \\
=& \sum_{((i,j),h) \text{picked out}} D'(P_{ih}, P_{jh})\\
\geq& N_2 \Omega(\frac{MK}{n^2}) \geq N_1 \Omega(\frac nK) \Omega(\frac{MK}{n^2}) \geq N(\hat\be) \Omega(\frac Mn)
\end{align*}
which shows \eqref{thm3blema2} and finishes the proof.
\end{proof}

\subsection{Theory for the Pseudo-Likelihood EM Algorithm for PCABM}
\label{sec::supple_plem}

\subsubsection{Derivation of the Algorithm}
The likelihood function under the covariate-adjusted model is 
\begin{equation} 
\mathcal L(\mathbf c, \pg, P, \pmb\pi |A,Z) \varpropto \prod_{i=1}^n \pi_{c_i} \prod_{i<j} B_{c_ic_j}^{A_{ij}} e^{A_{ij} \bz_{ij}^\top \pg } \exp(-B_{c_ic_j} e^{ \bz_{ij}^\top \pg }).
\end{equation}
Now fix some coefficient estimate $\hat\pg$ since we are interested in maximizing over the label assignments. 
Enlighted by the latent class nature of the problem, we design some kind of EM algorithm to more efficiently estimate the community labels.
Using the latent class variable $w_i^l = \mathbbm{1}{(c_i=l)} $, we could write the log-likelihood as
\begin{equation}
l(\mathbf w, P, \pmb\pi |A,Z) = \sum_{i=1}^n \sum_{l=1}^K w_i^l \left[ \log(\pi_l) - \frac12 \sum_{j=1}^n e^{ \bz_{ij}^\top \hat\pg } \sum_{k=1}^K B_{lk} w_j^k +\frac12 \sum_{j=1}^n A_{ij} \sum_{k=1}^K w_j^k \log B_{lk} \right].
\label{lw/lat}
\end{equation}
The log-likelihood \eqref{lw/lat} does not directly lend itself to an efficient EM algorithm since in the E-step one has to optimize over all the latent class variables jointly, which is also an NP-hard discrete optimization problem \citep{amini2013pseudo}. In order to separate all the latent class variables so that we can optimize them separately and even analytically, we approximate the $w_j^k$'s in the square bracket in \eqref{lw/lat} with an estimated community label assignment $\mathbf e = (e_1,...,e_n) \in \{1,...,K\}^n$, i.e., in the likelihood of each edge $A_{ij}$, we see node $i$ as from true community $c_i$ and node $j$ as from an estimated community $e_j$. After that approximation, the pseudo-likelihood reads
\begin{equation}
\begin{aligned}
l(\mathbf w, P, \pmb\pi |A,Z) &= \sum_{i=1}^n \sum_{l=1}^K w_i^l \left[ \log(\pi_l) - \frac12 \sum_{j=1}^n e^{ \bz_{ij}^\top \hat\pg } \sum_{k=1}^K B_{lk} \mathbbm{1}{(e_j = k)} +\frac12 \sum_{j=1}^n A_{ij} \sum_{k=1}^K \mathbbm{1}{(e_j = k)} \log B_{lk} \right] \\
&=  \sum_{i=1}^n \sum_{l=1}^K w_i^l \left[ \log(\pi_l) - \frac12\sum_{k=1}^K B_{lk} \sum_{j=1}^n e^{ \bz_{ij}^\top \hat\pg }  \mathbbm{1}{(e_j = k)} +\frac12\sum_{k=1}^K \log B_{lk} \sum_{j=1}^n A_{ij}  \mathbbm{1}{(e_j = k)}  \right] \\
&= \sum_{i=1}^n \sum_{l=1}^K w_i^l \left[ \log(\pi_l) - \frac12\sum_{k=1}^K \hat \Xi_{ik} B_{lk}  +\frac12\sum_{k=1}^K b_{ik} \log B_{lk}  \right]
\end{aligned}
\label{pl}
\end{equation}
where $\hat \Xi_{ik} = \sum_{j=1}^n e^{ \bz_{ij}^\top \hat\pg }  \mathbbm{1}{(e_j = k)}$ and $b_{ik} = \sum_{j=1}^n A_{ij}  \mathbbm{1}{(e_j = k)}$. For any fixed $\mathbf e$ the pseudo-likelihood \eqref{pl} can be maximized over $(P,\pmb\pi)$ via standard EM algorithm, and with the outcome of the EM algorithm one can update label assignment estimate $\mathbf e$. Iterating this alternating process gives the proposed Algorithm \ref{alg:pcabm}. 

The initializations of $\hat{\pmb\pi}$ and $\hat P$ are their MLE under a given class assignment $\be$.
We would like to remark that in the inner loop of the algorithm, standard EM algorithm theory guarantees the convergence of pseudo likelihood when $\be$ is fixed and maximizing over $(P,\pmb\pi)$.
Besides, in Algorithm \ref{alg:pcabm} we use the same $\hat\pg$ throughout the algorithm instead of updating $\hat\pg$ with each new label estimate, partly because the asymptotic property of $\hat\pg$ does not depend on the label estimate we use,
and partly because the algorithm for estimating $\hat\pg$ is much more time consuming than the PLEM algorithm itself. Empirically, we tried updating $\hat\pg$ every  iteration in Algorithm \ref{alg:pcabm}, and the performance is almost the same as Algorithm \ref{alg:pcabm} itself.

\subsubsection{Consistency Result for the Algorithm \ref{alg:pcabm}}

In this section we consider the two balanced communities case, i.e., $K=2$, and each class has $m=\frac n2$ nodes. We further assume that the initial labeling estimation $\be$ is also balanced, i.e. $\hat\pi_1 = \hat\pi_2 = \frac12$ and $\sum_{i=1}^n \mathbbm{1}{(e_i=1)} = \sum_{i=1}^n \mathbbm{1}{(e_i=2)} = m$. Parameter $\varsigma \in(0,1)\setminus\{\frac12\} $ is the proportion of correctly labeled nodes by $\be$, i.e. $\sum_{i=1}^n \mathbbm{1}{(e_i=1,c_i=1)} = \sum_{i=1}^n \mathbbm{1}{(e_i=2,c_i=2)} =\varsigma m$. 
To more clearly parametrize the connection probability matrix $B$ in terms of network sparsity  we write   $B = \frac{1}m  \begin{pmatrix}a&b\\b&a\end{pmatrix} $ which is a $2\times 2$ matrix with $\varphi_n = {a+b} \to\infty$ being the average degree parameter which characterizes the sparsity of the graph.

For our theoretical analysis, we show the consistency of the output after one E-step of the PLEM algorithm
under the two balanced communities setting. The first E-step of the algorithm actually compares $\hat\pi_{il} \varpropto \hat\pi_l \prod_{k=1}^K \exp(b_{ik} \log\hat B_{lk} -\hat \Xi_{ik} \hat B_{lk}), l = 1,2$, and sets $e_i$ to be the $l$ corresponding to the larger $\hat\pi_{il}$. More explicitly, suppose the initial estimator of $B$ is given by $\hat B = \frac1m \begin{pmatrix}\hat a&\hat b\\\hat b&\hat a\end{pmatrix} $. Then the first E-step gives:
\begin{equation}
e_i = 1\quad \text{if} \quad (b_{i1} - b_{i2}) \log\frac{\hat a}{\hat b} + (\hat \Xi_{i1} - \hat \Xi_{i2})\frac{\hat b - \hat a}m > 0
\end{equation}
and $e_i=2$ otherwise. Recall $\hat \Xi_{ik} = \sum_{j=1}^n e^{ \bz_{ij}^\top \hat\pg }  \mathbbm{1}{(e_j = k)}$ and $b_{ik} = \sum_{j=1}^n A_{ij}  \mathbbm{1}{(e_j = k)}$. Intuitively, if node $i$ belongs to community 1, then $b_{i1} - b_{i2}$ should be positive as there are more links inside a community than between communities, and $\hat \Xi_{i1} - \hat \Xi_{i2}$ should cancel out to be around 0. Thus, with some concentration argument, one would expect the estimated label $e_i$ to be 1 with high probability.

\paragraph{Directed Case.}

The undirected graph does not immediately lend itself to concentration inequalities since $A$ has dependent entries when it is subject to the symmetric constraint. Thus, we first consider a directed graph in which the $A_{ij}$'s are fully independent. 
The consistency is first established on the directed model, and then in the proof for the undirected model a coupling is introduced to connect with the directed case result.
The original undirected block model is (we see $\bc$ as a fixed parameter)
\begin{equation}
A_{ij}|Z \stackrel{indpt}\sim \text{Poisson}(B_{c_ic_j} \exp(\bz_{ij}^\top\pg^0))\ \text{and} \ A_{ij} = A_{ji}, \bz_{ij} = \bz_{ji} \ \text{for} \ i<j,
\label{undirmod}
\end{equation}
where $indpt$ means $A_{ij}, i<j$ are independent conditional on $Z$.
In the directed block model, the symmetry assumptions on $A$ and $Z$ are withdrawn and one has
\begin{equation}
\tilde A_{ij}|\tilde Z \stackrel{indpt}\sim \text{Poisson}(\tilde B_{c_ic_j} \exp(\tbz_{ij}^\top\pg^0))\  \text{for all} \ i\neq j,
\label{dirmod}
\end{equation}
where we use `tilde' to indicate the directed model. 
Besides, we assume $\tilde{ Z}$ is asymmetric. In detail, similarly to Condition \ref{cond:zbd}, we assume the following condition on $\tilde{ Z}$:

\begin{cond}[Directed Case]
$\{\tilde{\mathbf{z}}_{ij},i\neq j\}$ are i.i.d. and uniformly bounded, i.e., for $\forall i\neq j$, $\|\tilde{\mathbf{z}}_{ij}\|_{\infty}\leq\zeta$, where $\zeta>0$ is the constant in Condition \ref{cond:zbd}. 
\label{cond:zbd1}
\end{cond}

Some notations in the upcoming theorems are defined as follows:
\begin{equation}
\label{init_collection}
\mathcal E_n^{\varsigma} = \left\{\be\in \{1,2\}^n: \sum_{i=1}^n \mathbbm{1}{(e_i=1,c_i=1)} = \sum_{i=1}^n \mathbbm{1}{(e_i=2,c_i=2)} =\varsigma m \right\}
\end{equation}
is the collection of initial labelings with correct proportion $\varsigma$; 
\begin{equation}
\label{misclass}
M_n(\be) = \min_{\phi \in\{(1,2),(2,1) \} } \frac1n \sum_{i=1}^n \mathbbm{1}{(\hat c_i(\be) \neq \phi(c_i))}
\end{equation}
is the misclassification rate of the algorithm, in which $\hat c_i(\be)$ is the algorithm's output label estimate after one E-step with input initial labels $\be$. By weak consistency we mean $M_n(\be) \stackrel{p}{\to} 0$. Let $L_1= 2(e^2+1)$ be a constant, $L_2 = \mathbb E \ezg$(or $\mathbb E \etzg$), $h(p) = -p\log p - (1-p)\log(1-p), p \in(0,1)$, and $\kappa_{\varsigma}(n):= \frac1n[\log(\frac{n}{4\pi\varsigma(1-\varsigma)}) + \frac1{3n} ] = o(1)$. Then, we have the following theorem. 

\begin{thm}[Directed Case] \label{thm::plem_directed} In the directed Pairwise Covariate Adjusted block model \eqref{dirmod} with two balanced communities and balanced initial labeling estimation, assume $\varsigma\in(0,1)\setminus \{\frac12\}, \tilde B = \frac1m \begin{pmatrix}a&b\\b&a\end{pmatrix}, a>b$ and the initial estimator $\hat a > \hat b>0$ with $\hat a= O(\hat a - \hat b )=O(a-b)$. Further assume Condition \ref{cond:zbd1} holds. Then there exists a sequence $\{u_n\}\in\mathbb R^+$ s.t.
\begin{equation}
\label{thm::plem_directedun}
\log u_n + \log\log u_n \geq \log(\frac2e h(\varsigma)) + \frac{L^2_2}{36(\chi+L_1L_2)}(2\varsigma-1)^2\frac{(a-b)^2}a,
\end{equation}
\begin{equation}
\label{thm::plem_directedMn} \mathbb
P\left( \sup_{\be \in \mathcal E_n^{\varsigma}} \tilde M_n(\be) \geq \frac{2h(\varsigma)}{\log u_n} \right) \leq
\exp(-n[h(\varsigma) - \kappa_{\varsigma}(n)]).
\end{equation}

In particular, if $\frac{(a-b)^2}a \to \infty$ we have $u_n\to\infty$ and algorithm's label estimate is consistent.
\end{thm}

\begin{remark}
Recall that $\varphi_n = {a+b}$ is the average degree of the graph. The condition $\frac{(a-b)^2}a \to \infty$ is almost equivalent to the $\varphi_n\to\infty$ condition for weak consistency in previous works.
\end{remark}

\paragraph{Undirected Case.}

To prove for the original undirected block model, we slightly refine our proof by
first conditioning on $Z$. 
 In order to adjust to the symmetry of $Z$, instead of assuming the strong condition of all $\{\tilde Z_{ij}\}_{i,j=1}^{n,n}$'s being i.i.d., in the refined proof we only need $\{Z_{ij} \}_{j=1}^n$ to be independent for all fixed $i$ (and also $\{Z_{ij} \}_{i=1}^n$ to be independent for all fixed $j$), while $\{Z_{i\cdot}\}_{i=1}^n$'s (and $\{Z_{\cdot j}\}_{j=1}^n$'s) do not need to be independent, where $Z_{i\cdot} =  \{Z_{ij} \}_{j=1}^n$. Thus, in the coupling with the directed case, we can build the coupling conditional on a symmetric $Z$.
Besides, the following technical condition is imposed:
\begin{cond} There exists a constant $d_0 >1$ s.t.
\begin{equation}
\min\left\{ \frac{(a-b)^2(2\varsigma-1)^2L_2^2}{36a^2\chi}, 
\frac{(a-b)^2(2\varsigma-1)^2L_2^2}{144(\hat a - \hat b)^2\chi}(\log\frac{\hat a}{\hat b})^2 \right\} \geq 2d_0 h(\varsigma), 
\end{equation}
where $\chi$ is the constant defined after Lemma \ref{lem:phi}.
\label{cond2}
\end{cond}
Similar to condition (17) in \cite{amini2013pseudo},
Condition \ref{cond2} mainly requires $\varsigma$ to be bounded away from $\frac12$, i.e. the initial labeling should have an accuracy close enough to $1$. It also postulates $a-b$ and ${\hat a}/{\hat b}$ can not be too small, i.e. there should be a gap between the connection probabilities within and between communities. These assumptions are heuristic: since we are considering one E-step it is natural that the consistency depends on a reasonably good initial labeling assignment, and a significant enough gap between within and cross-class entries of $B$ is the basis for our block identification \citep{decelle2011asymptotic}. Last but not least, we remark that the constants (36,144) in Condition \ref{cond2} are loose, and we write them for the sake of technical simplicity.

\begin{thm}[Undirected case] \label{thm::plem_undirect} In the undirected Pairwise Covariate Adjusted block model \eqref{undirmod} with two balanced communities and balanced initial labeling estimation, assume $\varsigma\in(0,1)\setminus \{\frac12\}, B = \frac1m \begin{pmatrix}a&b\\b&a\end{pmatrix}, a>b$ and the initial estimator $\hat a > \hat b>0$ with $\hat a= O(\hat a - \hat b )=O(a-b)$. Further assume Conditions \ref{cond:zbd} and \ref{cond2} hold. Then there exists a sequence $\{u_n\}\in\mathbb R^+$ s.t.
\begin{equation}
\label{thm::plem_undirectun}
\log u_n + \log\log u_n \geq \log(\frac2e h(\varsigma)) + \frac{L^2_2}{72(\chi+L_1L_2)}(2\varsigma-1)^2\frac{(a-b)^2}a,
\end{equation}
\begin{equation}
\label{thm::plem_undirectMn} \mathbb
P\left( \sup_{\be \in \mathcal E_n^{\varsigma}}  M_n(\be) \geq \frac{4h(\varsigma)}{\log u_n} \right) \leq 2
\exp(-n[h(\varsigma) - \kappa_{\varsigma}(n)]) 
+ 2\exp\left( -n[(d_0 -1)h(\varsigma) - \kappa'_{\varsigma}(n)]\right)
,
\end{equation}
where $\kappa_{\varsigma}(n) := \frac1n[\log(\frac{n}{4\pi\varsigma(1-\varsigma)}) + \frac1{3n} ] = o(1)$ is defined as in Theorem \ref{thm::plem_directed} and $\kappa'_{\varsigma}(n) := \frac1n[\log(\frac{n}{4\pi\varsigma(1-\varsigma)}) + \log(2n) + \frac1{3n} ] = \kappa_{\varsigma}(n) + \frac{\log(2n)}n =  o(1)$.

In particular, if $\frac{(a-b)^2}a \to \infty$ we have $u_n\to\infty$ and the algorithm's label estimate is consistent.
\end{thm}

\subsubsection{Proof of Theorem \ref{thm::plem_directed} and \ref{thm::plem_undirect}}
\label{subsubsec:proof_plem}
By Lemma \ref{lem:phi},
throughout this subsection without loss of generality, we work under the following two conditions: (i) $|e^{\bz_{ij}^\top \pg^0}|$ ,$|e^{\bz_{ij}^\top \hat\pg}|$ and $|e^{\bz_{ij}^\top \hat\pg} - \mathbb E e^{\bz_{ij}^\top \hat\pg}|$ are bounded by a constant $\chi$; (ii) $\|\hat\pg - \pg^0\| = o(1) $. 

\begin{proof}[Proof of Theorem \ref{thm::plem_directed}]
We give the proof for the $\varsigma > \frac12$ case. For the $\varsigma<\frac12$ case it suffices to flip the labels of $\bc$ by taking permutation $\phi=(2,1)$ in \eqref{misclass}.
For simplicity we assume $c_i=1, i=1,...,m; c_i=2, i=m+1,...,n$. First, consider a single node $i$. Without loss of generality suppose $c_i=1$ since $c_i=2$ renders the same bounds. Then $\hat c_i(\be) = 1 \Leftrightarrow$
\begin{equation}
\label{xi}
\tilde \xi_i := (\tilde b_{i1} - \tilde b_{i2}) \log\frac{\hat a}{\hat b} + (\hat \Xi_{i1} - \hat \Xi_{i2})\frac{\hat b - \hat a}m > 0.
\end{equation}
To show consistency on this single node we seek to bound $\Pr(-\tilde\xi_i \geq 0)$. Since 
$$\tilde b_{i2} - \tilde b_{i1} = \sum_{j=1}^n \tilde A_{ij} \mathbbm{1}{(e_j=2)} - \sum_{j=1}^n \tilde A_{ij} \mathbbm{1}{(e_j=1)}
$$
and $\tilde A_{ij} \sim \text{Poisson}(\tilde B_{c_ic_j} \exp(\tbz_{ij}^\top\pg^0))$, by Bernstein's inequality for Poisson variable (Lemma \ref{lem:bernin} and \ref{lem:bern}), we have (first conditional on $Z$, then unconditional)
\begin{equation}
\label{concen1}
\begin{aligned}
\Pr\left(\tilde b_{i2} - \tilde b_{i1} \geq t_1 + \sum_{j=1}^n \tilde B_{c_ic_j} \etzg \sigma_j\right) &\leq \exp\left(-\frac{t_1^2}{2(\sum_{j=1}^n \tilde B_{c_ic_j} \etzg + L_1t_1) } \right)\\
&\leq \exp\left(-\frac{t_1^2}{2(2a \chi + L_1t_1) } \right)
\end{aligned}
\end{equation}
where $\sigma_j = \mathbbm{1}{(e_j=2)} - \mathbbm{1}{(e_j=1)}$. 
Further, since $\mathbb E(\sum_{j=1}^n \tilde B_{c_ic_j} \etzg \sigma_j) = (a-b)(1-2\varsigma)L_2$,
by Hoeffding's inequality for bounded variables, 
\begin{equation}
\label{concen2}
\begin{aligned}
\Pr\left(\sum_{j=1}^n \tilde B_{c_ic_j} \etzg \sigma_j \geq t_2 + (a-b)(1-2\varsigma)L_2 \right) \leq 
\exp\left(-\frac{m^2t_2^2}{2na^2\chi^2 } \right) = \exp\left(-\frac{mt_2^2}{4a^2\chi^2} \right).
\end{aligned}
\end{equation}
For the second half of \eqref{xi}, again applying Hoeffding for bounded variables, we have
\begin{equation}
\label{concen3}
\begin{aligned}
&~\Pr\left( \frac{\hat a - \hat b}m (\hat \Xi_{i1}-\hat \Xi_{i2})\geq t_3 \right)\\ 
\leq&~ \Pr\left( \frac{\hat a - \hat b}m (\sum_{j=1}^m e^{\bz_{ij}^\top\pg^0} - \sum_{j=m+1}^n e^{\bz_{ij}^\top\pg^0} )\geq \frac{t_3}{2} \right) + \Pr\left( \frac{\hat a - \hat b}m (\sum_{j=1}^n |e^{\bz_{ij}^\top\pg^0}-e^{\bz_{ij}^\top\hat\pg}|)\geq \frac{t_3}{2}  \right)
\\
\leq&
\exp\left(- \frac{mt_3^2}{16\chi^2 (\hat a -\hat b)^2 } \right)
\end{aligned}
\end{equation} for $t_3 = \Omega(\hat a - \hat b) $, where the second probability on the second line is $0$ by the conditions stated at the beginning of subsection \ref{subsubsec:proof_plem}.
Define $p_i(r) = \Pr(- \tilde \xi_i \geq r)$. Combining \eqref{concen1},\eqref{concen2},\eqref{concen3} with $t_1 = t_2 = \frac{(a-b)(2\varsigma-1)L_2}3, t_3 = \frac{(a-b)(2\varsigma-1)L_2}3 \log \frac{\hat a}{\hat b}$, we get
\begin{equation}
\label{pi0}
\begin{aligned}
p_i(0) =& \Pr\left( (\tilde b_{i2} - \tilde b_{i1}) \log\frac{\hat a}{\hat b} + (\hat \Xi_{i1} - \hat \Xi_{i2})\frac{\hat a - \hat b}m \geq 0 \right) \\
\leq& \exp\left( -\frac{ (a-b)^2(2\varsigma-1)^2L_2^2 }{ 36a\chi + 6 L_1(a-b)(2\varsigma-1)L_2 } \right) + \exp\left( -\frac{ m (a-b)^2(2\varsigma-1)^2L_2^2 }{ 36a^2\chi^2 } \right)\\
&+ \exp\left( -\frac{ m (a-b)^2(2\varsigma-1)^2L_2^2 }{ 144(\hat a - \hat b)^2\chi^2 } (\log \frac{\hat a}{\hat b})^2 \right) \\
\leq& \exp\left( -\frac{L_2^2}{ 36(\chi + L_1L_2) }(2\varsigma-1)^2\frac{(a-b)^2}{a} \right),
\end{aligned}
\end{equation}
where the last inequality holds when $n$ is large enough, since the scale of $a,b$ is the same as the average degree $\varphi_n$ which should go to $\infty$ much slower than $n$ in a sparse network (even could be slower than $\log n$ for usual weak consistency). Denote by the bound on the last line of \eqref{pi0} by $\bar \mu(0)$. 

The same argument applies to all nodes in both communities so we have $p_i(0) \leq \bar\mu$ for $i=1,...,n$.
Note $\tilde M_n(\be) \leq \frac1n\sum_{i=1}^m \mathbbm{1}{(-\tilde\xi_i \geq 0)} + \frac1n\sum_{i=m+1}^n \mathbbm{1}{(\tilde\xi_i \geq 0)}$ which is a sum of independent Bernoulli random variables by the asymmetry. Thus, applying a Chernoff bound (Lemma 5 in \cite{amini2013pseudo}) on this sum we have for any $u>\frac1e$,
\begin{equation}
\label{BiLeChernoff}
\Pr\left( \tilde M_n(\be) \geq eu\bar\mu(0) \right) \leq \Pr\left( \frac1n \sum_{i=1}^m \mathbbm{1}{(-\tilde\xi_i \geq 0)}+ \frac1n \sum_{i=m+1}^n \mathbbm{1}{(\tilde\xi_i \geq 0)} \geq eu\bar\mu(0)  \right) \leq \exp(-en\bar\mu(0) u\log u).
\end{equation}
Finally by a union bound over $\be$'s such that the initial correct classification rate is $\varsigma$ we obtain
\begin{equation}
\label{final1}
P[\sup_{\be\in\mathcal E_n^{\varsigma}} \tilde M_n(\be) \geq e u_n \bar \mu(0) ] \leq \exp\{n[h(\varsigma) - e \bar \mu(0)u_n \log u_n + \kappa_{\varsigma}(n) ] \}
\end{equation}
where the constant term comes from the cardinality bound $|\mathcal E_n^{\varsigma}| = {m \choose \varsigma m}^2 \leq
\exp(2m[h(\varsigma) + \kappa_{\varsigma}(2m)] )$. Take $u_n\log u_n = \frac{2h(\varsigma) }{e\bar\mu(0)}$ in \eqref{final1} so that $eu_n\bar\mu(0)\to 0$ as $a,b\to\infty$ while the right hand side bound of \eqref{final1} goes to 0 as $n\to\infty$, we derive \eqref{thm::plem_directedun} and \eqref{thm::plem_directedMn} which completes the proof of Theorem \ref{thm::plem_directed}.
\end{proof}

\begin{proof}[Proof of Theorem \ref{thm::plem_undirect}]
Again without loss of generality we consider the $\varsigma>\frac12$ case and assume $c_i=1,i=1,...,m; c_i = 2 , i=m+1,...,n$.
We establish a coupling between the above undirected model and a directed model conditional on the symmetric covariate matrix $Z$. Let 
\begin{equation}
\tilde A_{ij} | Z \stackrel{indpt}{\sim} \text{Poisson}(\frac12 B_{c_ic_j} \etzg)
\label{pardirmod}
\end{equation}
where $\tbz_{ij} = \tilde \bz_{ji} = \bz_{ij} = \bz_{ji}$, and all $\tilde A_{ij}, i\neq j\in \{1,...,n\}$ are independent. Note \eqref{pardirmod} is different from the directed model \eqref{dirmod} considered in the previous subsection as $\tilde Z$ is now subject to the symmetry constraint. Let
\begin{equation}
A_{ij} = A_{ji} = \tilde A_{ij} + \tilde A_{ji},
\label{coupling}
\end{equation}
then we have  
\begin{align*}
A_{ij}|Z \stackrel{indpt}\sim \text{Poisson}(B_{c_ic_j} \exp(\bz_{ij}^\top\pg^0))\ \text{and} \ A_{ij} = A_{ji}, \bz_{ij} = \bz_{ji} \ \text{for} \ i< j
\end{align*}
which is exactly model \eqref{undirmod}. In other words,
\eqref{coupling} defines a coupling between the directed model \eqref{pardirmod} and undirected model \eqref{undirmod}.

Consider the classification by E-step for a node $i$ whose $c_i = 1$. $\hat c_i(\be) = 1 \Leftrightarrow$
\begin{equation}
\label{xi_undir}
\xi_i := ( b_{i1} - b_{i2}) \log\frac{\hat a}{\hat b} + (\hat \Xi_{i1} - \hat \Xi_{i2})\frac{\hat b - \hat a}m > 0,
\end{equation}
where $ b_{ik}  = \sum_{j=1}^n  A_{ij} \mathbbm{1}{(e_j=k)} $ and $\hat \Xi_{ik} = \sum_{j=1}^n e^{ \bz_{ij}^\top \hat\pg }  \mathbbm{1}{(e_j = k)}$. Using relationship \eqref{coupling} we decompose $\xi_i$ as 
\begin{equation}
\label{xi_decomp}
\xi_i = \tilde \xi_{i*} + \tilde \xi_{*i} 
\end{equation}
where
\begin{align*}
\tilde \xi_{i*} &:= (\tilde b_{i1} - \tilde b_{i2}) \log\frac{\hat a}{\hat b} + (\hat \Xi_{i1} - \hat \Xi_{i2})\frac{\hat b - \hat a}{2m}; \\
\tilde b_{ik} &= \sum_{j=1}^n  \tilde A_{ij} \mathbbm{1}{(e_j=k)}, k = 1,2; \\
\tilde \xi_{*i} &:= (\tilde b_{1i} - \tilde b_{2i}) \log\frac{\hat a}{\hat b} + (\hat \Xi_{1i} - \hat \Xi_{2i})\frac{\hat b - \hat a}{2m}; \\
\tilde b_{ki} &:= \sum_{j=1}^n  \tilde A_{ji} \mathbbm{1}{(e_j=k)} , k= 1,2;\\
\hat \Xi_{ki} &:= \sum_{j=1}^n e^{\bz_{ji}^\top \hat\pg} \mathbbm{1}{(e_j=k)}, k=1,2.
\end{align*}
Note that $\hat \Xi_{ik} = \hat \Xi_{ki}$ since $Z$ is symmetric, so we will not use the notation $\hat \Xi_{ki}$ and will only use $\hat \Xi_{ik}$ below for clarity.

Equation \eqref{xi_decomp} indicates $\mathbbm{1}{(-\xi_i \geq 0)} \leq \mathbbm{1}{(-\tilde\xi_{i*}\geq 0)} + \mathbbm{1}{(-\tilde\xi_{*i}\geq 0)}$. Thus, in order to bound the misclassification rate $M_n(\be) \leq \frac1n\sum_{i=1}^m \mathbbm{1}{(-\xi_i \geq 0)} + \frac1n\sum_{i=m+1}^n \mathbbm{1}{(\xi_i \geq 0)}$ it suffices to bound 
\begin{align*}
\tilde M_{n*}(\be) &:= \frac1n \sum_{i=1}^m \mathbbm{1}{(-\tilde\xi_{i*}\geq 0)} + \frac1n \sum_{i=m+1}^n \mathbbm{1}{(\tilde\xi_{i*}\geq 0)} \quad \text{and} \\
\tilde M_{*n}(\be) &:= \frac1n \sum_{i=1}^m \mathbbm{1}{(-\tilde\xi_{*i}\geq 0)} + \frac1n \sum_{i=m+1}^n \mathbbm{1}{(\tilde\xi_{*i}\geq 0)}.
\end{align*}
We went through all those coupling and decomposition arguments because $\xi_i$'s are not independent due to the symmetry constraint
but conditional on $Z$ the $\tilde \xi_{i*}$'s are independent (and so are $\tilde \xi_{*i}$'s). Thus, we can derive concentration bounds on $\tilde M_{n*}(\be)$ and $\tilde M_{*n}(\be)$ respectively (conditional on $Z$) and unite them together to achieve our final result.
Now we repeat the steps in the proof of theorem \ref{thm::plem_directed}, but with a modification by conditional on $Z$ arguments to adapt to the symmetric $Z$.

Same as in \eqref{concen1} we have conditional on any given $Z$
\begin{equation}
\label{concen21}
\begin{aligned}
\Pr\left( \left. \tilde b_{i2} - \tilde b_{i1} \geq t_1' + \sum_{j=1}^n \frac12 B_{c_ic_j} \ezg \sigma_j\right|Z \right) &\leq \exp\left(-\frac{(t_1')^2}{2(\sum_{j=1}^n \frac12 B_{c_ic_j} \ezg + L_1t_1') } \right)\\
&\leq \exp\left(-\frac{(t_1')^2}{2(a \chi + L_1t_1') } \right).
\end{aligned}
\end{equation}
Consider the set of $Z$'s that belong to the compliments of events in \eqref{concen2} and \eqref{concen3}:
\begin{equation}
\label{Di*}
D_{i*} := \left\{  Z_{i\cdot} \left| \sum_{j=1}^n \frac12 B_{c_ic_j} \ezg \sigma_j < t_2' + \frac{(a-b)(1-2\varsigma)L_2}2;\ \frac{\hat a - \hat b}{2m} (\hat \Xi_{i1}-\hat \Xi_{i2}) <t_3'
\right. \right\}
\end{equation}
where $Z_{i\cdot} = \{\bz_{ij}\}_{j=1}^n$ and $t_2' = \frac{(a-b)(2\varsigma-1)L_2}6, t_3' = \frac{(a-b)(2\varsigma-1)L_2}6 \log \frac{\hat a}{\hat b}$. Then let $t_1' = \frac{(a-b)(2\varsigma-1)L_2}6$ we have 
\begin{equation}
\label{concen212}
\begin{aligned}
\Pr\left( \left. \tilde b_{i2} - \tilde b_{i1} \geq t_1' + \sum_{j=1}^n \frac12 B_{c_ic_j} \ezg \sigma_j\right|Z\in D_{i*} \right) 
&\leq \exp\left(-\frac{(t_1')^2}{2(a \chi + L_1t_1') } \right) \\
&= \exp\left( -\frac{ (a-b)^2(2\varsigma-1)^2L_2^2 }{ 72a\chi + 12 L_1(a-b)(2\varsigma-1)L_2 } \right) \\
&\leq \exp\left( -\frac{L_2^2}{ 72(\chi + L_1L_2) }(2\varsigma-1)^2\frac{(a-b)^2}{a} \right)
\end{aligned}
\end{equation}
Denote by the bound on the last line of \eqref{concen212} by $\bar \mu_1(0)$. By the definition \eqref{Di*}, 
\begin{align*}
\Pr\left(\left. -\tilde\xi_{i*} \geq 0\right| Z\in D_{i*} \right) \leq
\Pr\left( \left. \tilde b_{i2} - \tilde b_{i1} \geq t_1' + \sum_{j=1}^n \frac12 B_{c_ic_j} \ezg \sigma_j\right|Z\in D_{i*} \right)
\leq \bar \mu_1(0). 
\end{align*}
Conditional on $Z$, $\tilde\xi_{i*}$'s are independent, so same as in \eqref{BiLeChernoff} but conditional on a fixed $Z \in \cap_{i=1}^n D_{i*}$ ($D_{i*}$'s for $c_i=2$ are correspondingly defined) we obtain
\begin{equation}
\label{BiLeChernoff2}
\Pr\left(\left. \tilde M_{n*}(\be) \geq eu\bar\mu_1(0) \right| Z \in \cap_{i=1}^n D_{i*} \right) \leq \exp(-en\bar\mu_1(0) u\log u).
\end{equation}
Note that \eqref{concen2} and \eqref{concen3} still hold as long as $\bz_{ij},j=1,...,n$ are independent when fixing any $i$, we can upper bound the probability of the event $\left(\cap_{i=1}^n D_{i*}\right)^c$ by (by canceling out a $\frac12$; $D_{i*}^c$ corresponds exactly to the events in \eqref{concen2} and \eqref{concen3})
\begin{equation}
\label{PofDi*}
\begin{aligned}
\Pr\left(\left(\cap_{i=1}^n D_{i*}\right)^c\right) &\leq \sum_{i=1}^n \Pr(D_{i*}^c)\\ &\leq
\sum_{i=1}^n \left[ \exp\left(-\frac{mt_2^2}{4a^2\chi^2} \right) + \exp\left(- \frac{mt_3^2}{16\chi^2 (\hat a -\hat b)^2 } \right) \right] \\ 
&= \exp\left(-\frac{mt_2^2}{4a^2\chi^2} + \log n \right) + \exp\left(- \frac{mt_3^2}{16\chi^2 (\hat a -\hat b)^2 } + \log n \right)\\
&\leq \exp(-2md_0h(\varsigma)+ \log(2n))
\end{aligned}
\end{equation}
where $t_2,t_3$ are defined as in the proof of Theorem \ref{thm::plem_directed}, and the last inequality comes from Condition \ref{cond2}.
Thus, combining \eqref{BiLeChernoff2} and \eqref{PofDi*}, and by applying a same argument on $\tilde M_{n*}(\be)$, we have
\begin{align*}
\Pr(M_n(\be) \geq 2eu\bar\mu_1(0)) \leq& \Pr(\tilde M_{n*}(\be)\geq eu\bar\mu_1(0)) + \Pr(\tilde M_{*n}(\be)\geq eu\bar\mu_1(0)) \\
\leq& \Pr\left(\left. \tilde M_{n*}(\be) \geq eu\bar\mu_1(0) \right| Z \in \cap_{i=1}^n D_{i*} \right) + \Pr\left(\left(\cap_{i=1}^n D_{i*}\right)^c\right)\\
&+ \Pr\left(\left. \tilde M_{*n}(\be) \geq eu\bar\mu_1(0) \right| Z \in \cap_{i=1}^n D_{*i} \right) + \Pr\left(\left(\cap_{i=1}^n D_{*i}\right)^c\right) \\
\leq& 2\left[ \exp(-en\bar\mu_1(0) u\log u) + \exp(-2md_0h(\varsigma)+ \log(2n)) \right]
\end{align*}

Finally, by applying the same union bound as in \eqref{final1} and again taking $u_n\log u_n = \frac{2h(\varsigma) }{e\bar\mu(0)}$(which results in \eqref{thm::plem_undirectun}), we get
\begin{equation}
\label{final2} \mathbb
P\left(\sup_{\be\in\mathcal E_n^{\varsigma}} M_n(\be) \geq \frac{4h(\varsigma)}{\log u_n} \right) 
\leq 2 \exp\{n[-h(\varsigma) + \kappa_{\varsigma}(n) ] \} + 2\exp\{n[-(d_0-1)h(\varsigma) + \kappa_{\varsigma}(n)] + \log(2n)\}
\end{equation}
which is exactly the result of \eqref{thm::plem_undirectMn}.
\end{proof}

\subsection{Proof of Theorem~\ref{THM:SC}}\label{sec:supple:spectralbd}

From Theorem~\ref{THM:ASY}, it is not hard to show the following error bound for $\hat{\boldsymbol{\gamma}}$.

\begin{lemma}\label{lem:eta}
    For any constant $\eta>0$, $\exists$ positive constants $C_\eta$ and $v_\eta$ s.t., $\emph{Pr}(\sqrt{n\rho_n}\|\boldsymbol{\gamma}^0-\hat{\boldsymbol{\gamma}}\|_\infty>\eta)<C_\eta\exp(-v_\eta n)$.
\end{lemma}

To prove Theorem \ref{THM:SC}, we first state a concentration result for directed, $\pg^0$ adjusted adjacency matrix (Theorem \ref{newthm21}), then derive Theorem \ref{THM:SC} based on Theorem \ref{newthm21}, and finally give a proof of Theorem \ref{newthm21}.

\begin{thm}[Concentraion for directed, $\pg^0$ adjusted adjacency matrix; A covariate adjusted, Poisson variant of Theorem 2.1 of \cite{le2017concentration}] \label{newthm21}
Let $A$ be the adjacency matrix generated by the directed PCABM $(M, B, Z, \pg^0)$.
Assume Condition \ref{cond:zbd} holds; and $\Bm\leq C_{\bar B}$.
Also let 
\begin{equation} 
\label{def_d}
d = \max_{i,j} n P_{ij} = \max_{i,j} n B_{c_ic_j}.
\end{equation}
Consider the adjusted adjacency matrix $A^{0}$  derived from $\pg^0$, i.e. $A^0_{ij} = A_{ij} / \exp(\bz^\top_{ij} \pg^0)$. For any $r>1$, the following holds with probability at least $1- 3n^{-r}$: Consider any subset consisting of at most $10n/d$ vertices, and reduce the weights of the edges incident to those vertices in an arbitrary way.
Denote the adjacency matrix of the new (weighted) graph by $A^{0R}$;
let $d^{(R)} $ be the maximal row and column $l_1$ norm of $A^{0R}$. Then $A^{0R}$ satisfies
\begin{equation} \label{specbd:A'R_assym}
\|A^{0R} - P\| \leq Cr\sqrt{\xi^3} (\sqrt{d} + \sqrt{d^{(R)}})
\end{equation}
where $C$ is a constant that does not depend on $\xi$. Moreover, the bound \eqref{specbd:A'R_assym} still holds when $d^{(R)}$ is the maximal row and column $l_2$ norm of $A^{0R}$.
\end{thm}
 In this result and in the rest of this section, $C$ denotes an absolute constant whose value may be different from line to line.
\begin{proof}[Proof of Theorem \ref{THM:SC}]
By Lemma \ref{lem:eta}, 
\begin{equation}
\label{gammaconcen}
\|\pg^0-\hat\pg\|_{\infty} \leq \eta / \sqrt{n\rho_n}
\end{equation}
 with probability at least $1- c_{\eta} \exp(-v_{\eta} n)$. All of our arguments in this proof are conditioned on event \eqref{gammaconcen}.

We first prove the result for directed case, i.e. $A_{ij} \sim \text{Poisson}(B_{c_ic_j} \exp(\bz_{ij}^\top \pg^0))$ independent for all $i\neq j$, and then use a coupling argument to extend the result to undirected case.

In the directed setting, we apply Thoerem \ref{newthm21} on $A^0$, with the ``arbitrary set of vertices incident to reweighted edges'' chosen as the set $\mathcal{I}^R = \{i \in [n] : d_i' \geq \lambda^R d'  \} $, and  edges incident to those vertices reweighted by $A_{ij}^{0R} = A_{ij}^0 \sqrt{\lambda_i\lambda_j}, \lambda_i = \min\{\lambda^R d' /d_i', 1 \} $. Thus the $A^{0R} $ in Theorem \ref{newthm21} is the same as the $A^{0R} $ as in Algorithm \ref{alg:scwa}.
To show this choice is valid, we need to verify the condition that $|\mathcal{I}^R| \leq 10n /d $. By Lemma \ref{newlem35}, with probability $1-n^{-r}$ there are at most $10n/d $ rows in $A^{0} $ with $l_1$ norm $\geq 0.8\xi^2 r d $. Then  by \eqref{gammaconcen}, there are at most $10n/d $ rows in $A'$ with $l_1$ norm $\geq 0.8 \xi^2 r d \exp(p\zeta \eta/\sqrt{n\rho_n})  $, i.e., there are at most $10n/d$ nodes $i$ such that $d_i' \geq 0.8\xi^2 r d \exp(p\zeta \eta/\sqrt{n\rho_n}) $.
On the other hand, by a Bernstein's inequality (Lemma \ref{lem:bernin}), we have $d^{(0)} \geq c_r(c_{\bar B}/C_{\bar B} ) d $ with probability $1-n^{-r} $, where $d^{(0)} $ is the average degree of $A^0$, and $c_r$ is a constant depending on $r$. Combining that with \eqref{gammaconcen}, we have $d' \geq \exp(-p\zeta \eta/\sqrt{n\rho_n}) c_r(c_{\bar B}/C_{\bar B} ) d  $. Thus, by choosing $\lambda^R = 0.8\xi^2 r \exp(2p\zeta \eta/\sqrt{n\rho_n}) C_{\bar B} / (c_rc_{\bar B}) $, we have 
$$
|\mathcal{I}^R| \leq |\{ i:d_i' \geq \lambda^R \exp(-p\zeta \eta/\sqrt{n\rho_n}) c_r(c_{\bar B}/C_{\bar B} ) d   \}  | = |\{i: d_i' \geq 0.8\xi^2 r d \exp(p\zeta \eta/\sqrt{n\rho_n}) \}|\leq 10n/d.
$$

Now by Theorem \ref{newthm21} it suffices to bound $\|A'^R - A^{0R} \|$. Let $w_{ij}=\sqrt{\lambda_i\lambda_j} \in[0,1]$ be the weight imposed on edge $A_{ij}$.

For any $\bx \in \mathbb R^n$ with $\|\bx\|_2=1$, 
\begin{equation}
\label{diffspec}
\begin{aligned}
| \bx^\top (A'^R-A^{0R}) \bx | =& \left| \sum_{i,j} x_ix_j\left( A_{ij} / \exp(\bz_{ij}^\top\hat\pg) - A_{ij} / \exp(\bz_{ij}^\top\pg^0) \right) w_{ij} \right| \\
\leq& \sum_{i,j} |x_i||x_j| \left| A^{0R}_{ij} \left[ \exp\{\bz_{ij}^\top (\pg^0-\hat\pg) \}-1 \right] \right| \\
<& \sum_{i,j} 2|x_i||x_j| A_{ij}^{0R} \left| \bz^\top_{ij} (\pg^0-\hat\pg) \right| \\
\leq& \frac{2}{\sqrt{n\rho_n}} \sum_{i,j} |x_i||x_j|A_{ij}^{0R}\\
\leq& \frac{2}{\sqrt{n\rho_n}} ||A^{0R}||
\end{aligned}
\end{equation}  
in which the third line is due to $|e^t-1|<2|t|$ when $|t|<1$, and the fourth line is due to \eqref{gammaconcen} and $\eta = (p\zeta)^{-1} $. From \eqref{diffspec} we get $ \| A'^R-A^{0R} \| \leq 2\|A^{0R}\|/\sqrt{n\rho_n}$. Furthermore, we have 
\begin{equation}
\label{Pspec}
\|P\|\leq \|P\|_F \leq \sqrt{n^2\rho_n^2 \Bm^2}\leq n\rho_n C_{\bar B}.
\end{equation} 
Combining \eqref{Pspec} with Theorem \ref{newthm21}, we could bound $\|A^{0R}\|$ by
$\|A^{0R}\| \leq \|A^{0R}-P\| + \|P\| \leq Cr\sqrt{\xi^3}(\sqrt{d} + \sqrt{d^{(R)}}) + n\rho_n C_{\bar B} $. Thus, 
\begin{equation}
\begin{aligned}
\|A'^R - P\|\leq&
\| A'^R-A^{0R} \| + \|A^{0R}-P\| \\
\leq&  2\|A^{0R}\|/\sqrt{n\rho_n} + \|A^{0R}-P\| \\
\leq& 2C_{\bar B} \sqrt{n\rho_n} + Cr\sqrt{\xi^3} (1+\frac1{\sqrt{n\rho_n}})(\sqrt{d}+\sqrt{d^{(R)}}).
\end{aligned} 
\label{A'r_dir_concen}
\end{equation}
It is not hard to see that after performing the previously stated reduce-weight regularization on $A^0$, the resulting $A^{0R}$ has row and column $l_2$ norms bounded by $d$ up to a constant with probability $1-n^{-r}$:
\begin{align*}
\sum_{j}{ A_{ij}^0}^2 \lambda_i\lambda_j \leq&  (\lambda^R)^2 \sum_{j}{ A_{ij}^0}^2 \cdot \frac{d'}{d_i'} \frac{d'}{d_j'}\\
\leq& (\lambda^R)^2e^{\frac{2p\eta\zeta}{\sqrt{n\rho_n}}} \sum_{j}{ A_{ij}^0}^2 \cdot \frac{d^{(0)}}{d_i^{(0)}} \frac{d^{(0)}}{d_j^{(0)}} \\
\leq& (\lambda^R)^2e^{\frac{2p\eta\zeta}{\sqrt{n\rho_n}}} {d^{(0)}}^2 \sum_j A_{ij}^0 / d_i^{(0)} \\
\leq& C d^2
\end{align*}
where $d_i^{(0)}$ and $d^{(0)}$ are $i$th node degree and average degree of $A^0$; in the last step a Poisson concentration (Lemma \ref{lem:bernin}) is used to obtain ${d^{(0)}}^2 \leq C_r d^2 $ with probability $1-n^{-r}$ for some constant $C_r$ depending on $r$.
Hence the concentration \eqref{A'r_dir_concen} reads $\|A'^R - P\|\lesssim\sqrt{\varphi_n}$.

Now we bridge our result for the directed case to the undirected case with a coupling approach \citep{amini2013pseudo}. Consider the directed model
\begin{align*}
\label{dirmodel}
\tilde A_{ij}|Z \stackrel{indpt}{\sim} \text{Poisson}(\frac12B_{c_ic_j}\exp(\bz_{ij}^\top\pg^0)) \quad \text{for} \quad \forall i\neq j.
\end{align*}
Now let $A_{ij} = \tilde A_{ij} + \tilde A_{ji}$. Then the resulting adjacency matrix $A = \{A_{ij}\} $ satisfies:
(1) $A_{ij} = A_{ji}$ for all $i<j$; (2) $A_{ij},i<j$ are all independent; and (3) $A_{ij}|Z {\sim} \text{Poisson}(B_{c_ic_j}\exp(\bz_{ij}^\top\pg^0))$. Thus, the $\{A_{ij}\}$ defined this way follows our original undirected PCABM. 
Denote by $\tilde A'^R$ the adjusted version of $\tilde A$ on the new weighted graph. The result for directed case gives $\|\tilde A'^R - P/2\| \lesssim 
\sqrt{\varphi_n} $, so that a triangle inequality $\|A'^R - P\| \leq \|\tilde A'^R - P/2\| + \|\tilde A'^{R\top} - P/2\|$ proves the statement of Theorem \ref{THM:SC} for the undirected case.
\end{proof}

\subsubsection{Proof of Theorem \ref{newthm21}}

\begin{thm}[Graph decomposition for the covariate adjusted Poisson adjacency matrix; counterpart of Theorem 2.6 in \cite{le2017concentration}] \label{newthm26}
 Consider the adjusted adjacency matrix $A^{0}$  derived from $\pg^0$ in Theorem \ref{newthm21}. Under the same assumptions of Theorem \ref{newthm21}, for any $r>1$, the following holds with probability at least $1- 3n^{-r}$: One can decompose the set of edges $[n] \times [n]$ into three classes $\mathcal N, \mathcal R$ and $\mcC$ so that the following properties are satisfied for the adjusted adjacency matrix $A^{0}$:
\begin{enumerate}
\item The graph concentrates on $\mcN$, namely $\|(A^{0} - \mathbb E A^{0})_{\mcN} \| \leq Cr\sqrt{\xi^3 d}$.
\item Each row of $A^0_{\mcR}$ and each column of $A^0_{\mcC}$ has $l_1$ norm $\leq 32r\xi^2$.
\item $\mcR$  intersects at most $\sqrt2n/(\sqrt2-1)d$ columns, and $\mcC$ intersects at most $\sqrt2n/(\sqrt2-1)d$ rows of $[n] \times [n]$.
\end{enumerate}
Moreover, the same result also holds for the second property being replaced by ``each row of $A^{0}_{\mcR}$, each column of $A^{0}_{\mcC}$ has $l_0$ norm $\leq 32r\xi$.''
\end{thm}
An illustration of the graph decomposition in Theorem \ref{newthm26} is given in Figure \ref{decomp26fig} (the picture comes from \cite{le2017concentration}). We put those well-concentrated edges into $\mcN$, while the number of the not well-behaved, high-degree node attached edges are bounded as $\mcR$ and $\mcC$ have bounded column and row norms.

\begin{figure}[H]
\centerline{
\includegraphics[ width=1.5 in,]{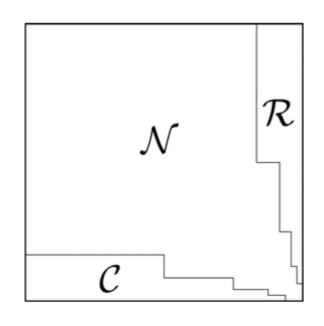}
}
\caption{An illustration of the graph decomposition in Theorem \ref{newthm26}.}
\label{decomp26fig}
\end{figure}

Several lemmas are established as steps for 
the proof of the decomposition Theorem \ref{newthm26}. 

\begin{lemma}[Concentration in $l_\infty\to l_2$ norm]
\label{newlem33}
Let $1 \leq m \leq n$ and $\alpha \geq m/n$. Then for $r \geq 1$ the following holds with probability at least $1 - n^{-r} $. Consider a block $I \times J$ of size $m \times m$. Let $I'$ be the set of indices of the rows of $A^{0}_{I\times J}$ whose $l_1$ norm $\leq \alpha d$. Then
\begin{align*}
\| (A^{0}-\mathbb E A^{0})_{I'\times J} \|_{\infty\to2} \leq Cg(\xi) \sqrt{ \alpha d m r \log(en/m) }
\end{align*}
where $g(\xi) = \sqrt{(1+2\bar L)\xi}$ is a constant that is a function of $\xi$ with $\bar L=2(e^2+1)$.
\end{lemma}
\begin{proof}[Proof of Lemma \ref{newlem33}]
Note that centralized Poisson random variables satisfy the Bernstein condition $\mathbb E |A_{ij}-\mathbb E A_{ij}|^k \leq \frac12 \mathbb E (A_{ij}-\mathbb E A_{ij})^2 L^{k-2} k!$ (Lemma \ref{lem:bern}). By the scale invariance of the Bernstein condition, the adjusted adjacency matrix elements also satisfy the Bernstein condition with the constant $L$ replaced by $L\xi$. Thus, defining 
\begin{align*}
x\in\{1,-1\}^m, \quad X_i := \sum_{j\in J} (A_{ij}^0- \mathbb E A_{ij}^0)x_j, \quad \eta_i:= \mathbbm{1}(\sum_{j\in J} A_{ij}^0 \leq \alpha d), 
\end{align*}
one can recover equations (3.5) and (3.6) in the proof of Lemma 3.3 of \cite{le2017concentration} by
\begin{align*}
\Pr(|X_i\eta_i|>tm) \leq 2\exp\left( \frac{-mt^2/2 }{\xi d/n + \bar L \xi t } \right) \leq 2\exp\left( \frac{-m^2t^2/2 }{\xi \alpha d + 2\bar L\xi\alpha d } \right)
\end{align*}
where the last inequality is because $\alpha \geq m/n$ and
\begin{align*}
|X_i\eta_i| \leq \sum_{j\in J} A_{ij}^0 + \mathbb E A^{0}_{ij} \leq \alpha d + m\cdot \frac dn \leq 2\alpha d.
\end{align*}
Thus, $|X_i\eta_i|$ has sub-gaussian norm at most $\sqrt{(1+2\bar L)\xi\alpha d}$, and the rest of the proof follows from the proof of Lemma 3.3 in \cite{le2017concentration}.
\end{proof}

Combining Lemma \ref{newlem33} with Theorem 3.2 (Grothendieck-Pietsch factorization, sub-matrix version) of \cite{le2017concentration}, we immediately get the following Lemma \ref{newlem34}.

\begin{lemma}[Concentration in spectral norm]
\label{newlem34}
Let $1 \leq m \leq n$ and $\alpha \geq m/n$. Then for $r \geq 1$ the following holds with probability at least $1 - n^{-r} $. Consider a block $I \times J$ of size $m \times m$. Let $I'$ be the set of indices of the rows of $A^{0}_{I\times J}$ whose $l_1$ norm $\leq \alpha d$. Then there exists a subset $J' \subseteq J$ of at least $3m/4$ columns such that 
\begin{align*}
\| (A^{0}-\mathbb E A^{0})_{I'\times J'} \| \leq C \sqrt{\xi \alpha d  r \log(en/m) }.
\end{align*}
\end{lemma}

The following Lemma \ref{newlem35} shows that most rows satisfy the condition for $I'$ in Lemma \ref{newlem34}. The proof of this Lemma involves three steps: first bound the probability of each row having large $l_1$ norm; then bound the number of high $l_1$ norm rows by seeing it as the sum of independent Bernoulli variables; finally apply a union bound for $m, I$ and $J$.
In the first step of the proof we need to deal with the covariates as well as the Poisson edges, which is different from the Erdos-Renyi scenario in \cite{le2017concentration}.

\begin{lemma}[Most rows have $l_1$ norm $\leq O(\alpha d)$]
\label{newlem35}
Let $1 \leq m \leq n$ and $\alpha \geq \sqrt{m/n}$. Then for $r \geq 1$ the following holds with probability at least $1 - n^{-r} $. Consider a block $I \times J$ of size $m \times m$.
Then all but $m/\alpha d$ rows of $A^{0}_{I\times J}$ have $l_1$ norm $\leq 8 \xi^2r\alpha d$.
\end{lemma}
\begin{proof}[Proof of Lemma \ref{newlem35}]
Fix a block $I\times J$, and denote by $d_i$ the $l_1$ norm of the $i$-th row of $A^{0}_{I\times J}$, i.e. $d_i = \sum_{j\in J} A^{0}_{ij}$. We apply a Poisson variable's Chernoff inequality (\cite{vershynin2018high}, p20) to bound $d_i$:
\begin{align*}
\begin{aligned}
\Pr(d_i \geq 8r\alpha d \xi^2 ) =& \Pr\left( \sum_{j\in J} \frac{A_{ij}}{\exp(\bz_{ij}^\top \pg^0) \xi} \geq 8r\alpha d \xi \right)\\
\leq& \Pr\left(\sum_{j\in J} A_{ij} \geq 8r\alpha d \xi\right)\\
\leq& \left( \frac{8r\alpha d \xi }{ e \cdot \frac{md}n\cdot \frac{\sum_{j\in J} \exp(\bz_{ij}^\top\pg^0) }{m} } \right)^{-8r\alpha d \xi}\\
\leq& \left(\frac{2\alpha n}m \right)^{-8r\alpha d},
\end{aligned}
\end{align*}
in which we used $\sum_{j\in J} \exp(\bz_{ij}^\top\pg^0) /m \leq \xi$ and $\xi\geq 1$, and all the inequalities should be understood as first conditioning on and then averaging out the covariates $Z$. The rest of the proof follows from the proof of Lemma 3.5 in \cite{le2017concentration}.
\end{proof}

\begin{lemma}[For block of large row $l_1$ norm, most columns have $O(1)$ $l_1$ norm; and most columns have $O(1)$ $l_0$ norm]
\label{newlem36}
Let $1 \leq m \leq n$ and $\alpha \geq \sqrt{m/n}$. Then for $r \geq 1$ the following holds with probability at least $1 - n^{-r} $. Consider a block $I \times J$ of size $k \times m$ with some $k \leq m/\alpha d$. 
Then all but $m/4$ columns of $A^{0}_{I\times J}$ have $l_1$ norm $\leq 32r\xi^2$. Moreover, all but $m/4$ columns of $A^{0}_{I\times J}$ have $l_0$ norm $\leq 32r\xi$. 
\end{lemma}
\begin{proof}[Proof of Lemma \ref{newlem36} ]
Fix a block $I\times J$, and denote by $d_j$ the $l_1$ norm of the $j$-th column of $A^{0}_{I\times J}$, i.e. $d_j = \sum_{i\in I} A^{0}_{ij}$. We apply a Poisson variable's Chernoff inequality (\cite{vershynin2018high}, p20) to bound $d_j$:
\begin{align*}
\Pr(d_j \geq 32r\xi^2 ) =& \Pr\left( \sum_{i\in I} \frac{A_{ij}}{\exp(\bz_{ij}^\top \pg^0) \xi} \geq 32r \xi \right)\\
\leq& \Pr\left(\sum_{j\in J} A_{ij} \geq 32r \xi\right)\\
\leq& \left( \frac{32r \xi }{ e \cdot \frac{kd}n\cdot \frac{\sum_{i\in I} \exp(\bz_{ij}^\top\pg^0) }{k} } \right)^{-32r\xi}\\
\leq& \left(\frac{10\alpha n}m \right)^{-32r}.
\end{align*}
For the ``moreover'' part, denote by $d_j'$ the $l_0$ norm, i.e. the number of none zero elements, of the $j$-th column of $A^{0}_{I\times J}$. By a similar argument as above, we could bound $d_j'$ by
\begin{align*}
\Pr(d_j' \geq 32r\xi) 
\leq \Pr\left(\sum_{j\in J} A_{ij} \geq 32r \xi\right)
\leq \left(\frac{10\alpha n}m \right)^{-32r}.
\end{align*}
Again all the inequalities should be understood as first conditioning on and then averaging out the covariates $Z$. The rest of the proof follows from the proof of Lemma 3.6 in \cite{le2017concentration}.
\end{proof}

Combining Lemma \ref{newlem34}, \ref{newlem35} and \ref{newlem36} we could get the following Lemma \ref{newlem37}, which gives the decomposition of one block. 

\begin{lemma}[Decomposition of one block]
\label{newlem37}
Let $1 \leq m \leq n$ and $\alpha \geq \sqrt{m/n}$. Then for $r \geq 1$ the following holds with probability at least $1 - 3n^{-r} $. Consider a block $I \times J$ of size $m \times m$. 
Then there exists an exceptional sub-block $I_1 \times J_1$ with dimensions at most $m/2 \times m/2$ such that the remaining part of the block, that is $(I \times J) \setminus (I_1 \times J_1)$, can be decomposed into three classes $\mcN , \mcR \subset (I \setminus I_1) \times J$ and $\mcC \subset I \times (J \setminus J_1)$ so that the following hold:
\begin{enumerate}
  \item The graph concentrates on $\mcN$, i.e. $\|(A^{0}-\mathbb E A^{0})_{\mcN}\| \leq Cr\sqrt{\xi^3\alpha d \log(en/m)}$.
  \item Each row of $A^{0}_{\mcR}$, each column of $A^{0}_{\mcC}$ has $l_1$ norm $\leq 32r\xi^2$.
  \item $\mcR$  intersects at most $m/\alpha d$ columns, and $\mcC$ intersects at most $m/\alpha d$ rows of $I \times J$.
\end{enumerate}
Moreover, the same result also holds for the second property being replaced by ``each row of $A^{0}_{\mcR}$, each column of $A^{0}_{\mcC}$ has $l_0$ norm $\leq 32r\xi$.''
\end{lemma}

Repeatingly apply Lemma \ref{newlem37} to the `exceptional' block in each iteration we would finally arrive at Theorem \ref{newthm26}. And with the decomposition in Theorem \ref{newthm26} we could prove Theorem \ref{newthm21} by bounding the spectral norm separately on $\mcN, \mcR$, and $\mcC$. The proof of Lemma \ref{newlem37}, Theorem \ref{newthm26} and Theorem \ref{newthm21} are all same as in \cite{le2017concentration}, except for some changes in the constants. Thus, we omit the proof for these three results.

\subsection{Proof of Theorem \ref{THM:ecv_consistency}}
\label{sec:supple:proof_chooseK}
First we give a spectral concentration bound on the matrix completion estimator based on the subsampled adjacency matrix under PCABM, which might be of independent interest.
\begin{thm}\label{THM:cv_spectralbd_trueK}[Spectral concentration of $\hat A'_K $ in Algorithm \ref{ecv_k_pcabm}]
Let $A$ be the adjacency matrix generated by the undirected PCABM $(M, B, Z, \pg^0)$.
Assume Conditions \ref{cond:zbd}, \ref{cond:zpd} hold, $\varphi_n \geq C_1 \log(n)$ for some absolute constant $C_1$, and the number of communities $K$ is fixed. 
Further assume each element of $\bar B$ is bounded from above by a constant $C_{\bar B}$, i.e. $\Bm\leq C_{\bar B}$.
Then for any $r>1$ and a constant training proportion $p\in (0,1]$, there exists a constant $\tilde C$ depending on $p,K,\xi, r, C_1$ and $C_{\bar B}$, such that with probability at least $1 - 4n^{-r} - C_\eta \exp(-v_\eta n ) $
(where $\eta = (p\zeta)^{-1}$, $C_{\eta}$, and $v_{\eta}$ are constants in Lemma \ref{lem:eta}),
$\hat A'_{K}$ in Algorithm \ref{ecv_k_pcabm}
 satisfies
\begin{equation}
\label{specbd:A'R_ecv}
\|\hat A'_{K} - P\| \leq \tilde C\sqrt{\varphi_n}.
\end{equation}
\end{thm}

\subsubsection{Proof of the Theorems \ref{THM:cv_spectralbd_trueK} and \ref{THM:ecv_consistency}}
In this subsection we give proofs for the main results Theorem \ref{THM:cv_spectralbd_trueK} and \ref{THM:ecv_consistency} for the ECV algorithm for selecting $K$. In the next subsection we give the proofs for supporting results used in proving those two theorems.
\begin{proof}[Proof of Theorem \ref{THM:cv_spectralbd_trueK}]
The following theorem is an extension from the matrix operator norm concentration in \cite{lei2015consistency} to the case of Poisson edge, covariate adjusted, and with subsampling.
\begin{thm}[Spectral bound of subsampled, adjusted Poisson random matrices]
    \label{THM:SC_Omega}
    Let $A$ be the adjacency matrix generated by PCABM $(M,B,Z,\boldsymbol{\gamma}^0)$, and the adjusted adjacency matrix $A'$ is defined by $A_{ij}' = A_{ij}\exp(-\bz_{ij}^\top \hat\pg) $.
    Let $\Omega$ be an index matrix for a set of node pairs selected independently with probability $p\in(0,1] $, with $\Omega_{ij} = 1$ if the node pair $(i, j)$ are selected and $0$ otherwise.
    Assume Conditions \ref{cond:zbd} and \ref{cond:zpd} hold, and $\varphi_n\geq C_1\log n$. 
    Further assume each element of $\bar B$ is bounded from above by a constant $C_{\bar B}$, i.e. $\Bm\leq C_{\bar B}$.
    Then for any constant $r>0$, there exists a constant $C$ depending only on $\xi,r,C_1$ and $C_{\bar B} $ such that $\|P_\Omega(A'-P)\|\leq C\sqrt{\varphi_n}$
    with probability at least $1 - 3n^{-r} - C_\eta \exp(-v_\eta n ) $
(where $\eta = (p\zeta)^{-1}$, $C_{\eta}$, and $v_{\eta}$ are constants in Lemma \ref{lem:eta}).
\end{thm}
The proof of Theorem \ref{THM:SC_Omega} is given in the next subsection.
As pointed out in \cite{lei2015consistency}, the above operator norm concentration bound is sharper than a matrix Bernstein \citep{tropp2012user} by a $\log n $ factor.

Now we consider the estimator from subsample $\hat A_K'$. Recall $\hat A_K'$ is derived from the truncated SVD: $\hat A'_K = S_H \left(\frac1p P_{\Omega}A', K\right)$. Following the analysis in \cite{li2020network}, we decompose 
\begin{equation}
\label{eq:hatA_K'_Lidecomp}
\begin{aligned}
\| \hat A_K' - P\| \leq& \|\hat A_K' - \frac1p P_{\Omega} A'\| + \| \frac1p P_{\Omega} A' - P \| \\
\leq& 2 \| \frac1p P_{\Omega} A' - P \| \\
\leq& 2  \| \frac1p P_{\Omega} (A' - P) \| + 2 \| \frac1p P_{\Omega} P -P \|   
\end{aligned}
\end{equation}
in which the second line is because $\hat A_K'$ minimizes the $L_2$ distance from $\frac1p P_{\Omega} A'$ to the set of rank $k$ matrices. The first term in the third line of \eqref{eq:hatA_K'_Lidecomp} is controlled by Theorem \ref{THM:SC_Omega}; for the second term in the third line, using the same arguments in \cite{li2020network}, we can bound it by
\begin{align*}
\| \frac1p P_{\Omega} P -P \|  \leq 4C(\delta,C_1)\sqrt{\frac{nK}p}\|P\|_\infty \leq 4C(\delta,C_1)\sqrt{\frac{nK}p} \rho_n C_{\bar B}
\end{align*}
with probability at least $1- n^{-\delta} $ for any $\delta>0$.
Substituting those results into \eqref{eq:hatA_K'_Lidecomp}, we get
\begin{align*}
\|\hat A_K' - P\| \leq \tilde C \sqrt{\varphi_n}
\end{align*}
with probability at least $1 - O(n^{-r}) $,
where the
constant $\tilde C$ depends on $p,K,\xi, r, C_1$ and $C_{\bar B}$.
\end{proof}

\bigskip

\begin{proof}[Proof of Theorem \ref{THM:ecv_consistency}]

We follow the same strategy as the proof of Theorem 3 in \cite{li2020network}. First we have the following two lemmas on the classification error of spectral clustering based on $\hat A'_{\hat K} $. Lemma \ref{lem_ecv_trueKclass} says for the true $K$, the number of errors in the classification of $\hat \be_K^{(m)}$ is at most $O(n\varphi_n^{-1})$ in each community; on the other hand, by Lemma \ref{lem_ecv_hatK<Kclass}, when $\hat K < K$, there are two sets of edges each of cardinality at least order $O(n^2)$ such that their true community labels are different, but their labels in $\hat\be_{\hat K}^{(m)} $ are the same.

\begin{lemma}[Classification error of an ECV split estimate under true $K$]
\label{lem_ecv_trueKclass}
Let $A$ be the adjacency matrix generated by the undirected PCABM $(M, B, Z, \pg^0)$ with $K$ blocks.
Assume the conditions in Theorem \ref{THM:ecv_consistency} hold, with $\varphi_n/\log n \to \infty$.
Let $\hat\be_K^{(m)} $ be the output of spectral clustering on $\hat A'_K$ defined in Algorithm \ref{ecv_k_pcabm} under the true $K$. Then $\hat\be_K^{(m)} $ coincides with the
true $\bc$ on all but $O(n/\varphi_n)$ nodes within each of the $K$ communities (up to a permutation of block labels), 
with probability tending to $1$.
\end{lemma}

\begin{lemma}[Classification error of an ECV split estimate under a $K' <K$]
\label{lem_ecv_hatK<Kclass}
Let $A$ be the adjacency matrix generated by the undirected PCABM $(M, B, Z, \pg^0)$ with $K$ blocks.
Assume the conditions in Theorem \ref{THM:ecv_consistency} hold. Consider one split of ECV. Suppose $\hat\be$ clusters the nodes into $K'$ communities, where $K' <K$. 
Recall $G_k = \{i : c_i = k\}$, and similarly let $\hat G_k$ be communities corresponding to an estimated label vector $\hat\be$.
Define $I_{k_1k_2} =(G_{k_1} \times G_{k_2} )\cap \Omega^c $ and $\hat I_{k_1k_2} =(\hat G_{k_1} \times \hat G_{k_2} )\cap \Omega^c $. 
Then with probability tending to $1$, there must exist $l_1, l_2, l_3 \in [K]$ and $k_1, k_2 \in [K']$ such that

$1$. $|\hat I_{k_1k_2} \cap  I_{l_1l_2}| \geq \tilde c n^2 $.

$2$. $|\hat I_{k_1k_2} \cap  I_{l_1l_2}| \geq \tilde c n^2 $.

$3$. $\bar B_{(l_1,l_2)}\neq\bar B_{(l_1,l_3)} $ where $\bar B_{(ij)} $ denotes the $(i,j)$-th element of $\bar B$.
\end{lemma}

Given the spectral concentration bound in Theorem \ref{THM:cv_spectralbd_trueK}, the proofs of Lemmas \ref{lem_ecv_trueKclass} and \ref{lem_ecv_hatK<Kclass} are essentially the same as Proposition 1 and Lemma 4 in \cite{li2020network}, and are hence omitted.
Using the above two results we upper and lower bound the loss in $\Omega^c$ under true $K$ and $\hat K<K$, respectively in the following two lemmas; the difference in those two bounds would guarantee the consistency of the ECV choose $K$ algorithm.
Recall the loss function we use in Algorithm \ref{ecv_k_pcabm} could be the scaled negative log-likelihood (snll) 
$$ L_1(A,K) = \sum_{(i,j)\in \Omega^c} snll(A_{ij}\nezhg, \hat B_{\hat e_i \hat e_j} )$$ 
or the scaled $L_2$ loss 
$$
L_2(A,K) = \sum_{(i,j)\in \Omega^c} sl_2(A_{ij}\nezhg, \hat B_{\hat e_i \hat e_j} ), $$
where 
\begin{equation}
\begin{aligned}
snll(A_{ij}\nezhg, \hat B_{\hat e_i \hat e_j} )=&
\left[\hat B_{\hat e_i \hat e_j}  - A_{ij} \nezhg \log \hat B_{\hat e_i \hat e_j}\right],\\
sl_2(A_{ij}\nezhg, \hat B_{\hat e_i \hat e_j})=&
\left[A_{ij}\nezhg   -  \hat B_{\hat e_i \hat e_j}\right]^2.
\end{aligned}
\label{eq::loss_def}
\end{equation}

To upper or lower bound the above losses under $\hat K=K$ and $\hat K <K$ circumstances, we compare them to their oracle counterparts defined as
$$ L_{1,0}(A,K) = \sum_{(i,j)\in \Omega^c} snll(A_{ij}\nezhg,  B_{c_i c_j} ) \
\text{ or } \
L_{2,0}(A,K) = \sum_{(i,j)\in \Omega^c} sl_2(A_{ij}\nezhg,  B_{c_i c_j} ).
$$ 

\begin{lemma}[Upper bound of $L_2$ and nll loss under true $K$]
\label{lem_ecv_trueKloss}
Let $A$ be the adjacency matrix generated by the undirected PCABM $(M, B, Z, \pg^0)$ with $K$ blocks.
Assume the conditions in Theorem \ref{THM:ecv_consistency} hold, with $\varphi_n/\log n \to \infty$.
Let $ \hat\be=\hat\be_K^{(m)} $ be the output of spectral clustering on $\hat A'_K$ defined in Algorithm \ref{ecv_k_pcabm} under the true $K$. Then for scaled $L_2$ loss we have
$$
L_2(A,K) - L_{2,0}(A,K) \leq O_P(n\rho_n);
$$
for snll loss, when additionally assuming lower boundedness of $\bar B$ as in Theorem \ref{THM:ecv_consistency}, we have
$$
L_1(A,K) - L_{1,0}(A,K) \leq O_P(n).
$$
\end{lemma}

\begin{lemma}[Upper bound of $L_2$ and nll loss under $K' <K$]
\label{lem_ecv_hatK<Kloss}
Let $A$ be the adjacency matrix generated by the undirected PCABM $(M, B, Z, \pg^0)$ with $K$ blocks.
Assume the conditions in Theorem \ref{THM:ecv_consistency} hold, with $\varphi_n/\log n \to \infty$.
Let $ \hat\be=\hat\be_{K'}^{(m)} $ be the output of spectral clustering on $\hat A'_{K'}$ defined in Algorithm \ref{ecv_k_pcabm} for a $ K' < K$. Then for scaled $L_2$ loss, there exists some constant $c$ s.t.\
$$
L_2(A, K') - L_{2,0}(A,K) \geq cn^2\rho_n^2;
$$
for snll loss, when additionally assuming $\varphi_n/n^{1/3} \to\infty$ and the lower boundedness of $\bar B$ as in Theorem \ref{THM:ecv_consistency},  there exists some constant $c$ s.t.\
$$
L_1(A,K') - L_{1,0}(A,K) \geq cn^2\rho_n^2.
$$
\end{lemma}

By comparing the bounds in Lemma \ref{lem_ecv_trueKloss} and \ref{lem_ecv_hatK<Kloss}, we can see that $\Pr(L(A,K') > L(A,K) ) \to 1$, for scaled $L_2$ loss when $\varphi_n/\log n \to \infty $, and also for snll loss when $\varphi_n/\sqrt n\to \infty $.
\end{proof}

\subsubsection{Proof of Supporting Theorems and Lemmas}
\begin{proof}[Proof of Theorem \ref{THM:SC_Omega}]
It suffices to consider $\tilde A_{ij}:= (P_{\Omega}A)_{ij}= A_{ij}\Omega_{ij}\sim$ Poisson$(\tilde P_{ij}\ezg) $, where $\tilde P_{ij} = \Omega_{ij} P_{ij} $.
For notational simplicity in the proof of this theorem we just write $A_{ij}$ and $P_{ij}$ for $\tilde{A}_{ij}$ and $\tilde P_{ij} $.
The following notations will be used in the proof.
\begin{itemize}
    \item $W=A'-P$ and denote by $w_{ij}$ the $(i,j)$-th entry of $W$.
    \item Let $\mathcal{S}=\{\mathbf{x}\in\mathbb{R}^n:\|\mathbf{x}\|_2\leq 1\}$ be the Euclidean ball of radius $1$. 
\end{itemize}

The proof of Theorem~\ref{THM:SC_Omega} is adapted from \cite{lei2015consistency}, so we will skip the common part and only clarify the modifications. The main idea is to bound 
\begin{equation}\label{eq:bound}
\sup_{\mathbf{x,y}\in \mathcal{S}}|\mathbf{x}^T(A'-P)\mathbf{y}|.
\end{equation}
The proof consists of three steps: discretization, bounding the light pairs, and bounding the heavy pairs. Discretization is to reduce \eqref{eq:bound} to the problem of bounding the supremum of $\mathbf{x}^T(A'-P)\mathbf{y}$ for $\mathbf{x},\mathbf{y}$ in a finite set of grid points in $\mathcal{S}$. Then we divide $\mathbf{x},\mathbf{y}$ into light and heavy pairs, and bound them respectively.  

\paragraph{Discretization}
For fixed $\delta_n\in(0,1)$, define 
\[\mathcal{T}=\{\mathbf{x}=(x_1,\cdots,x_n)\in \mathcal{S}:\sqrt{n}x_i/\delta_n\in\mathbb{Z},\forall i\},\]
where $\mathbb{Z}$ stands for the set of integers. The following lemma is the same as Lemma B.1 in \cite{lei2015consistency} and we will skip the proof.

\begin{lemma}
    $\mathcal{S}_{1-\delta_n}\subset convhull(\mathcal{T})$. As a consequence, for all $W\in\mathbb{R}^{n\times n}$,
    \[\|W\|\leq(1-\delta_n)^{-2}\sup_{\mathbf{x},\mathbf{y}\in \mathcal{T}}|\mathbf{x}^TW\mathbf{y}|.\]
\end{lemma}

For any $\mathbf{x},\mathbf{y}\in\mathcal{T}$, we have 
$$\mathbf{x}^T(A'-P)\mathbf{y}=\sum_{1\leq i,j\leq n}x_iy_j(A'_{ij}-P_{ij}).$$
We only need bound the above quantity now. We divide $(x_i,y_j)$ into \emph{light pairs} $\mathcal{L}=\{(i,j):|x_iy_j|\leq\sqrt{\varphi_n}/n\}$ and \emph{heavy pairs} $\mathcal{H}=\{(i,j):|x_iy_j|>\sqrt{\varphi_n}/n\}$. We will show that the tail for light pairs can be bounded exponentially while heavy pairs have a heavier tail. Thus, the rate of the latter one dominates.

\paragraph{Bounding the light pairs}

\begin{lemma}\label{lem:light}
    Under estimation error condition, for $c>0$, there exist constants $C_c, v_c>0$ s.t. 
    \[P\left(\sup_{x,y\in \mathcal{T}}\left|\sum_{(i,j)\in\mathcal{L}(x,y)}x_iy_jw_{ij}\right|\geq c\sqrt{\varphi_n}\right)\leq C_c\exp\left[-\left(v_c-\log\left(\frac{7}{\delta}\right)\right)n\right].\]
\end{lemma}

\begin{proof}[Proof of Lemma \ref{lem:light}]
Define $w'_{ij}=A_{ij}e^{-\mathbf{z}_{ij}^T\boldsymbol{\gamma}^0}-P_{ij}$ and $\delta_{ij}=A_{ij}e^{-\mathbf{z}_{ij}^T\hat{\boldsymbol{\gamma}}}-A_{ij}e^{-\mathbf{z}_{ij}^T\boldsymbol{\gamma}^0}$, then  $w_{ij}=w'_{ij}+\delta_{ij}$, and notice that
    $$\text{Pr}(\sum w_{ij}>2t)\leq \text{Pr}(\sum w'_{ij}>t)+\text{Pr}(\sum \delta_{ij}>t),$$
    so we could bound two parts respectively. Also, denoting $u_{ij} = x_i y_j + x_j y_i $, we keep in mind that
    $$\sum_{i<j}u_{ij}^2\leq\sum_{i<j}2(x_i^2y_j^2+x_j^2y_i^2)\leq2\sum_{1\leq i,j\leq n}x_i^2y_j^2=2\|x\|_2^2\|y\|_2^2\leq2.$$

\emph{Step 1 : Bound $w'_{ij}$.}

Let $\beta_l \leq \ezg \leq \beta_u $. (One could take $\beta_u=\beta_l^{-1} = \xi $.)
Since
\begin{align*}
    \sum_{i<j}\mathbb{E}[(w'_{ij}u_{ij})^2|\mathbf{z}_{ij}]=&\sum_{i<j}u_{ij}^2\lambda_{ij}e^{-2\boldsymbol{z}_{ij}^T\boldsymbol{\gamma}^0}=\sum_{i<j}u_{ij}^2P_{ij}e^{-\boldsymbol{z}_{ij}^T\boldsymbol{\gamma}^0}\\
    \leq&\rho_n\beta_l^{-1}\|\bar{B}\|_{\max}\sum_{i<j}u^2_{ij}\leq2\beta_l^{-1}\|\bar{B}\|_{\max}\rho_n
\end{align*}
define $M=2\beta_l^{-1}\|\bar{B}\|_{\max}\rho_n$, $L=2\bar{L}\sqrt{\varphi_n}(n\beta_l)^{-1}$ and $x=c\sqrt{\varphi_n}$, we could applying Lemma~\ref{lem:bernin} to $u_{ij}w'_{ij}$ to get
\begin{align*}
\text{Pr}(\sum_{i<j}w'_{ij}u_{ij}\geq c\sqrt{\varphi_n})\leq&\exp\left(-\frac{c^2\varphi_n}{4c\bar{L}\varphi_n(n\beta_l)^{-1}+4\rho_n\beta_l^{-1}\|\bar{B}\|_{\max}}\right)\\
=&\exp\left(-\frac{c^2\beta_ln}{4c\bar{L}+4\|\bar{B}\|_{\max}}\right)
\end{align*}

\emph{Step 2 : Bound $\delta_{ij}$.}

We consider two cases $\|\boldsymbol{\gamma}^0-\hat{\boldsymbol{\gamma}}\|_{\infty}>\eta/\sqrt{n\rho_n}$ and $\|\boldsymbol{\gamma}^0-\hat{\boldsymbol{\gamma}}\|_{\infty}\leq\eta/\sqrt{n\rho_n}$ separately. Conditioning on the second case, by choosing $\eta<(p\zeta)^{-1}$, we have
\begin{align*}
u_{ij}\delta_{ij}\leq&|u_{ij}||[A_{ij}e^{-\mathbf{z}_{ij}^T\boldsymbol{\gamma}^0}(e^{\mathbf{z}_{ij}^T(\boldsymbol{\gamma}^0-\hat{\boldsymbol{\gamma}})}-1)]|<2|u_{ij}||[A_{ij}e^{-\mathbf{z}_{ij}^T\boldsymbol{\gamma}^0}\mathbf{z}_{ij}^T(\boldsymbol{\gamma}^0-\hat{\boldsymbol{\gamma}})]|\\
<&2\zeta\eta p|u_{ij}|A_{ij}/(\beta_l\sqrt{n\rho_n})<2|u_{ij}|A_{ij}/(\beta_l\sqrt{n\rho_n})
\end{align*}
The first inequality is due to $|e^t-1|<2|t|$ when $|t|<1$. Define $M=2\beta_u\|\bar{B}\|_{\max}\rho_n\geq\sum_{i<j}u_{ij}^2\lambda_{ij}=\text{var}(\sum_{i<j}u_{ij}A_{ij}|Z)$, $L=2\bar{L}\sqrt{\varphi_n}/n$ and $x=c\sqrt{\varphi_nn\rho_n}$, by Lemma~\ref{lem:bernin}, 
\begin{align*}
&\text{Pr}(\sum_{i<j}|u_{ij}|(A_{ij}-P_{ij})>c\sqrt{\varphi_nn\rho_n})\\
\leq&\exp\left(-\frac{c^2\varphi_nn\rho_n}{4\|\bar{B}\|_{\max}\beta_u\rho_n+4c\sqrt{n\rho_n}\varphi_n\bar{L}/n}\right)\\
=&\exp\left(-\frac{c^2n\sqrt{n\rho_n}}{4\|\bar{B}\|_{\max}\beta_u/\sqrt{n\rho_n}+4c\bar{L}}\right)\\
\leq&\exp\left(-\frac{c^2n\sqrt{n\rho_n}}{4\|\bar{B}\|_{\max}\beta_u+4c\bar{L}}\right)
\end{align*}

Because
$$\sum_{i<j}|u_{ij}|P_{ij}\leq\rho_n\|\bar{B}\|_{\max}\beta_u\sum_{i<j}|u_{ij}|\leq\sqrt{2N_n}\rho_n\|\bar{B}\|_{\max}\beta_u\leq n\rho_n\|\bar{B}\|_{\max}\beta_u,$$
we have
\begin{align*}
    \exp\left(-\frac{c^2n\sqrt{n\rho_n}}{4\|\bar{B}\|_{\max}\beta_u+4c\bar{L}}\right)\geq&\text{Pr}(\sum_{i<j}|u_{ij}|(A_{ij}-P_{ij})>c\sqrt{\varphi_nn\rho_n})\\
    \geq&\text{Pr}(\sum_{i<j}|u_{ij}|A_{ij}>(c+\|\bar{B}\|_{\max}\beta_u)\varphi_n),
\end{align*}
which is equivalent to $\text{Pr}(\sum_{i<j}u_{ij}A_{ij}>c\varphi_n)\leq\exp(-C_cn\sqrt{\varphi_n})$, where $C_c$ is constant.

Thus, for $\eta<(p\zeta)^{-1}$,
\begin{align*}
&\text{Pr}\left(\sum_{i<j}\delta_{ij}u_{ij}>c\sqrt{\varphi_n}\right)\\
=&\text{Pr}\left(\sum_{i<j}\delta_{ij}u_{ij}>c\sqrt{\varphi_n}\middle\vert\|\boldsymbol{\gamma}^0-\hat{\boldsymbol{\gamma}}\|_{\infty}\leq\frac{\eta}{\sqrt{n\rho_n}}\right)\text{Pr}\left(\|\boldsymbol{\gamma}^0-\hat{\boldsymbol{\gamma}}\|_{\infty}\leq\frac{\eta}{\sqrt{n\rho_n}}\right)\\
&+\text{Pr}\left(\sum_{i<j}\delta_{ij}u_{ij}>c\sqrt{\varphi_n}\middle\vert\|\boldsymbol{\gamma}^0-\hat{\boldsymbol{\gamma}}\|_{\infty}>\frac{\eta}{\sqrt{n\rho_n}}\right)\text{Pr}\left(\|\boldsymbol{\gamma}^0-\hat{\boldsymbol{\gamma}}\|_{\infty}>\frac{\eta}{\sqrt{n\rho_n}}\right)\\
\leq&\text{Pr}\left(\sum_{i<j}\delta_{ij}u_{ij}>c\sqrt{\varphi_n}\middle\vert\|\boldsymbol{\gamma}^0-\hat{\boldsymbol{\gamma}}\|_{\infty}\leq\frac{\eta}{\sqrt{n\rho_n}}\right)+C_\eta\exp(-v_\eta n)\\
\leq&\text{Pr}\left(\sum_{ij}|u_{ij}|A_{ij}>c\beta_l\varphi_n/2\right)+C_\eta\exp(-v_\eta n)\\
\leq&\exp(-C_cn\sqrt{n\rho_n})+C_\eta\exp(-v_\eta n)\\
\leq& C_{c,\eta}\exp(-v_{c,\eta}n),
\end{align*}
where $C_{c,\eta}$ and $v_{c,\eta}$ are two constants determined by $c$ and $\eta$.

By a standard volume argument we have $|T|\leq e^{n\log(7/\delta)}$ (see Claim 2.9 of \cite{feige2005spectral}), so the desired result follows from the union bound.
\end{proof}

\paragraph{Bounding the heavy pairs}
By the same argument as in Section 4 of Supplement to \cite{lei2015consistency}, to bound $\sup_{\mathbf{x},\mathbf{y}\in \mathcal{T}}|\sum_{(i,j)\in\mathcal{H}(x,y)}x_iy_jw_{ij}|$, it suffices to show
$$\sum\limits_{(i,j)\in\mathcal{H}}x_iy_jA'_{ij}=O(\sqrt{\varphi_n})$$
with high probability. Since $A'_{ij}=A_{ij}e^{-\mathbf{z}_            {ij}^T\boldsymbol{\gamma}^0}-(A_{ij}e^{-\mathbf{z}_            {ij}^T\boldsymbol{\gamma}^0}-A_{ij}e^{-\mathbf{z}_            {ij}^T\hat{\boldsymbol{\gamma}}})\leq A_{ij}\beta_l^{-1}(1+o_p(1))$, we only need to show 
$$\sum\limits_{(i,j)\in\mathcal{H}}x_iy_jA_{ij}=O(\sqrt{\varphi_n})$$
with high probability. 
The above relationship can be shown very similarly to the proof of Lemma 4.3 in the Supplement to \cite{lei2015consistency}. Particularly, to obtain a same result as their Lemma 4.1, we apply the Bernstein bound of Lemma \ref{lem:bern}; to prove a same result as their Lemma 4.2, we apply the Chernoff bound for Poisson tails in Exercise 2.3.3 of \cite{vershynin2018high} in the place of Corollary A.1.10 of \cite{alon2016probabilistic}; the rest of the proof of their Lemma 4.3 is exactly the same in our setting.
Thus we get the following lemma.

\begin{lemma}(Heavy pair bound). For any given $r>0$, there exists a constant $C_r$ such that 
    $$\sup_{x,y\in T}|\sum\limits_{(i,j)\in\mathcal{H}}x_iy_jw_{ij}|\leq C_r\sqrt{\varphi_n}$$
    with probability at least $1-2n^{-r}$.
\end{lemma}
\end{proof}


\begin{proof}[Proof of Lemma \ref{lem_ecv_trueKloss}]
Same as in the proof of Theorem 3 in \cite{li2020network}, we consider the following sets of edges:
\begin{align*}
T_{k_1,k_2,l_1,l_2} =& \{(i,j)\in\Omega^c: c_i = l_1, \hat e_i = k_1, c_j = l_2, \hat e_j = k_2 \},\\
U_{k_1,k_2,l_1,l_2} =& \{(i,j)\in\Omega: c_i = l_1, \hat e_i = k_1, c_j = l_2, \hat e_j = k_2 \},
\end{align*}
$T_{k_1,k_2,\cdot,\cdot} = \cup_{l_1,l_2} T_{k_1,k_2,l_1,l_2} $, and $T_{\cdot,\cdot,l_1,l_2},U_{k_1,k_2,\cdot,\cdot}$ which are defined similarly. By arguments in \cite{li2020network}, we have $|U_{k_1,k_2,\cdot,\cdot}|\geq cn^2 $ for some constant $c$. 
Define 
$$EU_{k_1,k_2,l_1,l_2} := \sum_{(i,j)\in U_{k_1,k_2,l_1,l_2}} \exp(\bz_{ij}^\top \pg^0), \quad
\hat EU_{k_1,k_2,l_1,l_2} := \sum_{(i,j)\in U_{k_1,k_2,l_1,l_2}} \exp(\bz_{ij}^\top \hat\pg), $$
and similarly $EU_{k_1,k_2,\cdot,\cdot}, ET_{k_1,k_2,l_1,l_2}, ET_{k_1,k_2,\cdot,\cdot}$, and their hat version. 
Then we have for $k_1\neq k_2$,
\begin{align*}
|\hat B_{k_1k_2} - B_{k_1k_2}| =& \left|\frac{\sum_{U_{k_1,k_2,\cdot,\cdot}} A_{ij} }{\hat EU_{k_1,k_2,\cdot,\cdot}} - B_{k_1k_2}  \right| \\
\leq&\frac{\hat EU_{k_1,k_2,k_1,k_2}}{\hat EU_{k_1,k_2,\cdot,\cdot}} \left|\frac{\sum_{U_{k_1,k_2,k_1,k_2}} A_{ij} }{EU_{k_1,k_2,k_1,k_2}} - \frac{\sum_{U_{k_1,k_2,k_1,k_2}} A_{ij} }{\hat EU_{k_1,k_2,k_1,k_2}} \right|  \\
 &+\frac{\hat EU_{k_1,k_2,k_1,k_2}}{\hat EU_{k_1,k_2,\cdot,\cdot}} \left|\frac{\sum_{U_{k_1,k_2,k_1,k_2}} A_{ij} }{EU_{k_1,k_2,k_1,k_2}} - B_{k_1k_2}  \right|
+\left|\left(1-\frac{\hat EU_{k_1, k_2, k_1, k_2}}{\hat EU_{k_1, k_2, \cdot, \cdot}}\right) B_{k_1 k_2}\right| \\
&+ \frac{\hat EU_{k_1, k_2, \cdot, \cdot} - \hat EU_{k_1, k_2, k_1, k_2}}{\hat EU_{k_1, k_2, \cdot , \cdot}}\left|\frac{\sum_{U_{k_1, k_2, \cdot , \cdot} \slash U_{k_1, k_2, k_1, k_2}} A_{i j}}{\hat EU_{k_1, k_2, \cdot , \cdot} -\hat EU_{k_1, k_2, k_1, k_2}}\right| \\
\leq& O_P(\frac{\sqrt{\rho_n}}{n^{2.5}})+ O_P(\sqrt{\frac{\rho_n}{n^2} } ) + O_P(\frac{1}{\varphi_n} )O_P(\rho_n) + O_P(\frac{1}{\varphi_n} )O_P(\rho_n) = O_P(\frac1n),
\end{align*}
in which the $O_P(\sqrt{\frac{\rho_n}{n^2}})$ bound is due to a Bernstein inequality, and in each bound of $\hat EU$ we use Condition \ref{cond:zbd} and Lemma \ref{lem:eta}. For $k_1=k_2, \hat{B}_{k_1 k_2}$ is the average over $U_{k_1, k_2, \cdot, \cdot} \cap\{(i, j): i<j\}$ which makes both its denominator and nominator half of those in the above calculation, and the same concentration holds.

Now we decompose (for a general $L$ that could be the snll loss $L_1$ or the scaled $L_2$ loss $L_2$)
\begin{align*}
 L(A, K)-L_{0}(A, K)
=&\sum_{k_1, k_2} \sum_{(i, j) \in T_{k_1, k_2, k_1, k_2}}\left[\ell\left(A_{i j}\nezhg, \hat{B}_{k_1 k_2} \right)-\ell\left(A_{i j}\nezhg, B_{k_1 k_2}\right)\right] \\
&+\sum_{\left(k_1, k_2\right) \neq\left(l_1, l_2\right)} \sum_{(i, j) \in T_{k_1, k_2, l_1, l_2}}\left[\ell\left(A_{i j}\nezhg, \hat{B}_{k_1 k_2}\right)-\ell\left(A_{i j}\nezhg, B_{l_1 l_2}\right)\right] \\
=&:\mathcal{I}+\mathcal{II},
\end{align*}
in which the loss $\ell$ could be snll or scaled $L_2$ loss defined in \eqref{eq::loss_def}.

For the (scaled) $L_2 \operatorname{loss} \ell(x,y)=(x-y)^2$, we have
$$
\left|\ell\left(x, y_1\right)-\ell\left(x, y_2\right)\right| \leq 2\left(|x|+\left|y_2\right|+\left|y_1-y_2\right|\right)\left|y_1-y_2\right|.
$$
Thus, we can bound $\mathcal{I}$ and $\mathcal{II}$ respectively by
\begin{align*}
|\mathcal{I}| \leq& 2 \sum_{k_1,k_2 } \sum_{(i,j)\in T_{k_1,k_2,k_1,k_2}}  |\hat B_{k_1k_2} - B_{k_1k_2}| (A_{ij}\nezhg + B_{k_1k_2} + |\hat B_{k_1k_2} - B_{k_1k_2}| )\\ \leq& O_P(n^2) O_P(\frac1n) O_P(\rho_n) = O_P(n\rho_n),\\
|\mathcal{II}| \leq& 2 \sum_{(k_1,k_2)\neq (l_1,l_2) } \sum_{(i,j)\in T_{k_1,k_2,l_1,l_2}}(\hat B_{k_1k_2} + B_{l_1l_2}) (A_{ij}\nezhg + B_{l_1l_2} + \hat B_{k_1k_2} )\\\leq& O_P(\frac{n}{\rho_n}) O_P(\rho_n^2) = O_P(n\rho_n ),
\end{align*}
where we use the bound on $\hat\pg$ from Condition \ref{cond:zbd} and Lemma \ref{lem:eta}. 
Hence, for $L_2$ loss, we have $L_2(A,K) - L_{2,0}(A,K) \leq O_P(n\rho_n) $.

For snll loss $\ell(x,y) = y - x\log y $, we similarly have
$$
\left|\ell\left(x, y_1\right)-\ell\left(x, y_2\right)\right| \leq \left|y_1-y_2\right| + x \frac{|y_1 - y_2| }{\min(y_1,y_2) }.
$$ 
Same as in \cite{li2020network}, we assume $|\hat B_{k_1k_2} - B_{k_1k_2}| \leq B_{k_1k_2}/2$, which can be seen from Lemma \ref{lem:eta} when $n$ is sufficiently large. Then we could bound
\begin{align*}
|\mathcal{I}| \leq& \sum_{k_1, k_2} \sum_{(i, j) \in T_{k_1, k_2, k_1, k_2}} \left[ 1 + 2A_{ij}\nezhg/B_{k_1k_2} \right] |\hat B_{k_1k_2} - B_{k_1k_2} | \leq O_P(n^2)O_P(\frac1n) = O_P(n);\\
|\mathcal{II}| \leq& \sum_{\left(k_1, k_2\right) \neq\left(l_1, l_2\right)} \sum_{(i, j) \in T_{k_1, k_2, l_1, l_2}} \left[ 1 + 2A_{ij}\nezhg/c_{\bar B} \right] (|\hat B_{k_1k_2}| + |B_{l_1kl_2} |) \leq O_P(\frac{n}{\rho_n}) O_P(\rho_n) =  O_P(n).
\end{align*}
Hence, for the snll loss, we have $L_1(A,K) - L_{1,0}(A,K) \leq O_P(n) $.
\end{proof}

\begin{proof}[Proof of Lemma \ref{lem_ecv_hatK<Kloss}]
Without loss of generality, assume the $k_1, k_2$ and $l_1, l_2, l_3$ in Lemma \ref{lem_ecv_hatK<Kclass} are 1,2 and 3,4,5 respectively. (But keep in mind that one or two among 3,4,5 could be the same as 1,2.) We have
\begin{equation}
\begin{aligned}
& L\left(A, K^{\prime}\right)-L_0(A, K)= \sum_{k_1, k_2, l_1, l_2} \sum_{(i, j) \in T_{k_1, k_2, l_1, l_2}}\left[\ell\left(A_{i j}\nezhg, \hat{P}_{i j}\right)-\ell\left(A_{i j}\nezhg, B_{l_1 l_2}\right)\right] \\
=& \sum_{(i, j) \in T_{1,2,3,4}}\left[\ell\left(A_{i j}\nezhg, \hat{B}_{12}\right)-\ell\left(A_{i j}\nezhg, B_{34}\right)\right]+\sum_{(i, j) \in T_{1,2,3,5}}\left[\ell\left(A_{i j}\nezhg, \hat{B}_{12}\right)-\ell\left(A_{i j}\nezhg, B_{35}\right)\right] \\
&+\sum_{\left(k_1, k_2, l_1, l_2\right) \notin\{(1,2,3,4),(1,2,3,5)\}} \sum_{(i, j) \in T_{k_1,k_2,l_1,l_2}} \left[\ell\left(A_{i j}\nezhg, \hat{P}_{i j}\right)-\ell\left(A_{i j}\nezhg, B_{l_1 l_2}\right)\right].
\end{aligned}
\label{eq:decomp_lem_losshatK<K_1}
\end{equation}
For both the scaled $L_2$ and the snll losses and any index set $T$, the function of the form
$$
f(p)=\sum_{(i, j) \in T} \ell\left(A_{i j}\nezhg, p\right)
$$
is always minimized when $p=\frac{ \sum_{(i, j) \in T} A_{i j}\nezhg }{ |T|  } $. Applying this in the above decomposition \eqref{eq:decomp_lem_losshatK<K_1}, we have
$$
\begin{aligned}
&L\left(A, K^{\prime}\right)-L_0(A, K)\\  \geq& \sum_{(i, j) \in T_{1,2,3,4}}\left[\ell\left(A_{i j}\nezhg, \hat{p}\right)-\ell\left(A_{i j}\nezhg, B_{34}\right)\right]+\sum_{(i, j) \in T_{1,2,3,5}}\left[\ell\left(A_{i j}\nezhg, \hat{p}\right)-\ell\left(A_{i j}\nezhg, B_{35}\right)\right] \\
&+\sum_{\left(k_1, k_2, l_1, l_2\right) \notin\{(1,2,3,4),(1,2,3,5)\}}\sum_{(i, j) \in T_{k_1, k_2, l_1, l_2}}\left[\ell\left(A_{i j}\nezhg, \hat{p}_{k_1, k_2, l_1, l_2}\right)-\ell\left(A_{i j}\nezhg, B_{l_1 l_2}\right)\right]\\
:=&\mathcal{I I I}+\mathcal{I} \mathcal{V}+\mathcal{V},
\end{aligned}
$$
where $\hat p_{k_1,k_2,l_1,l_2}:= { \sum_{(i, j) \in T_{k_1,k_2,l_1,l_2}} A_{i j}\nezhg }/{  |T_{k_1,k_2,l_1,l_2}|  }  $, and $\hat{p}$ denotes\\ ${ \sum_{(i, j) \in T_{1,2,3,4} \cup T_{1,2,3,5}} A_{i j}\nezhg } / { |T_{1,2,3,4} \cup T_{1,2,3,5}| }$. Note that $\hat{p}=t \hat{p}_1+(1-t) \hat{p}_2$, where $\hat{p}_1=\hat{p}_{1,2,3,4}$ and $\hat{p}_2=\hat{p}_{1,2,3,5}$ and $t=\frac{\left|T_{1,2,3,4}\right|}{\left|T_{1,2,3,4 }\right|+\left|T_{1,2,3,5}\right|}$. Also let $p_1=B_{34}$ and $p_2=B_{35}$. 

We first bound the term $\mathcal{III}$, in the same fashion as in \cite{li2020network}, but taking the covariate adjusting into account.
Define $f(\mathbf{x}, p)=\sum_{(i, j) \in T_{1,2,3,4}} \ell\left(A_{i j}\nezhg, p\right)$, where $\mathbf{x}=\left\{A_{i j}\nezhg\right\}_{T_{1,2,3,4}}$.
Then $\mathcal{III}$ could be written as
$$
\mathcal{I I I}=f\left(\mathbf{x}, t \hat{p}_1+(1-t) \hat{p}_2\right)-f\left(\mathbf{x}, p_1\right) .
$$
By the property of strong convexity \citep{boyd2004convex}, for both losses we have
$$
f\left(\mathbf{x}, t \hat{p}_1+(1-t) \hat{p}_2\right) \geq f\left(\mathbf{x}, \hat{p}_1\right)+\frac{m}{2}(1-t)^2\left|\hat{p}_1-\hat{p}_2\right|^2,
$$
in which $m =2 |T_{1,2,3,4}| $ for the scaled $L_2$ loss, and $m$ can be chosen as $\sum_{(i,j)\in T_{1,2,3,4}} A_{ij}\nezhg / C_{\bar B}^2 $ for the snll loss.
Now term $\mathcal{III}$ can be bounded by
\begin{equation}
\label{eq:III_in_hatK<Kloss}
f\left(\mathbf{x}, t \hat{p}_1+(1-t) \hat{p}_2\right)-f\left(\mathbf{x}, p_1\right) \geq
\frac{m}2(1-t)^2 |\hat p_1 - \hat p_2|^2 + f(\mathbf{x},\hat p_1) - f(\mathbf{x}, p_1).
\end{equation}
As is argued in \cite{li2020network}, $|\hat p_1 - \hat p_2|$ is lower bounded by $c_{K'}\rho_n$ for some constant $c_{K'}$, since $|\hat p_1 - p_1| $ and $|\hat p_2 - p_2| $ are bounded by $O_P(\sqrt{\frac{\rho_n}{n^2}})$ by a Bernstein inequality (note that $\hat p_1 = \hat p_{1,2,3,4}$ and $p_1 = B_{34}$).
Moreover, the $f(\mathbf{x},\hat p_1) - f(\mathbf{x}, p_1)$ could be bounded by 
\begin{align*}
|f(\mathbf{x},\hat p_1) - f(\mathbf{x}, p_1)|\leq& \sum_{(i,j)\in T_{1,2,3,4}} 2|\hat p_1 -p_1 | (A_{ij} + p_1 + |\hat p_1 - p_1| )\\ \leq& O_P(n^2) O_P(\sqrt{\frac{\rho_n}{n^2}} ) O_P(\rho_n) = O_P(n\rho_n^{3/2} )
\end{align*}
for scaled $L_2$ loss, and
\begin{align*}
|f(\mathbf{x},\hat p_1) - f(\mathbf{x}, p_1)|\leq& \sum_{(i,j)\in T_{1,2,3,4}} 2|\hat p_1 -p_1 | (1+2 A_{ij}\nezhg /c_{\bar B} )\\ \leq& O_P(n^2) O_P(\sqrt{\frac{\rho_n}{n^2}} ) O_P(1) = O_P(n\rho_n^{1/2} )
\end{align*}
for the snll loss. To sum up, the $m(1-t)^2|\hat p_1-\hat p_2|^2/2 $ term in \eqref{eq:III_in_hatK<Kloss} dominates the rate, and we get there exists some constant $c$ s.t.\
$$
\mathcal{III} \geq cn^2\rho_n^2
$$
for scaled $L_2$ loss, and also for snll loss under the condition that $\rho_n^{-1} = o(n^{2/3})$ (which is weaker than the $\varphi_n/\sqrt n\to\infty$ condition in our lemma and theorem). The term $\mathcal{IV}$ is bounded in exactly the same way as $\mathcal{III}$. 

Next, consider term $\mathcal{V}$. For the scaled $L_2$ loss, same as in \cite{li2020network} we obtain $\mathcal{V}\geq -O_P(\rho_n) $. For the snll loss,
\begin{align*}
\mathcal{V} \geq - \sum_{\left(k_1, k_2, l_1, l_2\right) \notin\{(1,2,3,4),(1,2,3,5)\}}\sum_{(i, j) \in T_{k_1, k_2, l_1, l_2}}|\hat p_{k_1k_2l_1l_2} - B_{l_1l_2} | \left[ 1 + 2 A_{ij}\nezhg / c_{\bar B}  \right] \geq -O_P(\sqrt{\rho_n n^2}).
\end{align*}
Combining all the bounds on terms $\mathcal{III},\mathcal{IV}$ and $\mathcal{V}$, we have  there exists some constant $c$ s.t.\
$$
L(A,K')\geq L_0(A,K) + c n^2\rho_n^2
$$
for scaled $L_2$ loss, and also for snll loss under the condition that $\varphi_n/ n^{1/3}\to\infty $.
\end{proof}

\section{Feature Selection}
\label{sec:ecv_z}
In the covariate-adjusted model, pairwise covariates $Z$ and the class labels $\bc$ are independent. On the contrary, in the covariates-confounding model, the covariates distribution is governed by the community labels. An interesting question to ask is what will happen if the covariates used in 
fitting 
the covariates-adjusted model are correlated with the community information.

Consider the following example: $n=500, \bar{B} = \bigl(\begin{smallmatrix} 2 & 1 \\ 1 & 2 \end{smallmatrix}\bigr), \rho_n = 4\log n /n $. A PCABM network is generated from $A_{ij} \sim \text{Poisson}(B_{c_ic_j} \exp(z_{ij}\gamma^0))$, where we have one pairwise covariate $z_{ij} \sim \text{Poisson}(0.09)$ and $\gamma^0=2$. However, when fitting the model, a ``false'' covariate $Z'$ is also included where $z_{ij}'\sim \text{Poisson}(0.09) + 0.6(\bar{B}_{c_ic_j}-1.5)r(1-r^2)^{-1/2}$, which makes the Pearson correlation between $[z_{ij}', i<j]$ and $[P_{ij}, i<j]$ to be $r$. We consider evaluating the community detection performance of fitting PCABM under three scenarios: (1) using only the true covariate $z_{ij}$, (2) using only the false covariate $z_{ij}'$, (3) using both covariates. We vary the correlation $r$ from 0 to 1 to see how it will impact the community detection accuracy.
As shown in Figure~\ref{fig:cov}, as the correlation increases, when fitting the model with $Z'$ or with $\{Z,Z'\} $, community detection performance becomes worse. Why does this happen? As far as we could understand, the reason is that when the false covariate is correlated with the matrix $P$, it will contribute substantially to fitting the model. When estimating $\pg$ the MLE mistakenly recognized the effect of $B_{c_ic_j}$ as the effect of $\exp(z'_{ij}\gamma)$, so that the $\pg$ estimate is very biased, and as a result what we get after adjusting for such a covariate will contain less community information. This is demonstrated in Table \ref{tab:gamma_corr}: when fitting PCABM with both covariates, the coefficient of $Z'$ is more and more biased as the correlation between $Z'$ and $B$ grows.
 In another word, the MLE cannot distinguish the effects of $B_{c_ic_j}$ and $\exp(z'_{ij}\gamma)$. 

\begin{figure}[h]
    \centering
  \includegraphics[width=.5\linewidth]{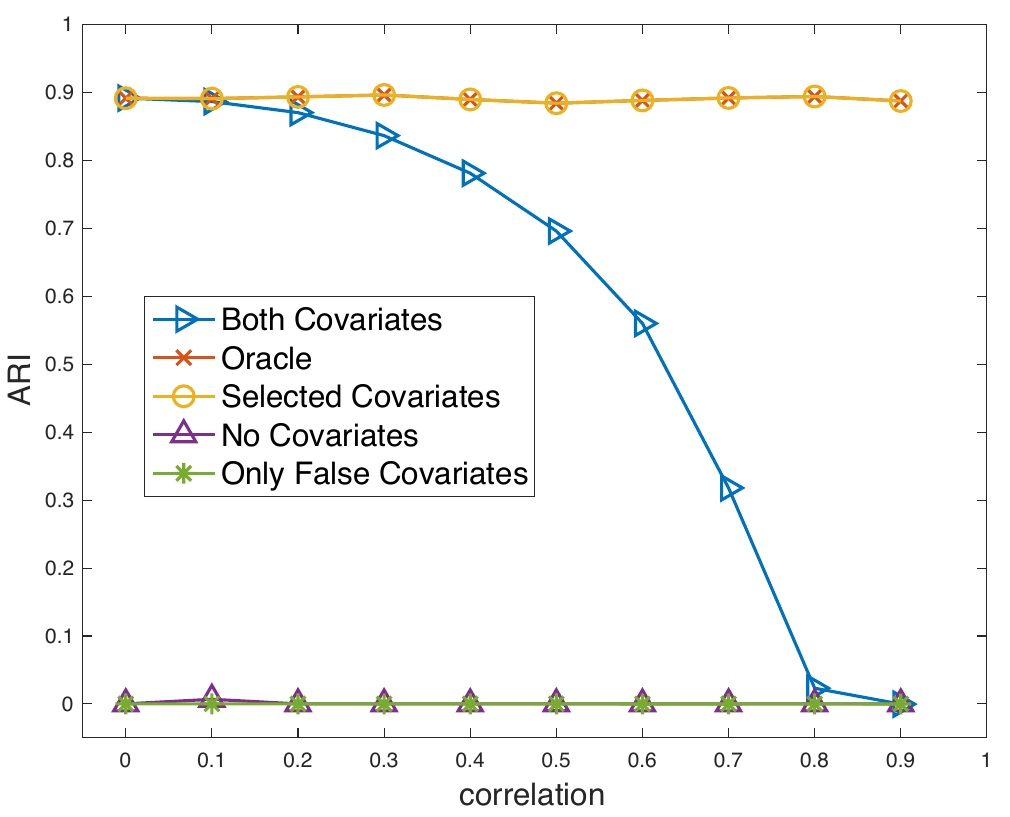}  
    \caption{Simulation results for adding covariate of different correlation with community structure.}
    \label{fig:cov}
\end{figure}

\begin{table}[htbp]
    \centering
    \caption{Estimation of $\hat{\boldsymbol{\gamma}}$ when including both covariates $\{Z,Z'\}$. Averaged over 100 replicates.}
    {\small
    \begin{tabular}{c|cccccccccc}
        \hline
        $r$&0&0.1&0.2&0.3&0.4&0.5&0.6&0.7&0.8&0.9\\
        \hline
        $\gamma(Z)$&2.001&2.001&2.000&2.002&2.002&2.001&2.000&1.999&2.001&2.003\\
        $\gamma(Z')$&0.000&0.108&0.198&0.282&0.351&0.412&0.456&0.485&0.489&0.424\\
        \hline
    \end{tabular}
    }
    \label{tab:gamma_corr}
\end{table}

However, one would imagine that the prediction of an inaccurate $(\pmb\gamma,\be)$ estimate on a test set will not be as satisfactory. In particular, the out-of-sample likelihood of $(\pmb\gamma,\be)$ estimation under the $\{Z,Z'\}$ model should be smaller than its counterpart of the model that only has the true covariate ${Z}$. 
The edge cross-validation model selection procedure introduced above is a very suitable framework to compare out-of-sample likelihoods when fitting different subsets of covariates. Therefore, we expect the ECV approach could select the best subset of covariates, and pick out the confounding $Z'$ in this case. An ECV algorithm for selecting covariates is given in Algorithm  \ref{ecv_z_pcabm}. It basically uses a forward selection method, with a stopping criterion defined by the convergence of out-of-sample likelihood.
Applying this variable selection procedure to the example introduced at the beginning of this subsection, where we used $p=0.9,N_{rep} = 5$ and $\epsilon_L=0.1$, the frequencies of selecting $Z$ or $Z'$ in 100 replicates when $Z'$ has different correlation $r$ with $B$ are presented in Table \ref{tab:select_prop}. We can see that the algorithm almost perfectly selects the true covariate $Z$ and screens out the false covariate $Z'$. As a result, the performance of spectral clustering using selected covariates is very close to using only true covariates (oracle), and is much better than fitting both covariates, as is shown in Figure \ref{fig:cov}. Thus, we conclude that while fitting false correlated covariates could be very harmful to clustering performance under PCABM, Algorithm \ref{ecv_z_pcabm} could screen out these false covariates, which makes the whole clustering procedure more robust to confounding variables.

\begin{table}[H]
    \centering
    \caption{Number of times that  covariate $Z$ or $Z'$ is selected by Algorithm \ref{ecv_z_pcabm} in 100 replicates.}
    {\small
    \begin{tabular}{c|cccccccccc}
        \hline
        $r$&0&0.1&0.2&0.3&0.4&0.5&0.6&0.7&0.8&0.9\\
        \hline
        $Z$&100&    100&    100&   100&    100&    100&    100 & 100&  100&100\\
        $Z'$&0 &   0&    0&    0&    0 &   0&    0&    0&    1&    0\\
        \hline
    \end{tabular}
    }
    \label{tab:select_prop}
\end{table} 

\begin{algorithm}[htbp]
        \SetKwInOut{Input}{input}
        \SetKwInOut{Output}{output}
  \caption{ECV for stepwise covariate selection in PCABM}
  \label{ecv_z_pcabm}
  \Input{Adjacency matrix $A$, covariates $\bz_{ij} \in \mathbb R^{d}, i,j = 1,...,n $,   number of communities $K$, training proportion $p$, number of replications / folds $N_{rep}$, constant $\epsilon_L$.}
  \Output{Selected covariate index set $S$.}
Initialize selected covariate index set $S = \emptyset $, and the best loss $L_{old} =$ some large number.\\
\While{$L_{old}$ has not converged}{
  \For{$m=1$ \KwTo $N_{rep}$}
  {
  Randomly choose a subset of node pairs $\Omega$: selecting each pair $(i, j), i < j$ independently
  with probability $p$, and adding $(j, i)$ if $(i, j)$ is selected.\\
  \For{$d_1$ in $[d]\setminus S $ }{
  Let $S' = S \cup \{d_1\}$. Consider the model with covariates $Z_{S'}$. \\
   Calculate MLE $\hat\pg_{S'}^{(\Omega)} $ with $P_{\Omega} A, P_{\Omega} Z$: Optimize
  $l_{\be}^{(\Omega)}(\pg) = \sum_{(i,j)\in \Omega} A_{ij} \bz^\top_{ij} \pg -\frac12 \sum_{kl} O^{(\Omega)}_{kl}(\be) \log E_{kl}^{(\Omega)}(\be,\pg)$ with $\bz_{ij}$ restricted on $S'$.
  \\
  Apply matrix completion to $P_{\Omega} A'$ with rank constraint $K$ to obtain $\hat A'_{d_1}$, where $A'$ denotes the adjacency matrix adjusted by $\hat\pg$.\\
  Run spectral clustering on $\hat A'_{d_1}$ to obtain the estimated membership vector $\hat\be^{(m)}_{d_1}$.
  Estimate the probability matrix $\hat B^{(m)}_{d_1}$: $ \hat B_{kl}(\hat\be,\hat\pg) = O_{kl}^{(\Omega)}(\hat\be)/E_{kl}^{(\Omega)}(\hat\be,\hat\pg) $.\\
Evaluate the corresponding losses $L^{(m)}_{d_1}$, by applying the loss function L with the 
 estimated parameters to $A_{ij} , (i, j) \in \Omega^c $.
    }}
  Let $L_{d_1} = \sum_{m=1}^{N_{rep}} L_{d_1}^{(m)} / N_{rep} $ for all $d_1 \in [d]\setminus S$.
   If $L_{old} - \min_{d_1 \in [d]\setminus S} L_{d_1} >  \epsilon_L |L_{old}|$, let 
   $d^* = \arg\min_{d_1 \in [d]\setminus S} L_{d_1}$, $S = S \cup \{d^*\} $, and set $L_{old} = L_{d^*}$. Otherwise, claim  $L_{old}$ as converged and stop the algorithm.}
  
\end{algorithm}

\section{Additional Simulation and Real Data Results}
\label{sec:supple:addi_simu_realdata}
\subsection{Estimating $\bm{\gamma}$ with Random Initial $\be$}
\label{subsec:supple:addi_randominit_e}
Here, we present in Table \ref{tab:gamma} the simulation results on the estimation of $\boldsymbol{\gamma}$ for Section \ref{subsec::simu-gamma} under random initial community assignments. It is very similar to Table \ref{tab:gamma_K=1} when we ignore the community structure.

\begin{table}[H]
    \centering
    \caption{Simulated results over 100 replicates of $\hat{\boldsymbol{\gamma}}$, displayed as mean (standard deviation).}
    {\small
    \begin{tabular}{c|ccccc}
        \hline
        $n$&$\boldsymbol{\gamma}^0_1=0.4$&$\boldsymbol{\gamma}^0_2=0.8$&$\boldsymbol{\gamma}^0_3=1.2$&$\boldsymbol{\gamma}^0_4=1.6$&$\boldsymbol{\gamma}^0_5=2$\\ \hline
        {100}&0.399 (0.0414) &0.797 (0.0354) &1.197 (0.0455) &1.596 (0.0467) &1.995 (0.0484)\\
         \hline
        {300}&0.399 (0.0205)&0.801 (0.0151)&1.199 (0.0227)&1.603 (0.0217)&2.000 (0.0234)\\
        \hline
        {500}&0.395 (0.0131)&0.799 (0.0118)&1.197 (0.0173)&1.599 (0.0140)&2.002 (0.0147)\\
        \hline
    \end{tabular} }
    \label{tab:gamma}
\end{table}

\subsection{PCABM Clustering Visualization in the School Friendship Data}
\label{subsec:supple:schoolfriend_visualize}
In Figure~\ref{fig:friend}, school and ethnicity are targeted communities, respectively. We use different shades to distinguish true communities. Predicted communities are separated by the middle dash line so that the ideal split would be shades vs.\ tints on two sides. By these criteria, our model performs pretty well in both cases.
\begin{figure}[p]
    \centering
    \begin{subfigure}{.9\textwidth}
        \centering
        \includegraphics[width=\linewidth]{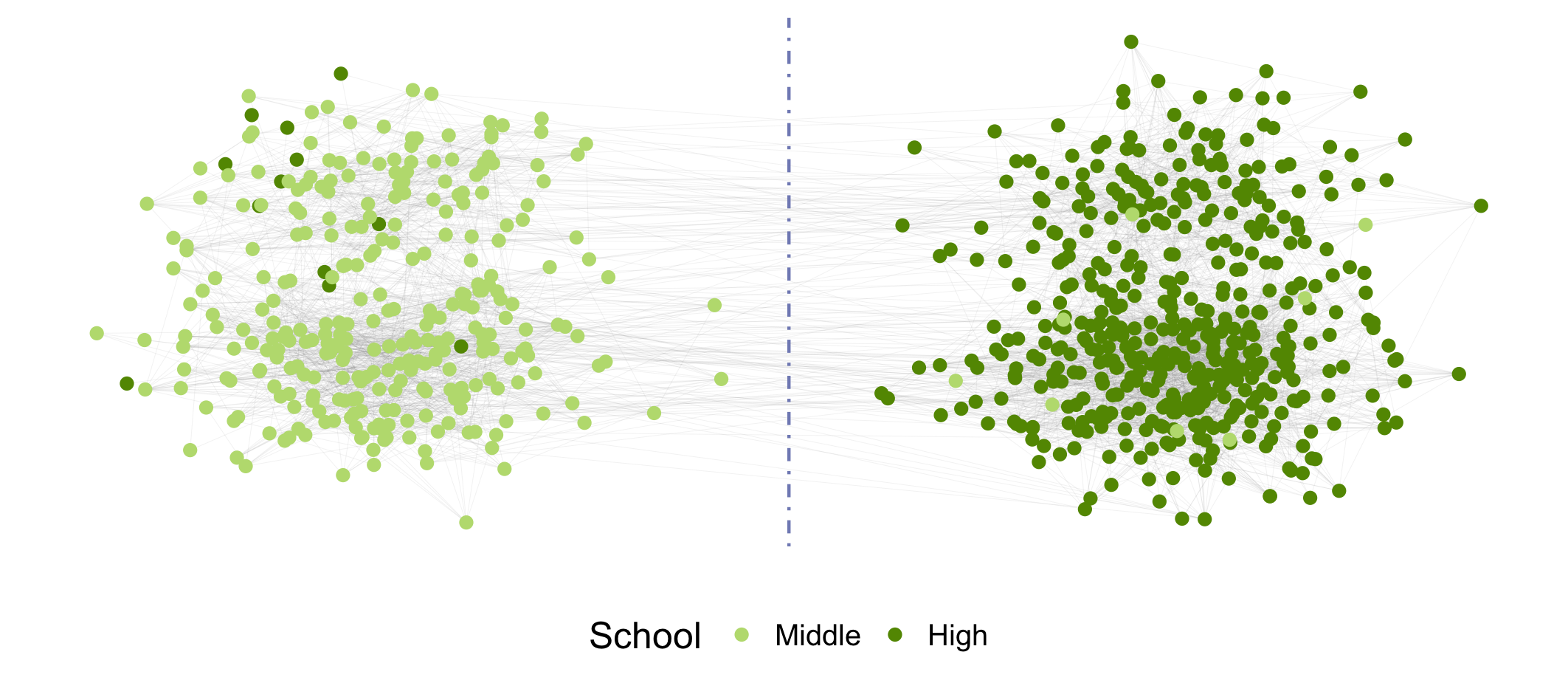}
        \label{fig:sfig1}
    \end{subfigure}%
    
    \begin{subfigure}{.9\textwidth}
        \centering
        \includegraphics[width=\linewidth]{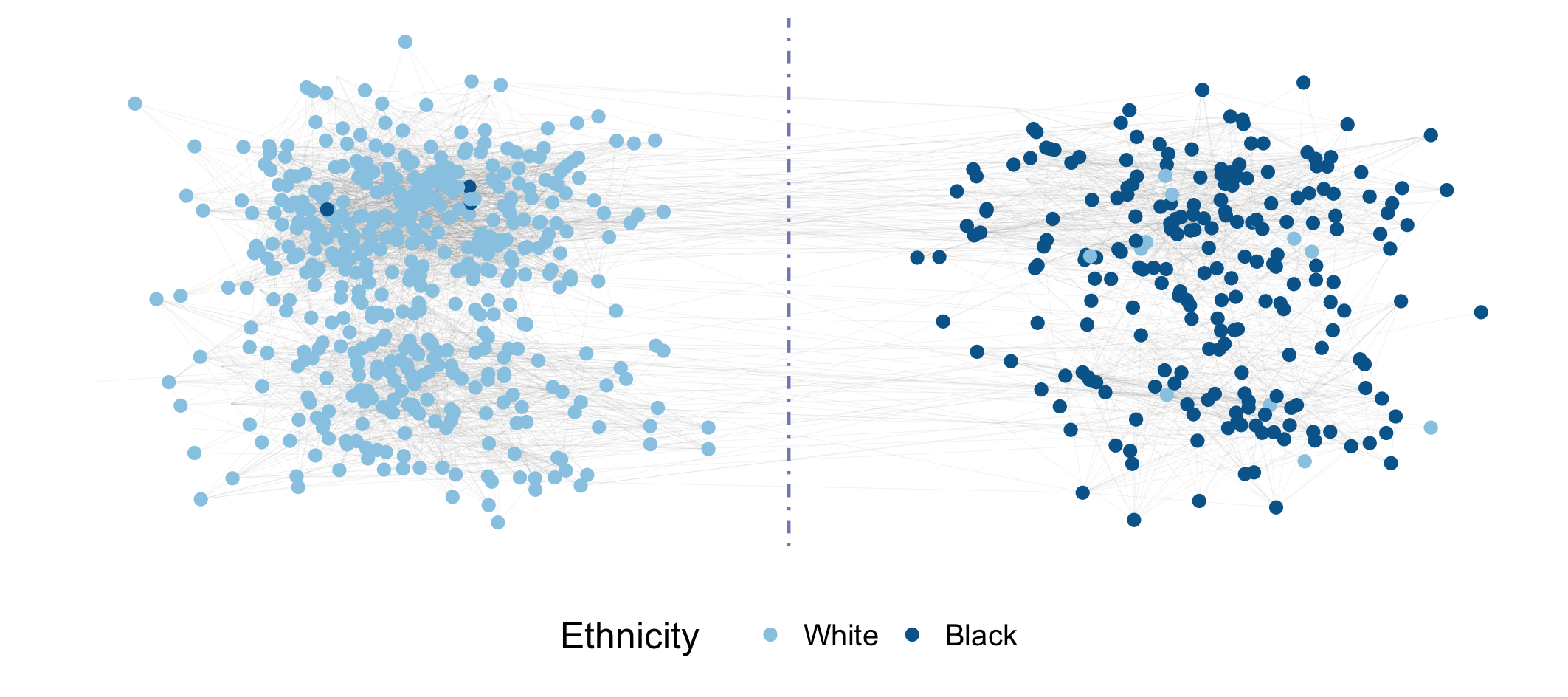}
        \label{fig:sfig2}
    \end{subfigure}
    
    \caption{Community detection with different pairwise covariates. From top to bottom, we present community prediction results for school and ethnicity.}
    \label{fig:friend}
\end{figure}

\end{document}